\documentclass[12pt]{article}
\usepackage[pdftex]{graphicx}
\usepackage[T2A]{fontenc}
\usepackage{amsthm,amsmath,amssymb}

\allowdisplaybreaks

\usepackage[usenames,dvipsnames]{color}
\usepackage[pdftex,colorlinks=true,unicode,bookmarksnumbered=true]{hyperref}

\usepackage[dvipsnames]{xcolor}

\hypersetup{
  colorlinks,
  citecolor=blue!80!black,
  linkcolor=red!80!black,
  urlcolor=blue!80!black}

\usepackage[onehalfspacing]{setspace}
\usepackage{textcomp}
\usepackage{indentfirst}
\usepackage{natbib}
\usepackage{listings}
\usepackage{pdflscape}
\usepackage{enumerate}

\usepackage{titling} 
\usepackage{xpatch}
\makeatletter
\newenvironment{mytitlepage}%
  {\begin{titlepage}\def\@thanks{}}%
  {\end{titlepage}}
\xpatchcmd\titlepage{\setcounter{page}\@ne}{}{}{}
\xpatchcmd\endtitlep

\usepackage{verbatim}
\usepackage{bbm}
\usepackage{tikz}
\usepackage{pgfplots}
\usetikzlibrary{arrows,shapes,trees}
\usepackage{float}
 \usepackage{multirow}
 
 \usepackage{authblk}

 \usepackage[toc,page]{appendix}

\usepackage{array,longtable}

\usepackage{mathtools}

\DeclarePairedDelimiter\floor{\lfloor}{\rfloor}

\DeclareMathOperator*{\Card}{Card}

\numberwithin{equation}{section}
\newtheorem{theorem}{Theorem}

\newtheorem{assumption}{Assumption}

\newtheorem{lemma}{Lemma}
\newtheorem{corollary}{Corollary}

\theoremstyle{definition}
\newtheorem{remark}{Remark}
\newtheorem{example}{Example}

\renewenvironment{abstract}
 {\small
  \begin{center}
  \bfseries \abstractname\vspace{-.5em}\vspace{0pt}
  \end{center}
  \list{}{%
    \setlength{\leftmargin}{5mm}
    \setlength{\rightmargin}{\leftmargin}%
    \listparindent 1.5em
    \itemindent    \listparindent
  }%
  \item\relax}
 {\endlist}

\newcommand{\lsim}{\mathrel{\hbox{\rlap{\lower.75ex \hbox{$\sim$}} \kern-.3em \raise.4ex \hbox{$<$}}}}

\usepackage{natbib}
\usepackage{geometry}
\begin{document}
\begin{mytitlepage}

\title{A Consistent Heteroskedasticity Robust LM Type \\ Specification Test for Semiparametric Models\thanks{I am grateful to my advisors Frank Wolak, Han Hong, and Peter Reiss for their support and guidance, as well as to Svetlana Bryzgalova, Guido Imbens, Brad Larsen, Joe Romano, and Paulo Somaini for very helpful conversations. I thank seminar participants at several universities and conferences for helpful comments. I thank Adonis Yatchew for permission to use the Canadian household gasoline consumption dataset from \citet{yatchew_no_2001}. I gratefully acknowledge the financial support from the Stanford Graduate Fellowship Fund as a Koret Fellow and from the Stanford Institute for Economic Policy Research as a B.F. Haley and E.S. Shaw Fellow. All remaining errors are mine.}}
\author{Ivan Korolev\thanks{Department of Economics, Binghamton University. E-mail: \mbox{\href{mailto:ikorolev@binghamton.edu}{ikorolev@binghamton.edu}}. Website: \url{https://sites.google.com/view/ivan-korolev/home}}}
\date{November 7, 2019}

\maketitle

\begin{abstract}
This paper develops a consistent heteroskedasticity robust Lagrange Multiplier (LM) type specification test for semiparametric conditional mean models. Consistency is achieved by turning a conditional moment restriction into a growing number of unconditional moment restrictions using series methods. The proposed test statistic is straightforward to compute and is asymptotically standard normal under the null. Compared with the earlier literature on series-based specification tests in parametric models, I rely on the projection property of series estimators and derive a different normalization of the test statistic. Compared with the recent test in \citet{gupta_2018}, I use a different way of accounting for heteroskedasticity. I demonstrate using Monte Carlo studies that my test has superior finite sample performance compared with the existing tests. I apply the test to one of the semiparametric gasoline demand specifications from \citet{yatchew_no_2001} and find no evidence against it.
\end{abstract}



\end{mytitlepage}

\newpage

\doublespacing

\section{Introduction}\label{introduction}

Applied economists often want to achieve two conflicting goals in their work. On the one hand, they wish to use the most flexible specification possible, so that their results are not driven by functional form assumptions. On the other hand, they wish to have a model that is consistent with the restrictions imposed by economic theory and can be used for valid counterfactual analysis.

Because economic theory usually specifies some portions of the model but leaves the others unrestricted, semiparametric models are especially attractive for empirical work in economics. They are more flexible than fully parametric models but more parsimonious than fully nonparametric models. However, since semiparametric models are restricted versions of fully nonparametric models, it is important to check the validity of implied restrictions. If a semiparametric model is correctly specified, then using it, as opposed to a nonparametric model, typically leads to more efficient estimates and may increase the range of counterfactual questions that can be answered using the model at hand. If a semiparametric model is misspecified, then the semiparametric estimates are likely to be misleading and may result in incorrect policy implications.

In this paper I develop a new heteroskedasticity robust Lagrange Multiplier (LM) type specification test for semiparametric models. The test determines whether a semiparametric conditional mean model that the researcher has estimated provides a statistically valid description of the data as compared to a general nonparametric model. It uses series methods to turn a conditional moment restriction into a growing number of unconditional moment restrictions. I show that if the series functions can approximate the nonparametric alternatives that are allowed as the sample size grows, the test is consistent. My assumptions and proofs make precise what is required of the approximation and its behavior as the number of series terms and the sample size grow. These arguments differ from standard parametric arguments, when the number of regressors in the model is fixed.

My test is closely related to the tests in \citet{gupta_2018}, written independently and simultaneously with earlier versions of this paper. \citet{gupta_2018} also develops series-based specification tests, including the LM type test, for semiparametric models. While there are many similarities between his paper and mine, there are important differences. Unlike this paper, \citet{gupta_2018} focuses on the homoskedastic case in his theoretical analysis. He proposes an extension to the heteroskedastic case in his paper but does not provide any theoretical results for that case. Moreover, the heteroskedasticity robust test in \citet{gupta_2018} is based on the FGLS residuals and uses the estimate of the FGLS weighting matrix that is generally inconsistent. As I show in Section~\ref{simulations} using Monte Carlo studies, his test is undersized and has virtually no power in finite samples when the errors are heteroskedastic.

In contrast with \citet{gupta_2018}, my test is based on the OLS residuals combined with the robust estimate of the variance-covariance matrix in the spirit of \citet{white_1980}. My test also closely resembles the parametric heteroskedasticity robust LM test in \citet{wooldridge_1987}. As a result, my test is valid even when a consistent estimate of the variance-covariance matrix of the errors is not available. I do not only prove the asymptotic validity of my test, but also demonstrate using simulations that it outperforms the test in \citet{gupta_2018} in finite samples. In particular, my test controls size well and has good power.

As compared with the earlier literature on consistent series-based specification tests, my test uses a different normalization of the test statistic. In order to achieve consistency, I let the number of restrictions grow with the sample size. Thus, I cannot obtain a convergence in distribution result for the standard LM type test statistic, which diverges to infinity, and need to normalize it. I normalize the LM test statistic by the number of restrictions the semiparametric model imposes on the nonparametric model. This normalization may seem natural, because the usual parametric LM test is asymptotically $\chi^2(r)$ under the null, and it is well known that $(\chi^2(r) - r)/\sqrt{2 r} \overset{d}{\to} N(0,1)$ as $r \to \infty$. However, the existing literature on consistent series-based specification tests in parametric models, e.g.  \citet{de_jong_bierens_1994}, \citet{hong_white_1995}, and \citet{donald_et_al_2003}, mostly uses the normalization given by the total number of parameters.

Intuitively, the difference arises because I nest the null model in the alternative and rely on the projection property of series estimators. As a result, the restricted residuals are orthogonal to the series terms included in the restricted model. This means that even under the alternative, only a subset of moment conditions, rather than all of them, can be violated, which in turn affects the normalization of the test statistic. I refer to the normalization of the test statistic by the number of restrictions, rather than the total number of parameters, as the degrees of freedom correction. While this distinction may appear subtle, it has important theoretical and practical consequences.

From the theoretical point of view, the degrees of freedom correction allows me to develop the asymptotic theory under mild rate conditions. In my asymptotic analysis, I decompose the test statistic into the leading term and the remainder. Without the projection property, the remainder term would contain both bias and variance of the semiparametric estimates. Thus, in order to bound the remainder term, I would need to require both bias and variance to vanish sufficiently fast, and the resulting rate conditions would be very restrictive. In contrast, by relying on the projection nature of series estimators, I can directly account for the estimation variance, so that only bias enters the remainder term. Because of this, I only need to control the rate at which bias goes to zero to bound the remainder term. As a result, I can derive the asymptotic distribution of the test statistic under fairly weak rate conditions (see Section~\ref{asymptotics_summary} for more details).

From the practical point of view, the degrees of freedom correction substantially improves the finite sample performance of the test. It has long been known in the literature on consistent specification tests that asymptotic approximations often do not work well in finite samples (see, e.g., \citet{li_wang_1998}). The bootstrap is typically used to improve the finite sample performance of consistent specification tests, but it may be computationally costly. While I propose a wild bootstrap procedure that can be used with my test and establish its asymptotic validity, my simulations suggest that the finite sample performance of the asymptotic version of the proposed test with the degrees of freedom correction is almost as good as, and often as good as, that of its bootstrap version. In other words, the degrees of freedom correction serves as a computationally attractive analytical way to obtain a test with very good small sample behavior.

\citet{guay_guerre_2006} also develop a series-based specification test that is based on a quadratic form in the restricted model residuals. They allow for a data-driven tuning parameters choice and obtain an adaptive rate-optimal test, but their test only deals with parametric null hypotheses and does not involve the degrees of freedom correction that is crucial in my analysis. I use a modified version of their approach to tuning parameter choice in simulations in the supplement and find that it performs well. It would be interesting to establish the theoretical validity of their approach when the null hypothesis is semiparametric in future work.

To my knowledge, there are two studies in addition to \citet{gupta_2018} that develop series-based specification tests that allow for semiparametric null hypotheses. \citet{gao_et_al_2002} only consider additive models and do not explicitly develop a consistent test against a general class of alternatives. Their test is based on the estimates from the unrestricted model and can be viewed as a Wald type test for variables significance in nonparametric additive conditional mean models. In contrast, my test is based on the residuals from the restricted semiparametric model and is consistent against a broad class of alternatives for the conditional mean function.

\citet{li_et_al_2003} use the approach that was first put forth in \citet{bierens_1982} and \citet{bierens_1990} and develop a series-based specification test based on weighting the moments, rather than considering an increasing number of series-based unconditional moments. Their test can detect local alternatives that approach zero at the parametric rate $n^{-1/2}$, but the asymptotic distribution of their test statistic depends on nuisance parameters, and it is difficult to obtain appropriate critical values. They propose using a residual-based wild bootstrap to approximate the critical values. In contrast, my test statistic is asymptotically standard normal under the null, so calculating critical values is straightforward.

Another strand of literature related to my paper studies inference in conditional moment restriction models, such as \citet{chen_pouzo_2015} and \citet{chernozhukov_et_al_2015}. While those papers are more general than mine and can handle a wider class of models, they do not explicitly propose using their methods for the purpose of specification testing. They also impose high level assumptions and do not derive primitive conditions under which their tests can serve as consistent specification tests. In contrast, I impose assumptions that are more primitive and familiar from the literature on series estimation and explicitly derive the rate conditions under which the proposed test is valid and consistent. My proof techniques are also different from those in \citet{chen_pouzo_2015} and \citet{chernozhukov_et_al_2015}.

Similar to the overidentifying restrictions test in GMM models or other omnibus specification tests, the proposed test is silent on how to proceed if the null is rejected. In this respect, it is clearly a test of a particular model specification but not a comprehensive model selection procedure. I rely on the results from this paper and study model selection methods for semiparametric and nonparametric models, such as series-based information criteria or upward/downward testing procedures based on the proposed test, in \citet{korolev_2018}.

The remainder of the paper is organized as follows. Section~\ref{example} presents several motivating examples. Section~\ref{model} introduces the model and describes how to construct the series-based specification test for semiparametric models. Section~\ref{asymptotics_series} develops the asymptotic theory for the proposed test. Section~\ref{simulations} studies the finite sample performance of the test. Section~\ref{empirical_example} applies the proposed test to one of the semiparametric household gasoline demand specifications from \citet{yatchew_no_2001} and shows that it is not rejected by the data. Section~\ref{conclusion} concludes.

\ref{appendix_tables_figures} collects all figures. \ref{appendix_proofs} contains auxiliary lemmas and proofs of my results. Additional details and results can be found in the supplement.

\section{Motivating Examples}\label{example}

This section presents several examples of applied papers in economics that use semiparametric models in their analysis. \citet{hausman_newey_1995} estimate gasoline demand as a flexible function of prices and incomes. They also include in the model other covariates and treat them parametrically to reduce the dimension of the nonparametric part of the model:
\[
y = g(PRICE,INCOME) + z' \beta + \varepsilon,
\]
where $y$ is the gasoline demand (in logs), $PRICE$ is the logarithm of the gasoline price, $INCOME$ is the logarithm of the household income, and $z$ includes other variables. For similar reasons, \citet{schmalensee_stoker_1999} use the following semiparametric gasoline demand specification:
\[
y = g(INCOME,AGE) + z' \beta + \varepsilon,
\]
where $AGE$ is the logarithm of the age. In turn, \citet{yatchew_no_2001} model gasoline demand as:
\[
y = \alpha_1 PRICE + \alpha_2 INCOME + g(AGE) + z' \beta + \varepsilon,
\]

Semiparametric models have also been used in estimation of Engel curves (\citet{gong_et_al_2005}), the labor force participation equation (\citet{martins_2001}), the relationship between land access and poverty (\citet{finan_et_al_2005}), and marginal returns to education (\citet{carneiro_et_al_2011}). My test can be used to check whether these semiparametric specifications are statistically adequate.

\section{The Model and Specification Test}\label{model}

This section describes the model, discusses semiparametric series estimators, and introduces the test statistic.

\subsection{The Model and Null Hypothesis}

Let $(Y_i, X_i')' \in \mathbb{R}^{1+d_x}$, $d_x \in \mathbb{N}$, $i = 1, ..., n$, be independent and identically distributed random variables with $E[Y_i^2] < \infty$. Then there exists a measurable function $g$ such that $g(X_i) = E[Y_i | X_i]$ a.s., and the nonparametric model can be written as
\[
Y_i = g(X_i) + \varepsilon_i, \quad E[\varepsilon_i | X_i] = 0
\]

The goal of this paper is to test the null hypothesis that the conditional mean function is semiparametric. A generic semiparametric null hypothesis is given by
\begin{align}\label{generic_h0}
H_0^{SP}: P_{X} \left(g(X_i) = f(X_i, \theta_0, h_0) \right) = 1 \text{ for some } \theta_0 \in \Theta, h_0 \in \mathcal{H},
\end{align}
where $f: \mathcal{X} \times \Theta \times \mathcal{H} \to \mathbb{R}$ is a known function, $\theta \in \Theta \subset \mathbb{R}^{d}$ is a finite-dimensional parameter, and $h \in \mathcal{H} = \mathcal{H}_1 \times ... \times \mathcal{H}_q$ is a vector of unknown functions.

The fixed alternative is
\begin{align}\label{generic_h1}
H_1: P_{X} \left(g(X_i) \neq f(X_i, \theta, h) \right) > 0 \text{ for all } \theta \in \Theta, h \in \mathcal{H}
\end{align}

When the semiparametric null hypothesis is true, the model becomes
\[
Y_i = f(X_i, \theta_0, h_0) + \varepsilon_i, \quad E[\varepsilon_i | X_i] = 0
\]

In order to test whether the semiparametric model is correctly specified, I turn the conditional moment restriction $E[\varepsilon_i | X_i] = 0$ into a sequence of unconditional moment restrictions using series methods. But first, I introduce series approximating functions and semiparametric series estimators.

For any variable $v$, let $Q^{a_n}(v) = (q_1(v), ..., q_{a_n}(v))'$ be an $a_n$-dimensional vector of approximating functions of $v$, where the number of series terms $a_n$ is allowed to grow with the sample size $n$. Possible choices of series functions include:
\begin{enumerate}[(a)]

\item

Power series. For univariate $v$, they are given by:
\begin{align}\label{series}
Q^{a_n}(v) = (1, v, ..., v^{a_n-1})'
\end{align}

\item

Splines. Let $s$ be a positive scalar giving the order of the spline, and let $t_{1}, ..., t_{a_n - s - 1}$ denote knots.  Then for univariate $v$, splines are given by:
\begin{align}\label{splines}
Q^{a_n}(v) = (1, v, ..., v^s, \mathbbm{1}\{ v > t_1 \} (v-t_1)^s, ..., \mathbbm{1}\{ v > t_{a_n - s -1} \} (v - t_{a_n - s -1})^s)
\end{align}

\end{enumerate}

Multivariate power series or splines can be formed from products of univariate ones.\footnote{For a more detailed discussion of series methods and for other possible choices of basis functions, see Section 5 in \citet{newey_1997}, Section 2 in \citet{donald_et_al_2003}, or Section 2.3 in \citet{chen_2007}.}

\subsection{Series Estimators}

In this section I introduce additional notation and define semiparametric estimators. First, I write the semiparametric null model in a series form:
\begin{align}\label{semiparametric_series_form}
Y_i = f(X_i,\theta,h) + \varepsilon_i = W_i' \beta_1 + R_i + \varepsilon_i  = W_i' \beta_1 + e_i,
\end{align}
where $W_i := W^{m_n}(X_i) := (W_{1}(X_i), ..., W_{m_n}(X_i))'$ are appropriate regressors or basis functions, $m_n$ is the number of parameters in the semiparametric null model, $R_i := R_i^{m_n} := \left( f(X_i,\theta,h) - W_i' \beta_1 \right)$ is the approximation error, and $e_i = \varepsilon_i + R_i$ is the composite error term.

\begin{example}

Suppose that the semiparametric null model is partially linear with $f(X_i,\theta,h) = X_{1i} \theta + h(X_{2i})$, where $X_{1i}$ and $X_{2i}$ are scalars and $X_i = (X_{1i},X_{2i})'$. Approximate $h(x_2) \approx \sum_{j=1}^{a_n}{\gamma_j q_j(x_2)} = Q^{a_n}(x_2)' \gamma$ and rewrite the model as:
\[
Y_i = X_{1i} \theta + h(X_{2i}) + \varepsilon = X_{1i} \theta + Q^{a_n}(X_{2i})' \gamma + R_i + \varepsilon_i = W_i' \beta_1 + e_i,
\]
where $W_i = (X_{1i}, Q^{a_n}(X_{2i})')'$, $m_n = a_n + 1$, $\beta_1 = (\theta, \gamma')'$, $R_i = h(X_{2i}) - Q^{a_n}(X_{2i})' \gamma$, and $e_i = R_i + \varepsilon_i$. If power series are used, then $W_i = (X_{1i}, 1, X_{2i},X_{2i}^2,...,X_{2i}^{a_n-1})'$.

\end{example}

\begin{remark}
I assume that the semiparametric model $f(X_i,\theta,h)$ can be well approximated by a linear in parameters representation $W_i' \beta_1$. This is the case in many popular classes of semiparametric  models, such as partially linear, additive, or varying coefficient. However, if the null model is nonlinear parametric, i.e. $f(X_i,\theta)$ for a known function $f(\cdot)$, then one should estimate it using parametric methods and use the tests from \citet{hong_white_1995} or \citet{donald_et_al_2003}.
\end{remark}

In order to test the semiparametric model, I test whether additional series terms should enter the model. These are not the extra series terms used to approximate the semiparametric null model better, but rather the series terms that cannot enter the model under the null and are supposed to capture possible deviations from it. Thus, the alternative is given by:
\begin{align}\label{nonparametric_series_form}
Y_i = W_i' \beta_1 + Z_i' \beta_2 + R_i + \varepsilon_i = P_i' \beta + e_i,
\end{align}
where $Z_i := Z^{r_n}(X_i) := (Z_{1}(X_i), ..., Z_{r_n}(X_i))'$ are the series terms that are present only under the alternative, $r_n$ is the number of restrictions, $P_i := P^{k_n}(X_i) := (W_i', Z_i')'$, $k_n = m_n + r_n$ is the total number of parameters, and $\beta = (\beta_1', \beta_2')'$.

\addtocounter{example}{-1}
\begin{example}[continued]

In the partially linear model example, a possible choice of $Z_i$ is $Z_i = (X_{1i}^2, ..., X_{1i}^{a_n-1}, X_{1i} X_{2i}, ..., X_{1i}^{a_n-1} X_{2i}^{a_n-1})'$. These series terms can enter the model if the null hypothesis is false, but they cannot enter the model if it is true.

\end{example}

If the moment condition $E[Y_i - f(X_i,\theta_0,h_0) | X_i ] = 0$ holds, then the errors $\varepsilon_i = Y_i - f(X_i,\theta_0,h_0)$ are uncorrelated with any function of $X_i$, not only with $W_i$, and the additional series terms should be insignificant. Hence, the null hypothesis corresponds to $\beta_2 = 0$. The estimate of $\beta_1$ under the null becomes $\tilde{\beta}_1 = (W' W)^{-1} W' Y$.\footnote{For any vector $V_i$, let $V = (V_1, ..., V_n)'$ be the matrix that stacks all observations together.} The restricted residuals are given by
\[
\tilde{\varepsilon} = Y - W \tilde{\beta}_1 = Y - W (W' W)^{-1} W' Y = M_W Y = M_W (\varepsilon + R),
\]
where $M_W = I - P_W$, $P_W = W (W' W)^{-1} W'$. The residuals satisfy the condition $W' \tilde{\varepsilon} = 0$.

$m_n$ is the number of terms in the semiparametric null model. It has to grow with the sample size in order to approximate nonparametric components of the model sufficiently well.\footnote{The notion of sufficiently well is made precise later.} $r_n$ is the number of series terms that capture possible deviations from the null. $k_n = m_n + r_n$ is the total number of series terms in the unrestricted model. This number has to grow with the sample size if it is to approximate any nonparametric alternative. Typically the number of terms in the restricted semiparametric model, $m_n$, is significantly smaller than the number of terms in the unrestricted nonparametric model, $k_n$, so that $m_n/k_n \to 0$ and $r_n/k_n \to 1$. However, my rate conditions can also accommodate the case when $m_n/k_n \to C < 1$.

\subsection{Test Statistics}

Instead of the conditional moment restriction implied by the null hypothesis: $E[\varepsilon_i | X_i] = E[Y_i - f(X_i, \theta_0, h_0) | X_i] = 0$, the test will be based on the unconditional moment restriction $E[P_i \varepsilon_i] = 0$. I will show that requiring $k_n$ to grow with the sample size and $P(x) := P^{k_n}(x)$ to approximate any unknown function in a wide class of functions will allow me to obtain a consistent specification test based on the growing number of unconditional moment restrictions. The sample analog of the population moment condition is given by $P' \tilde{\varepsilon}/n$. However, because the semiparametric series estimators satisfy the condition $W' \tilde{\varepsilon}/n = 0$, the LM type test statistic is given by:
\begin{align}\label{xi_hc}
\xi_{HC} = \tilde{\varepsilon}' \tilde{Z} (\tilde{Z}' \tilde{\Sigma} \tilde{Z})^{-1} \tilde{Z}' \tilde{\varepsilon},
\end{align}
where $\tilde{Z} = M_W Z$ are the residuals from the regression of each element of $Z_i$ on $W_i$, $\tilde{\varepsilon} = (\tilde{\varepsilon}_1, ..., \tilde{\varepsilon}_n)'$, $\tilde{\varepsilon}_i = Y_i - f(X_i, \tilde{\theta}, \tilde{h}) = Y - W \tilde{\beta}_1$ are the semiparametric residuals, and $\tilde{\Sigma} = diag(\tilde{\varepsilon}_1^2,...,\tilde{\varepsilon}_n^2)$.

The quadratic form $\xi_{HC}$ can be computed in a regression-based way, which is discussed in detail in Section~\ref{implementation} in the supplement.\footnote{For regression-based ways to compute the parametric LM test statistic, see \citet{wooldridge_1987} or Chapters 7 and 8 in \citet{cameron_trivedi_2005}.}

Because the dimensionality $r_n$ of $Z_i$ grows with the sample size, a normalization is needed to obtain a convergence in distribution result. I show that the test statistic
\begin{align}\label{t_test_statistic}
t_{HC,\tau_n} = \frac{\xi_{HC} - \tau_n}{\sqrt{2 \tau_n}}
\end{align}
for a suitable sequence $\tau_n \to \infty$ as $n \to \infty$ works and is asymptotically pivotal. The test is one-sided in the sense that the null hypothesis is rejected when the value of the test statistic is large and positive. I will discuss the choice of $\tau_n$ in the next section.

\section{Asymptotic Theory}\label{asymptotics_series}

\subsection{Behavior of the Test Statistic under the Null}\label{null_series}

In this subsection I derive the asymptotic properties of the proposed test under the null. I start by introducing my assumptions. First, I impose basic assumptions on the data generating process.

\begin{assumption}\label{dgp}

$(Y_i, X_i')' \in \mathbb{R}^{1+d_x}, d_x \in \mathbb{N}, i = 1, ..., n$ are i.i.d. random draws of the random variables $(Y, X')'$, and the support of $X$, $\mathcal{X}$, is a compact subset of $\mathbb{R}^{d_x}$.

\end{assumption}

Next, I define the error term and require it to satisfy certain restrictions.

\begin{assumption}\label{errors_unified}

Let $\varepsilon_i = Y_i - E[Y_i | X_i]$. Then $E[\varepsilon_i^2 | X_i] = \sigma^2(X_i) \overset{def}{=} \sigma_i^2$, $P(0 < \inf_{i}{\sigma_i^2} \leq \sup_{i}{\sigma_i^2} < \infty) = 1$, $E[\varepsilon_i^3 | X_i] = \mu_{3,i}$, $E[\varepsilon_i^4 | X_i] = \mu_{4,i}$, and $\sup_{i}{E[|\varepsilon_i|^{4 + \delta}]} \leq C$ for some constant $C$ and $\delta > 0$.

\end{assumption}

The following assumption deals with the behavior of the approximating series functions. From now on, let $\| A \| = [tr(A'A)]^{1/2}$ be the Euclidian norm of a matrix $A$.

\begin{assumption}\label{series_norms_eigenvalues}

For each $m$, $r$, and $k$ there are nonsingular matrices $B_1$ and $B_2$ such that, for $\bar{W}^{m}(x) = B_1 W^{m}(x)$, $\bar{Z}^{r}(x) = B_2 Z^{r}(x)$, and $\bar{P}^{k}(x) = (\bar{W}^{m}(x)', \bar{Z}^{r}(x)')'$,

\begin{enumerate}[(a)]

\item

There exists a sequences of constants $\zeta(\cdot)$ that satisfies the conditions $\sup_{x \in \mathcal{X}} \| \bar{W}^{m}(x) \| \leq \zeta(m)$, $\sup_{x \in \mathcal{X}} \| \bar{Z}^{r}(x) \| \leq \zeta(r)$, and $\sup_{x \in \mathcal{X}} \| \bar{P}^{k}(x) \| \leq \zeta(k)$.

\item

The smallest eigenvalue of $E[\bar{P}^{k}(X_i) \bar{P}^{k}(X_i)']$ is bounded away from zero uniformly in $k$.

\end{enumerate}

\end{assumption}

\begin{remark}[$\zeta(\cdot)$ for Common Basis Functions]

Explicit expressions for $\zeta(\cdot)$ are available for certain families of basis approximating functions. For instance, it has been shown (see \citet{newey_1997} or Section 15.1.1 in \citet{li_racine_2007}) that under additional assumptions, $\zeta(a) = O(a^{1/2})$ for splines and $\zeta(a) = O(a)$ for power series.

\end{remark}

Next, I impose an assumption that requires the approximation error for the semiparametric model to vanish sufficiently fast under the null.

\begin{assumption}\label{series_approx}
Suppose that $H_0$ holds. There exist $\alpha > 0$ and $\beta_1 \in \mathbb{R}^{m_n}$ such that
\[
\sup_{x \in \mathcal{X}}{| f(x,\theta_0,h_0) - W^{m_n}(x)' \beta_1 |} = O(m_n^{-\alpha})
\]
\end{assumption}

$\beta_1$ in this assumption can be defined in various ways. One natural definition is projection: $\beta_1 = E[W^{m_n}(X_i) W^{m_n}(X_i)']^{-1} E[W^{m_n}(X_i) f(X_i,\theta_0,h_0)]$ (see, e.g., Chapter 21 in \citet{hansen_2019}).

\begin{remark}[$\alpha$ for Common Basis Functions]

In certain special cases, it is possible to characterize $\alpha$ explicitly. Suppose that power series or splines are used and that $f(x,\theta,h) = x_1' \theta + h(x_2)$, where $h$ has $s_h$ continuous derivatives and $x_2$ is $d_{x2}$-dimensional. Then, as shown in \citet{newey_1997} or Chapter 15 in \citet{li_racine_2007}, $\alpha = s_h/d_{x2}$.

\end{remark}

Finally, I impose certain rate conditions:

\begin{assumption}\label{rate_conditions}
$r_n \rightarrow \infty$ as $n \to \infty$ and the following rate conditions hold:
\begin{eqnarray}
\label{rate_cond_r_1_hc} (m_n/n + m_n^{-2\alpha}) \zeta(r_n)^2 r_n^{1/2} &\rightarrow 0 \\
\label{rate_cond_r_2_hc} \zeta(r_n) r_n / n^{1/2} &\rightarrow 0 \\
\label{rate_cond_r_3_hc} n m_n^{-2\alpha}/ r_n^{1/2} &\rightarrow 0,
\end{eqnarray}
\end{assumption}

It is now possible to derive the asymptotic distribution of the test statistic~\ref{t_test_statistic} under the null. It is given below:

\begin{theorem}\label{asy_distr_t_r_n_hc}
Under $H_0$ and Assumptions \ref{dgp}--\ref{rate_conditions},
\begin{eqnarray}\label{eqn_t_r_n_HC}
t_{HC,r_n} = \frac{\xi_{HC} - r_n}{\sqrt{2 r_n}} \stackrel{d}{\rightarrow} N(0,1),
\end{eqnarray}
where $\xi_{HC}$ is as in Equation~(\ref{xi_hc}).
\end{theorem}

The following corollary shows that using a $\chi^2$ approximation with a growing number of degrees of freedom also results in an asymptotically exact test.

\begin{corollary}\label{corollary_t_r_n_hc}
Under assumptions of Theorem~\ref{asy_distr_t_r_n_hc},
\[
P ( \xi_{HC} \geq \chi_{1-\alpha}^2(r_n) )  = P \bigg{(} \frac{\xi_{HC} - r_n}{\sqrt{2 r_n}} \geq \frac{\chi_{1-\alpha}^2(r_n) - r_n}{\sqrt{2 r_n}} \bigg{)} \rightarrow \alpha,
\] 
where $\chi_{1-\alpha}^2(r_n)$ is the $(1-\alpha)$-quantile of the $\chi^2$ distribution with $r_n$ degrees of freedom.
\end{corollary}

\subsection{Behavior of the Test Statistic under a Fixed Alternative}\label{global_alternative_series}

In this subsection, I study the behavior of the test statistic under a fixed alternative. The true model is now nonparametric:
\[
Y_i = g(X_i) + \varepsilon_i = W_i' \beta_1 + Z_i' \beta_2 + R_i + \varepsilon_i, \quad E[\varepsilon_i | X_i] = 0,
\]
where $\beta_2' \beta_2 > 0$. The pseudo-true parameter value $\beta_1^1$ solves the moment condition $E[W_i (Y_i - W_i' \beta_1^{1})] = 0$, and the semiparametric estimator $\tilde{\beta}_1$ solves its sample analog $W' (Y - W \tilde{\beta}_1)/n = 0$. Let $\eta_i = Y_i - W_i' \beta_1^{1}$.

The following theorem provides the divergence rate of the test statistic under a fixed alternative:

\begin{theorem}\label{global_alternative_t_r}
Let $\Omega_{\eta,HC} = E[\eta_i^{2} \tilde{Z}_i \tilde{Z}_i']$, $\tilde{\Omega}_{HC} = \tilde{Z}' \tilde{\Sigma} \tilde{Z}/n$. Suppose that there exists $\beta = (\beta_1', \beta_2')'$ such that $\sup_{x \in \mathcal{X}}{| g(x) - P^{k_n}(x)' \beta |} \to 0$ as $k_n \to \infty$, $\| \tilde{\Omega}_{HC} - \Omega_{\eta,HC} \| \overset{p}{\to} 0$, the smallest eigenvalue of $\Omega_{\eta,HC}$ is bounded away from zero, $r_n \to \infty$, $r_n^{1/2}/n \to 0$. Then under $H_1$ and for any nonstochastic sequence $\{ C_n \}$, $C_n = o(n r_n^{-1/2})$, $P(t_{HC,r_n} > C_n) \to 1$.
\end{theorem}

The class of functions $g(x)$, for which the test is consistent, consists of functions that can be approximated (in the sup norm sense) using series as the number of series terms grows. While it is difficult to give a necessary and sufficient primitive condition that would describe this class of functions, the test will likely be consistent against continuous and smooth alternatives, while it may not be consistent against alternatives that exhibit jumps. 

\begin{remark}
An advantage of using series methods in specification testing is that the researcher can alter the class of alternatives agains which the test is consistent. For instance, if there are two regressors $X_{1i}$ and $X_{2i}$ and the researcher is interested in consistency against a general nonparametric alternative $g(X_{1i},X_{2i})$, then she should include interactions between the univariate series terms $Q^{a_n}(X_{1i})$ and $Q^{a_n}(X_{2i})$. If the researcher wants to aim power at additive alternatives $g_1(X_{1i}) + g_2(X_{2i})$, she could use $Q^{a_n}(X_{1i})$ and $Q^{a_n}(X_{2i})$ separately, without any interactions between them.
\end{remark}

The divergence rate of the test statistic under the alternative is $n/\sqrt{r_n}$. However, in most cases, the restricted semiparametric model is of lower dimension than the unrestricted nonparametric model, so that $m_n/k_n \to 0$ and $r_n/k_n \to 1$. Thus, the divergence rate in the semiparametric case discussed here is the same as in the parametric case in \citet{hong_white_1995} and \citet{donald_et_al_2003}, so the fact that the null hypothesis is semiparametric does not affect the global power of the test.

\begin{remark}

It is possible to show that under sequences of local alternatives that approach the null at the rate $r_n^{1/4}/n^{1/2}$, the test statistic $t_{HC,r_n}$ is asymptotically $N(2^{-1/2},1)$. I omit the details for brevity.

\end{remark}

\subsection{A Wild Bootstrap Procedure}\label{wild_bootstrap}

In this subsection I propose a wild bootstrap procedure that can be used with my test and establish its asymptotic validity. While my test performs well in simulations even when asymptotic critical values are used, the bootstrap may yield an even further improvement in the finite sample performance of the test. I will compare the performance of the asymptotic and bootstrap versions of the test in simulations.

Because I am interested in approximating the asymptotic distribution of the test under the null hypothesis, the bootstrap data generating process should satisfy the null. Moreover, because my test is robust to heteroskedasticity, the bootstrap data generating process should also be able to accommodate heteroskedastic errors. The wild bootstrap can satisfy both these requirements.

My bootstrap procedure is fairly similar to that in \citet{li_wang_1998}. I use $\varepsilon_i^*$ to denote the bootstrap errors based on the restricted residuals $\tilde{\varepsilon}_i$. The bootstrap errors should satisfy the following conditions:
\[
\text{(i) } E^*[\varepsilon_i^*] = 0, \quad \text{(ii) } E^*[\varepsilon_i^{*2}] = \tilde{\varepsilon}_i^2,
\]
where $E^*[\cdot] = E[\cdot | \mathcal{U}_n]$ is the expectation conditional on the data $\mathcal{U}_n = \{ (Y_i, X_i')' \}_{i=1}^{n}$. A common way to satisfy these requirements is to let $\varepsilon_i^* = V_i^* \tilde{\varepsilon}_i$, where $V_i^*$ is a two-point distribution. A popular choice is Mammen's two point distribution, originally introduced in \citet{mammen_1993}: 
\[ 
V_i^* =
  \begin{cases}
    (1-\sqrt{5})/2       & \quad \text{with probability } (\sqrt{5}+1)/(2 \sqrt{5}),\\
    (1+\sqrt{5})/2      & \quad \text{with probability } (\sqrt{5}-1)/(2 \sqrt{5}).
  \end{cases}
\]

Another possible choice is the Rademacher distribution, as suggested in \citet{davidson_flachaire_2008}:
\[ 
V_i^* =
  \begin{cases}
    -1       & \quad \text{with probability } \frac{1}{2}, \\
    1      & \quad \text{with probability } \frac{1}{2}.
  \end{cases}
\]

The wild bootstrap procedure then works as follows:

\begin{enumerate}

\item

Obtain the estimates $\tilde{\beta}_1$ and residuals $\tilde{\varepsilon}_i$ from the restricted model $Y_i = W_i' \beta_1 + e_i$.

\item

Generate the wild bootstrap error $\varepsilon_i^* = V_i^* \tilde{\varepsilon}_i$.

\item

Obtain $Y_i^* = W_i' \tilde{\beta}_1 + \varepsilon_i^*$. Then estimate the restricted model and obtain the restricted bootstrap residuals $\tilde{\varepsilon}_i^*$ using the bootstrap sample $\{ (Y_i^*, X_i')' \}_{i=1}^{n}$.

\item

Use $\tilde{\varepsilon}_i^*$ to compute the bootstrap test statistic $t_{HC,r_n}^*$.

\item

Repeat steps 2--4 $B$ times (e.g. $B=399$) and obtain the empirical distribution of the $B$ test statistics $t_{HC,r_n}^*$. Use this empirical distribution to compute the bootstrap critical values of the bootstrap $p$-values.

\end{enumerate}

The following result is true:

\begin{theorem}\label{asy_distr_t_r_n_hc_boot}
Let $\mathcal{U}_n = \{ (Y_i, X_i')' \}_{i=1}^{n}$. Under assumptions of Theorem~\ref{asy_distr_t_r_n_hc},
\[
F_{HC,n}^*(t) \rightarrow \Phi(t) \text{ in probability},
\]
for all $t$, as $n \rightarrow \infty$, where $F_{HC,n}^*(t)$ is the bootstrap distribution of $t_{HC,r_n}^*|\mathcal{U}_n$ and $\Phi(\cdot)$ is the standard normal CDF.
\end{theorem}

\subsection{Discussion of the Results}\label{asymptotics_summary}

My results show that the LM type test statistic, normalized by the number of restrictions $r_n$ that the null hypothesis imposes on the nonparametric model, is asymptotically standard normal. In this section, I discuss my results and compare them with those in the literature.

First, it is natural to compare my test with the test in \citet{gupta_2018}. While I build on his proof techniques, my test differs from his in two ways. First, \citet{gupta_2018} uses an FGLS type test based on the transformed data $\tilde{\Sigma}^{-1/2} Y = \tilde{\Sigma}^{-1/2} W \beta_1 + \tilde{\Sigma}^{-1/2} e$, where $\tilde{\Sigma} = diag(\tilde{\varepsilon}_1^2,...,\tilde{\varepsilon}_n^2)$. Second, while his test relies on the projection property of series estimators and the fact that $P' \tilde{\Sigma}^{-1} \tilde{\varepsilon} = (0_{m_n \times 1}', (Z' \tilde{\Sigma}^{-1} \tilde{\varepsilon})')'$, he uses the variance estimate based on the full vector of moment conditions, $(P' \tilde{\Sigma}^{-1} P)$. As a result, his test statistic is given by
\[
t_{HC,G} = \frac{\tilde{\varepsilon}' \tilde{\Sigma}^{-1} Z \left( Z' \tilde{\Sigma}^{-1} Z - Z' \tilde{\Sigma}^{-1} W (W' \tilde{\Sigma}^{-1} W)^{-1} W' \tilde{\Sigma}^{-1} Z \right)^{-1} Z' \tilde{\Sigma}^{-1} \tilde{\varepsilon} - r_n}{\sqrt{2 r_n}}
\]

In contrast, I explicitly take into account the fact that a subset of moment conditions is exactly equal to zero when estimating variance. Moreover, instead of transforming the data to eliminate heteroskedasticity, I use the raw data to estimate the model and then use a heteroskedasticity robust variance estimate, as originally proposed in \citet{white_1980}. As a result, my test resembles the parametric heteroskedasticity robust LM test proposed in \citet{wooldridge_1987} and takes the following form:
\[
t_{HC} = \frac{\tilde{\varepsilon}' \tilde{Z} (\tilde{Z}' \tilde{\Sigma} \tilde{Z})^{-1} \tilde{Z}' \tilde{\varepsilon} - r_n}{\sqrt{2 r_n}}
\]

It turns out that both these modifications are crucial, and I discuss them in more detail in Section~\ref{simulations}.

Compared with most other existing series-based specification tests, I derive a different normalization of the test statistic. I normalize the test statistic by the number of restrictions $r_n$, while the existing literature mostly uses the total number of parameters in the nonparametric model $k_n$ (see Theorem 1 in \citet{de_jong_bierens_1994}, Equation (2.1) in \citet{hong_white_1995}, and proof of Lemma 6.2 in \citet{donald_et_al_2003}). This difference can be viewed as a degrees of freedom correction. I will study its practical importance in Section~\ref{simulations}, while here I discuss its theoretical significance.

My approach differs from the approach used in the previous literature in how it copes with a key step in the proof, going from the semiparametric regression residuals $\tilde{\varepsilon}$ to the true errors $\varepsilon$. My approach relies on the projection property of series estimators to eliminate the estimation variance and hence only needs to deal with the approximation bias. Specifically, it uses the equality $\tilde{\varepsilon} = M_W \varepsilon + M_W R$ (see equation~(\ref{eqn_a1_hc}) in the Appendix), applies a central limit theorem from \citet{scott_1973} to the quadratic form in $M_W \varepsilon$, and bounds the remainder terms by requiring the approximation error $R$ to be small (see equation~(\ref{eqn_a2_hc}) in the Appendix).

The conventional approach does not impose any special structure on the model residuals and uses the equality $\tilde{\varepsilon} = \varepsilon + (g - \tilde{g})$. In parametric models, $g - \tilde{g} = X (\beta - \hat{\beta})$, and $\hat{\beta}$ is $\sqrt{n}$-consistent. This makes it possible to apply a central limit theorem for $U$-statistics to the quadratic form in $\varepsilon$ and bound the remainder terms that depend on $X (\beta - \hat{\beta})$. However, in semiparametric models this approach needs to deal with both the bias and variance of semiparametric estimators. Specifically, $g - \tilde{g} = R + W (\beta_1 - \tilde{\beta}_1)$, where $R$ can be viewed as the bias term and $W (\beta_1 - \tilde{\beta}_1)$ as the variance term. Thus, in order for $(g - \tilde{g})$ to be small, both bias and variance need to vanish sufficiently fast, and the resulting rate conditions turn out to be very restrictive. To see this, it is instructive to look at the rate conditions with and without the degrees of freedom correction.

With the degrees of freedom correction, rate condition~\ref{rate_cond_r_3_hc} is given by $n m_n^{-2\alpha}/ r_n^{1/2} \to 0$, while without the degrees of freedom correction it would become $(m_n + n m_n^{-2\alpha})/ r_n^{1/2} \to 0$. Thus, without the degrees of freedom correction, I would need to require $m_n^2/r_n \to 0$, while my rate conditions even allow for $m_n/r_n \to C$. Intuitively, the estimation variance in semiparametric model increases with the number of series terms $m_n$. With the degrees of freedom correction, the estimation variance is explicitly taken into account, so $m_n$ can grow fairly fast. In contrast, without the degrees of freedom correction, in order to ensure that the estimation variance is negligible, $m_n$ has to be small relative to $r_n$.

In turn, the smaller $m_n$ is, the harder it is to satisfy the condition $n m_n^{-2\alpha}/ r_n^{1/2} \to 0$, which ensures that the bias is negligible. Without the degrees of freedom correction, $m_n$ has to grow very slowly to control the estimation variance. But small $m_n$ means that it is difficult to control the bias, and unknown functions need to be very well behaved to make the approximation error sufficiently small. More specifically, when the rate condition changes from $m_n/r_n \to C$ to $m_n^2/r_n \to 0$, $\alpha$ needs to increase by at least factor of two in order for the bias to remain negligible. Thus, the degrees of freedom correction allows me to use more series terms to estimate the semiparametric model and makes it possible to deal with a wider class of unknown functions that enter semiparametric models. The next section demonstrates that it also substantially improves the finite sample performance of the test.

\section{Simulations}\label{simulations}

In this section I study the finite sample behavior of different versions of the proposed test in a partially linear model, due to its popularity in applied work. I start by describing the setup of my simulations. I consider two sample sizes, $n=250$ and $n=1,000$. I generate the regressors as $X_{1i} = -2 + 4 (0.8 V_{1i} + 0.2 V_{2i})$ and $X_{2i} = -2 + 4 (0.2 V_{1i} + 0.8 V_{2i})$, where $V_{1i}$ and $V_{2i}$ are $U[0,1]$. By doing so, I obtain two regressors that are correlated (the correlation is close to 0.5) and bounded. I generate the errors as $\varepsilon_i \sim \text{i.n.i.d. } N(0,1+1.75 \exp(0.75(X_{1i}+X_{2i})))$. I simulate data from two data generating processes:

\begin{enumerate}

\item

Semiparametric partially linear, which corresponds to $H_0^{SP}$:
\begin{align*}
Y_i &= 3 + 2 X_{1i} + 2(\exp(X_{2i})-2 \ln(X_{2i}+3)) + \varepsilon_i
\end{align*}

\item

Nonparametric, which corresponds to $H_1$:
\begin{align*}
Y_i &= 3 + 2 X_{1i}+ 2(\exp(X_{2i})-2 \ln(X_{2i}+3)) + 1.21 \cos(X_{1i}-2) \sin(0.75 X_{2i}) + \varepsilon_i
\end{align*}

\end{enumerate}

Before I move on and discuss the behavior of the proposed test in finite samples, I need to make two choices in order to implement the test. First, I need to choose the family of basis functions; second, choose the number of series terms in the restricted and unrestricted models.

I use both power series and cubic splines (i.e. splines with $s = 3$ in Equation~(\ref{splines})) as basis functions due to their popularity. In turn, the choice of tuning parameters presents a big practical challenge in implementing the proposed test, as well as many other specification tests.\footnote{For a discussion of regularization parameters choice in the context of kernel-based tests, see a review by \citet{sperlich_2014}.} Instead of studying the behavior of my test for a given (arbitrary) number of series terms, I vary the number of terms in univariate series expansions from $a_n = 4$ to $a_n = 8$ when $n=250$ and $a_n = 9$ when $n=1,000$ to investigate how the behavior of the test changes as a result (see Equations~(\ref{series}) and~(\ref{splines}) for the definition of $a_n$). In order to alleviate numerical problems associated with the quickly growing number of series terms, I restrict the number of interaction terms (see Section~\ref{simulations_implementation} in the supplement for details). As a result, the total number of parameter $k_n$ ranges from 16 when $a_n = 4$ to 40 when $a_n = 8$ and 53 when $a_n = 9$. 

In my simulations, I study the finite sample behavior of several tests. First, I consider the asymptotic version of the proposed test with the degrees of freedom correction:
\[
t_{HC,r_n} = \frac{\xi_{HC} - r_n}{\sqrt{2 r_n}}  \overset{a}{\sim} N(0,1).
\]

Second, I include the asymptotic version of my test that normalizes the test statistic by the total number of parameters $k_n$, as opposed to the number of restrictions $r_n$, to illustrate the practical importance of the degrees of freedom correction. As discussed above, the normalization by $k_n$ is prevalent in the existing literature on consistent series-based specification tests in parametric models, see, e.g. \citet{hong_white_1995}. Next, I consider the wild bootstrap version of my test based on the Rademacher distribution.\footnote{The results for Mammen's distribution are very similar and are omitted for brevity.} Finally, I use the test statistic from \citet{gupta_2018} given by
\[
t_{HC,G} = \frac{\tilde{\varepsilon}' \tilde{\Sigma}^{-1} Z \left( Z' \tilde{\Sigma}^{-1} Z - Z' \tilde{\Sigma}^{-1} W (W' \tilde{\Sigma}^{-1} W)^{-1} W' \tilde{\Sigma}^{-1} Z \right)^{-1} Z' \tilde{\Sigma}^{-1} \tilde{\varepsilon} - r_n}{\sqrt{2 r_n}} \overset{a}{\sim} N(0,1).
\]

Figure~\ref{fig_simulated_size_power} plots the simulated size and power of the test as a function of the number of series terms in univariate series expansions $a_n$. For each sample size, I show the simulated size in the upper two plots and simulated power in the lower two plots. The red solid line corresponds to the asymptotic version of the proposed test with the degrees of freedom correction. The cyan dotted line corresponds to the asymptotic version of the proposed test without the degrees of freedom correction. The magenta dashed line corresponds to the bootstrap version of the proposed test. The blue dash-dotted line corresponds to the test from \citet{gupta_2018}.

There are several takeaways from the figure. First, the degrees of freedom correction is crucial for good finite sample performance of the test. The simulated size of the asymptotic test with the degrees of freedom correction is close to the nominal level, and the test has good power. In contrast, the asymptotic test without the degrees of freedom correction is severely undersized in all setups and therefore has low power. Second, the test from \citet{gupta_2018} does not seem to work well in finite samples. It is noticeably undersized and has virtually no power: its simulated power is roughly equal to the nominal size in all four setups. Finally, the performance of the asymptotic version of my test with the degrees of freedom correction is very similar to the performance of the wild bootstrap test even when $n=250$, indicating that my test has attractive finite sample properties that do not rely on the use of resampling methods.

Section~\ref{simulations_comparison} of the supplement provides a more detailed comparison of various versions of my test and the test from \citet{gupta_2018}. In short, I find that both using the heteroskedasticity robust variance estimate and using the variance estimate based on the subset of moment conditions, rather than all of them, is crucial for good finite sample behavior of the test.

As discussed above, the choice of tuning parameters plays an important role in my test. Though the results above demonstrate that the test is fairly robust to the choice of tuning parameters, it may be useful to have a data-driven way of choosing them. I investigate the finite sample performance of a data-driven test in Section~\ref{simulations_data_driven} of the supplement.

\section{Empirical Example}\label{empirical_example}

In this section, I apply the proposed test to the Canadian household gasoline consumption data from \citet{yatchew_no_2001}.\footnote{Available at \url{https://www.economics.utoronto.ca/yatchew/}.} They estimate gasoline demand ($y$), measured as the logarithm of the total distance driven in a given month, as a function of the logarithm of the gasoline price ($PRICE$), the logarithm of the household income ($INCOME$), the logarithm of the age of the primary driver of the car ($AGE$), and other variables ($z$), which include the logarithm of the number of drivers in the household ($DRIVERS$), the logarithm of the household size ($HHSIZE$), an urban dummy, a dummy for singles under 35 years old, and monthly dummies.

\citet{yatchew_no_2001} employ several demand models, including semiparametric specifications. They use differencing (see \citet{yatchew_1997} and \citet{yatchew_1998} for details) to estimate semiparametric models. The relevance of semiparametric models in gasoline demand estimation was first pointed out by \citet{hausman_newey_1995} and \citet{schmalensee_stoker_1999}. \citet{yatchew_no_2001} follow these papers in using semiparametric specifications and pay special attention to specification testing. However, they only compare semiparametric specifications with parametric ones, while I use series methods to estimate semiparametric specifications and implement the proposed specification test to assess their validity as compared to a general nonparametric model.

I focus on the model that is nonparametric in $AGE$ but parametric in $PRICE$ and $INCOME$ (roughly corresponds to Model (3.4) in \citet{yatchew_no_2001}):
\begin{eqnarray}\label{gasoline_sp_age}
y = \alpha_1 PRICE + \alpha_2 INCOME + g(AGE) + z' \beta + \varepsilon
\end{eqnarray}

In this model, the relationship between gasoline demand, price, and household income has a familiar log-log form that could be derived from a Cobb-Douglas utility function. The age of the primary driver enters the model nonparametrically to allow for possible nonlinearities, while the remaining demographics enter the model linearly. Next, I investigate whether demand model~(\ref{gasoline_sp_age}) is flexible enough as compared to a more general model.

In order to apply my test, I need to choose the series functions $P^{k_n}(x)$ or, equivalently, define a nonparametric alternative that can be approximated by these series functions. Ideally, I would want to use a fully nonparametric alternative
\[
y = h(PRICE, AGE, INCOME, z) + \varepsilon
\]
However, this is impractical in the current setting: the dataset from \citet{yatchew_no_2001} contains 12 monthly dummies, an urban dummy, and a dummy for singles under 35 years old. The fully nonparametric alternative would require me to completely saturate the model with the dummies, i.e. interact all series terms in continuous regressors with a full set of dummies. This would be equivalent to dividing the dataset into $12 \cdot 2 \cdot 2 = 48$ bins and estimating the model within each bin separately. Given that the total number of observations is 6,230, this would leave me with about 125 observations per bin on average and would make semiparametric estimation and testing problematic.

I choose a different approach to deal with the dummies. I maintain the assumption that the nonparametric alternative is separable in the continuous variables and the dummies. Separating $z$ into $z_1$, which includes the logarithm of the number of drivers in the household and the logarithm of the household size, and $z_2$, which includes the dummies, I consider the nonparametric alternative given by:
\begin{eqnarray}\label{gasoline_np}
y = h(PRICE, AGE, INCOME, z_1) + z_2' \lambda + \varepsilon
\end{eqnarray}

To implement my test, I first need to decide how many series terms in $AGE$ to use to estimate the semiparametric model~(\ref{gasoline_sp_age}). To guide this choice, I use Mallows's $C_p$ and generalized cross-validation, as discussed in Section 15.2 in \citet{li_racine_2007}. I try both power series and splines, and both procedures suggest that $a_n = 3$ (not including the constant) series terms in $AGE$ is optimal. Because power series and splines are identical for $a_n = 3$ as there are no knots yet, I only use power series with a cubic function of $AGE$ to estimate the model.

To construct the regressors $P^{k_n}$ used under the alternative, I use $a_n = 3$ power series terms in $AGE$, $PRICE$, and $INCOME$, $j_n = 2$ (not including the constant) power series terms in $DRIVERS$ and $HHSIZE$, and the set of dummies discussed above. I then use pairwise interactions (tensor products) of univariate power series, and add all possible three, four, and five element interactions between $AGE$, $PRICE$, $INCOME$, $DRIVERS$, and $HHSIZE$, without using higher powers in these interaction terms to avoid multicollinearity. This gives rise to $m_n = 21$ terms under the null, $k_n = 110$ terms under the alternative, and $r_n = 89$ restrictions.

I estimate specification~(\ref{gasoline_sp_age}) using series methods and the obtain the value of the heteroskedasticity robust test statistic $t_{HC,r_n} = 0.526$. The critical value at the 5\% level equals 1.645, so the null hypothesis that model~(\ref{gasoline_sp_age}) is correctly specified is not rejected.

This is in line with the results in \citet{yatchew_no_2001}: even though they do not test their semiparametric specifications against a general nonparametric alternative, they find no evidence against a specification similar to~(\ref{gasoline_sp_age}) when compared to the following semiparametric specification:
\[
y = g_1(PRICE,AGE) + \alpha_2 INCOME + z' \beta + \varepsilon
\]

Next, I investigate what would happen if instead of the semiparametric model~(\ref{gasoline_sp_age}) the researcher estimated the following parametric model:
\begin{eqnarray}\label{gasoline_p}
y = \alpha_1 PRICE + \alpha_2 INCOME + \gamma_0 + \gamma_1 AGE + z' \beta + \varepsilon
\end{eqnarray}

This model leads to $m_n = 19$ terms under the null, $k_n = 110$ terms under the altenative, and $r_n = 91$ restrictions. When testing it against specification~(\ref{gasoline_np}), I obtain $t_{HC,r_n} = 4.065$, so that specification~(\ref{gasoline_p}) is rejected at the 5\% significance level. Thus, it seems that controlling for $AGE$ flexibly is crucial in this gasoline demand application.

Finally, I use the proposed test to compare specification~(\ref{gasoline_p}) with the following semiparametric alternative:
\begin{eqnarray}\label{gasoline_sp}
y =  \alpha_1 PRICE + \alpha_2 INCOME + h_1(AGE, z_1) + z_2' \lambda + \varepsilon,
\end{eqnarray}
which is more restricted than the nonparametric model~(\ref{gasoline_np}). In this case, $P^{k_n}$ includes only series terms in $AGE$ and $z_1$ and their interactions. This results in $k_n = 40$ terms under the alternative and $r_n = 21$ restrictions. I obtain $t_{HC,r_n} = 10.187$, so that specification~(\ref{gasoline_p}) is again rejected at the 5\% significance level. As we can see, restricting the class of alternatives can improve the power of the test if the true model is close to the conjectured restricted class. However, this will also result in the loss of consistency against alternatives that do not belong to the restricted set.

For practical purposes, if the null model under consideration can be nested in several more general models, it may make sense to test it against these several alternatives simultaneously and apply the Bonferroni correction to control the test size. Using a general nonparametric alternative will result in consistency, while using more restrictive alternatives may result in better power if these alternatives are close to being correct.

\section{Conclusion}\label{conclusion}

In this paper, I develop a new heteroskedasticity robust Lagrange Multiplier type specification test for semiparametric conditional mean models. The proposed test achieves consistency by turning a conditional moment restriction into a growing number of unconditional moment restrictions using series methods. Because the number of series terms grows with the sample size, the usual asymptotic theory for the parametric Lagrange Multiplier test is no longer valid. I develop a new asymptotic theory that explicitly allows the number of terms to grow and show that the normalized test statistic converges in distribution to the standard normal. I demonstrate using simulations that my test outperforms the existing series-based specification tests. I apply the proposed test to the Canadian household gasoline consumption data from \citet{yatchew_no_2001} and find no evidence against one of the semiparametric specifications used in their paper. However, I show that my test does reject a less flexible parametric model.

An important avenue for future research is to develop a data-driven procedure to choose tuning parameters (the number of series terms under the null and alternative). An interesting possibility is to combine the test statistics computed over the range of tuning parameters, as suggested for parametric models in \citet{horowitz_spokoiny_2001} or \citet{gorgens_wurtz_2012}. Alternatively, one could use a modified version of the approach proposed in \citet{guay_guerre_2006} to pick tuning parameters in a data-driven way. I try the data-driven approach in simulations and find that it works well in finite samples. I plan to study whether any of these approaches can be formally extended to semiparametric models in future work.

\renewcommand\thesection{Appendix \Alph{section}}

\renewcommand\thesubsection{\Alph{section}.\arabic{subsection}}

\setcounter{section}{0}

\renewcommand{\thetheorem}{A.\arabic{theorem}}

\renewcommand{\thelemma}{A.\arabic{lemma}}

\renewcommand{\theassumption}{A.\arabic{assumption}}

\renewcommand{\theremark}{A.\arabic{remark}}

\renewcommand{\theequation}{A.\arabic{equation}}

\setcounter{theorem}{0}

\setcounter{lemma}{0}

\setcounter{assumption}{0}

\setcounter{remark}{0}

\clearpage

\section{Tables and Figures}\label{appendix_tables_figures}

\onehalfspacing

\begin{figure}[H]
\begin{center}
\caption{Simulated Size and Power of the Test}\label{fig_simulated_size_power}

$n=250$

\includegraphics[scale=0.33]{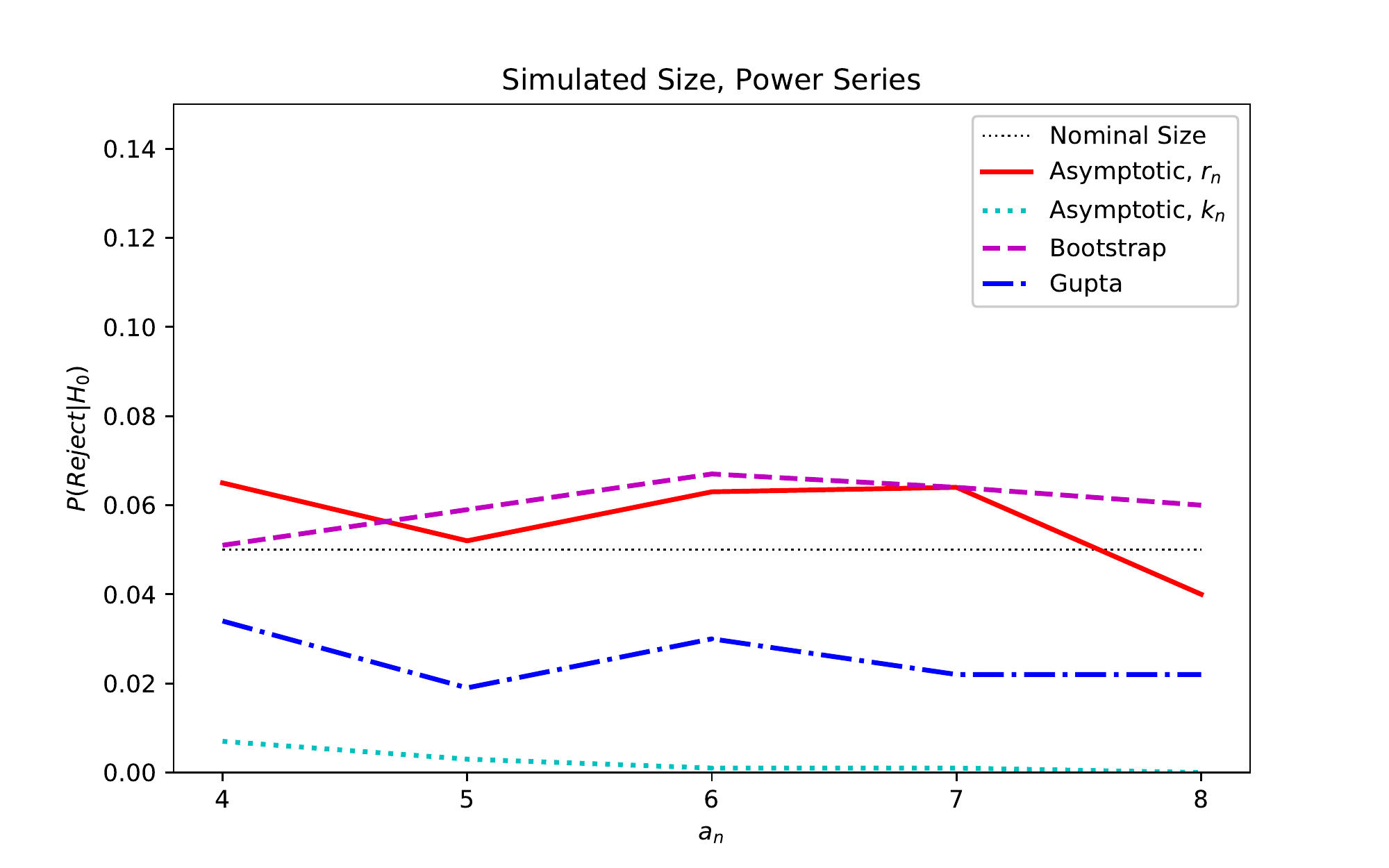} \includegraphics[scale=0.33]{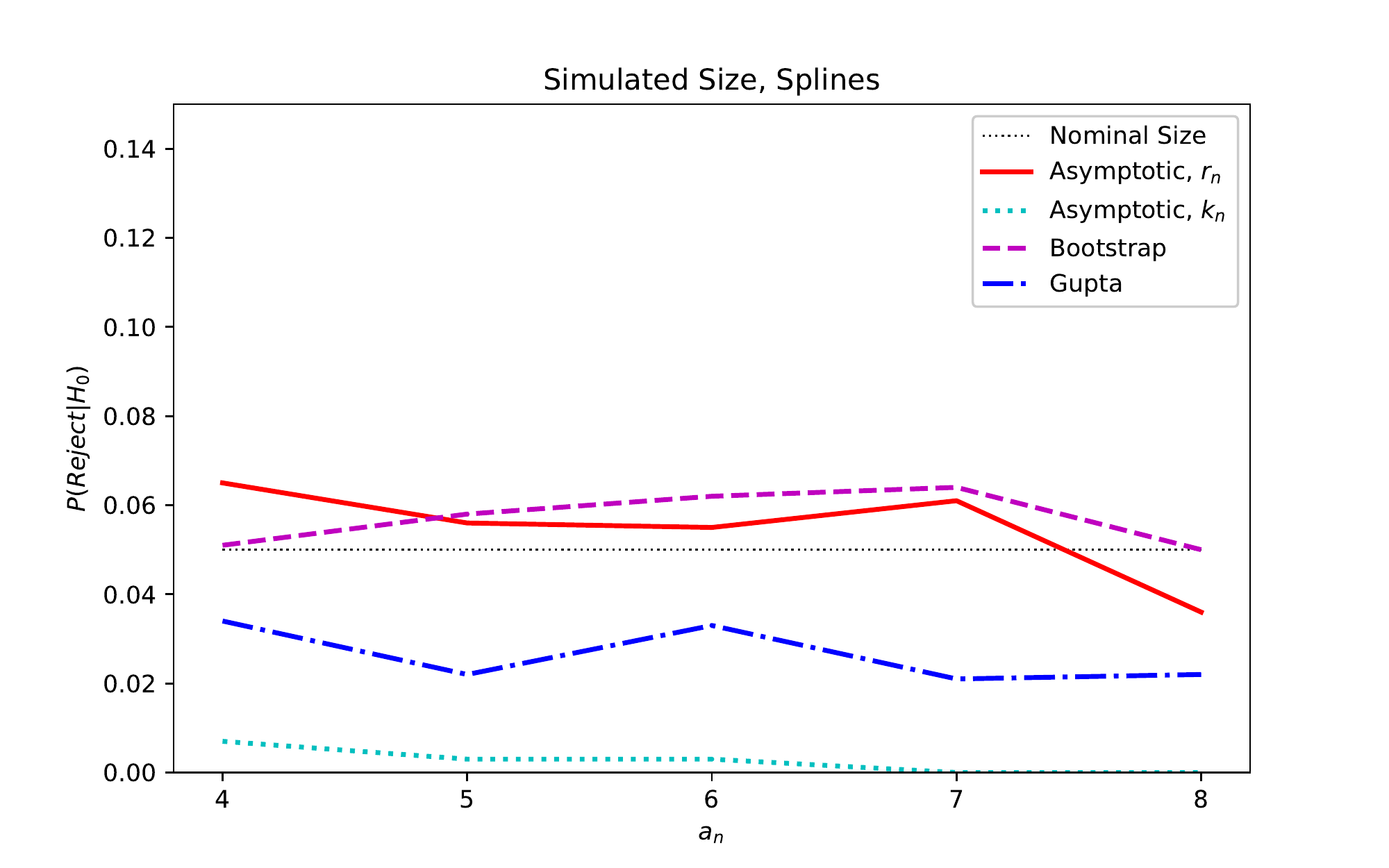}

\includegraphics[scale=0.33]{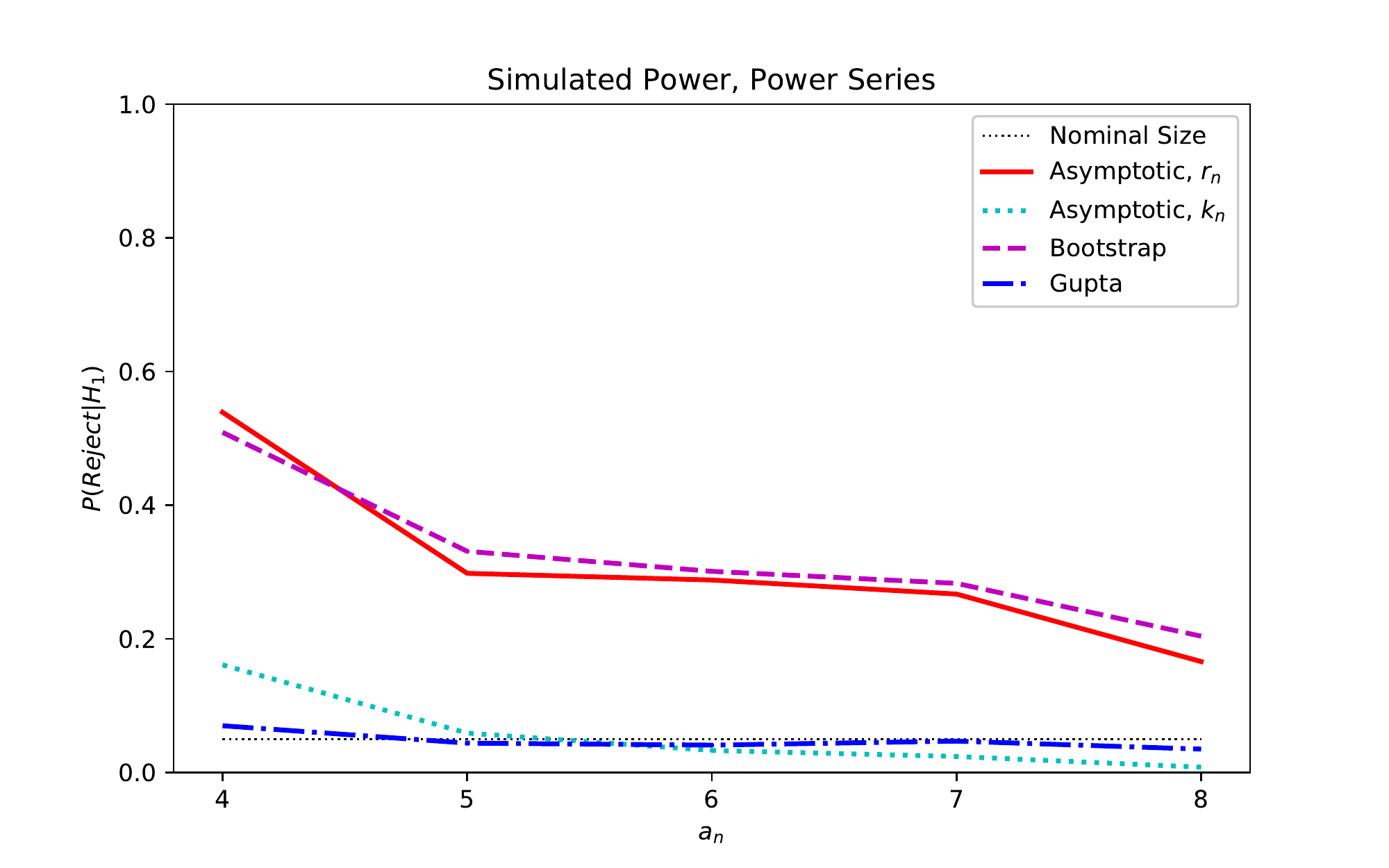} \includegraphics[scale=0.33]{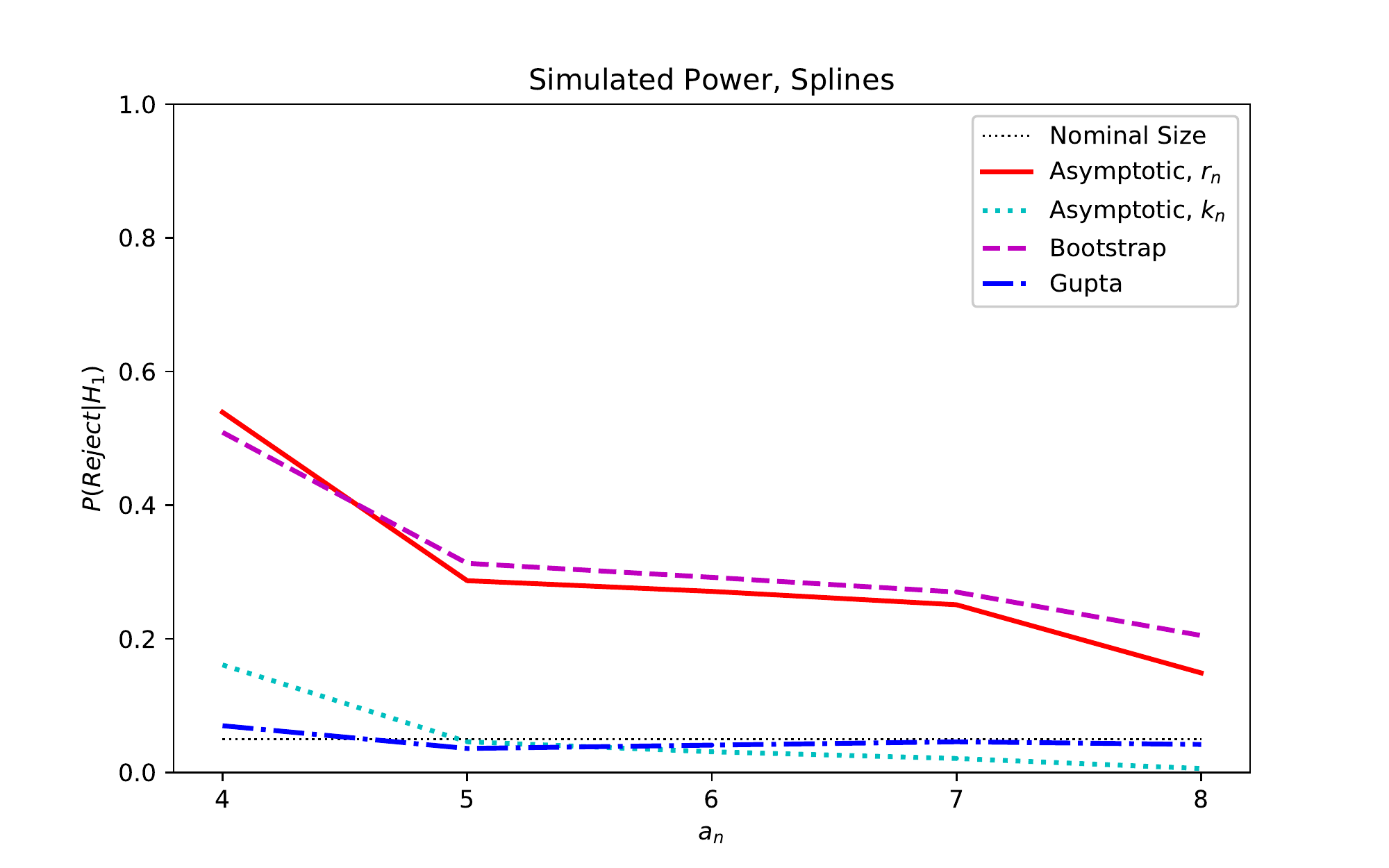}

$n=1,000$

\includegraphics[scale=0.33]{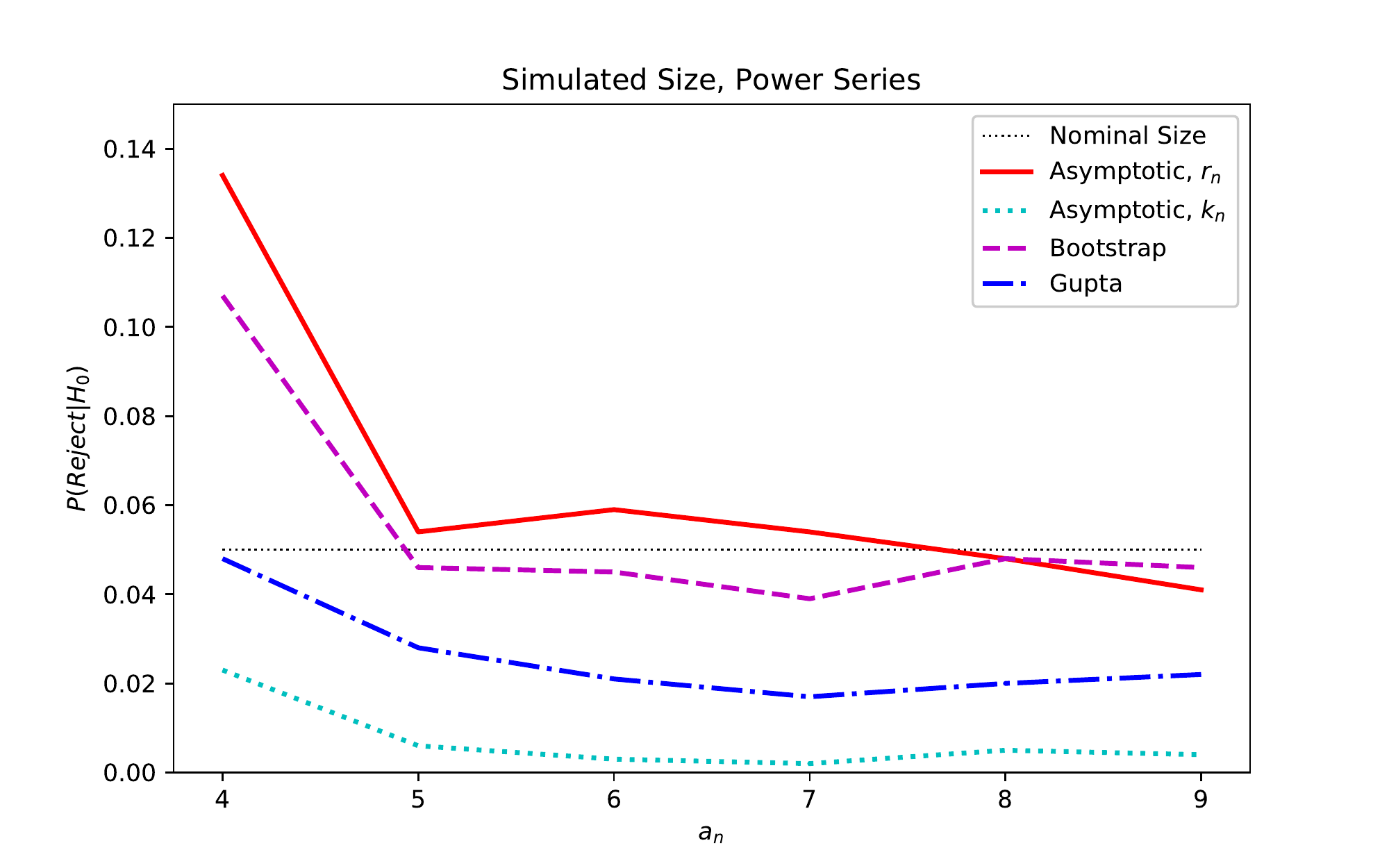} \includegraphics[scale=0.33]{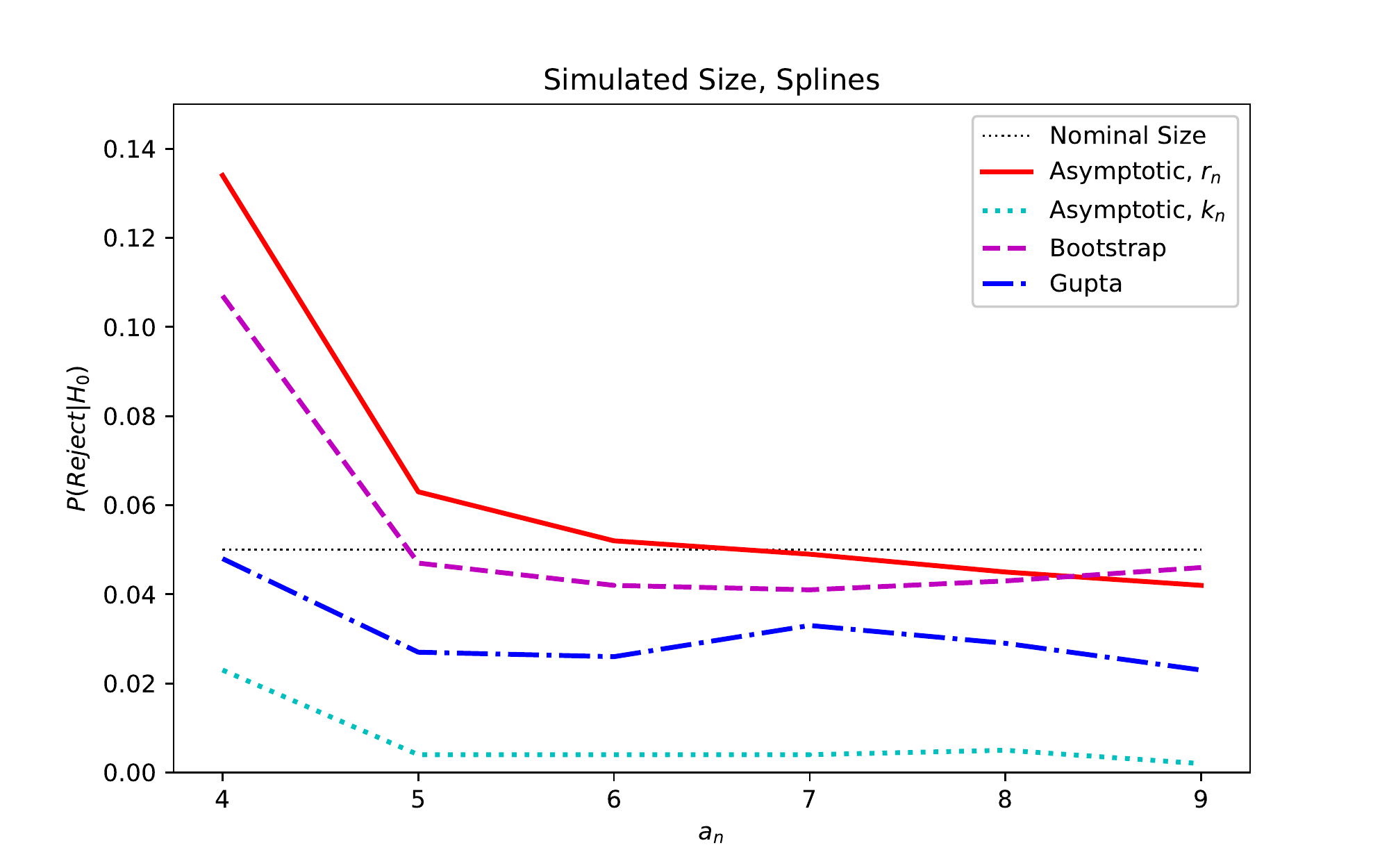}

\includegraphics[scale=0.33]{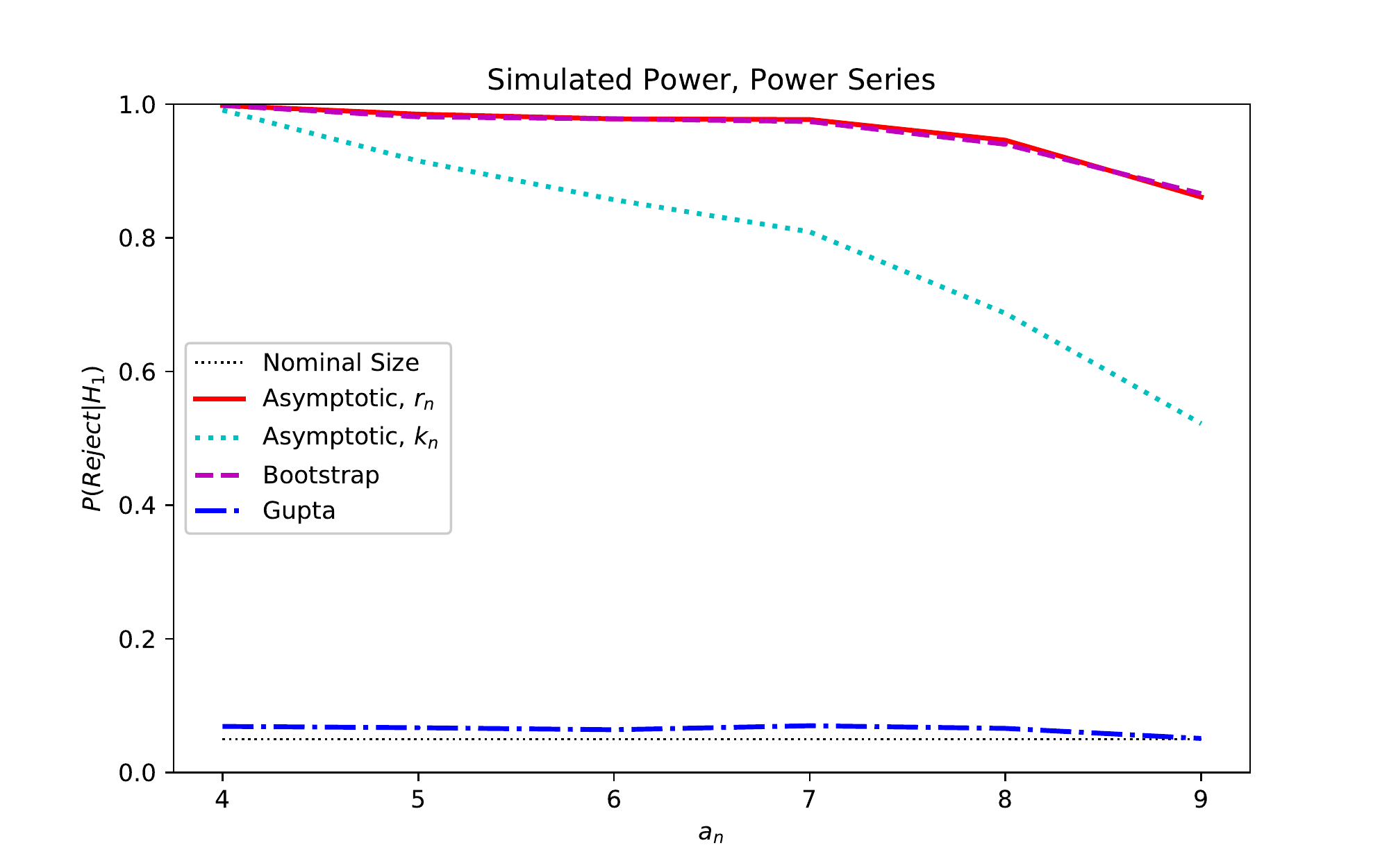} \includegraphics[scale=0.33]{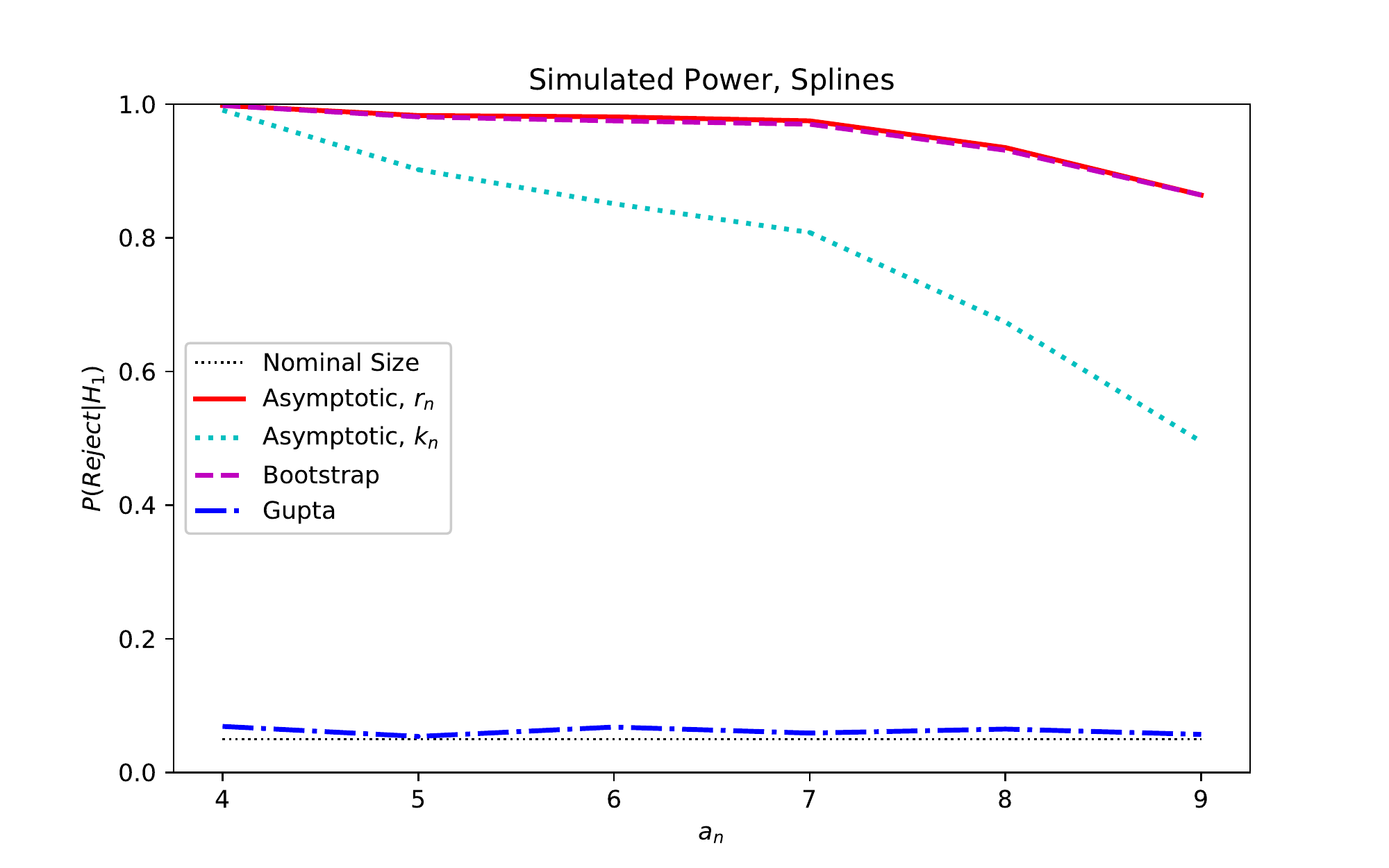}

\footnotesize{\vspace{-0.4cm}\singlespacing Left: power series. Right: splines. The results are based on $M=1,000$ simulation draws.}
\end{center}
\end{figure}

\clearpage \section{Proofs and Auxiliary Results}\label{appendix_proofs}

\onehalfspacing

\subsection{Auxiliary Lemmas}

In this section, I present auxiliary lemmas that will be used in my proofs. Their proofs are presented after the proofs of main results.

The following lemma shows that a convenient normalization can be used. This normalization is typical in the literature on series methods and will be used to simplify the subsequent proofs.

\begin{lemma}[\citet{donald_et_al_2003}, Lemma A.2]\label{series_normalization}

If Assumption \ref{series_norms_eigenvalues} is satisfied then it can be assumed without loss of generality that $\tilde{P}^{k}(x) = P^{k}(x)$ and that $E[P^{k}(X_i) P^{k}(X_i)'] = I_{k}$.

\end{lemma}

\begin{remark}
This normalization is common in the literature on series estimation, when all elements of $P^{k}(X_i)$ are used to estimate the model. In my setting, $P^{k}(X_i)$ is partitioned into $W^{m}(X_i)$, used in estimation, and $T^{r}(X_i)$, used in testing. The normalization implies that $W_i$ and $Z_i$ are orthogonal to each other. This can be justified as follows. Suppose that $W_i$ and $Z_i$ are not orthogonal. Then one can take all elements of $(W_i', Z_i')'$ and apply the Gram-Schmidt process to them. Because the orthogonalization process is sequential, it will yield the vector $(W_i^{0},Z_i^{0})'$ such that $W_i^{0}$ spans the same space as $W_i$, $Z_i^{0}$ spans the same space as $Z_i$, and $W_i^{0}$ and $Z_i^{0}$ are orthogonal. Thus, the normalization is indeed without loss of generality.
\end{remark}

The following lemma provides the rates of convergence (in the sample mean squared error sense) for semiparametric series estimators under the null hypothesis that will be used in my proofs.

\begin{lemma}[\citet{li_racine_2007}, Theorem 15.1]\label{series_f_rates}

Let $f(x) = f(x,\theta_0,h_0)$, $f_i = f(X_i)$, $\tilde{f}(x)= f(x, \tilde{\theta}, \tilde{h}) = W^{m_n}(x)' \tilde{\beta}_1$, and $\tilde{f}_i = \tilde{f}(X_i) = W_i' \tilde{\beta}_1$. Under Assumptions \ref{dgp}, \ref{series_norms_eigenvalues}, and \ref{series_approx}, the following is true:
\[
\frac{1}{n}\sum_{i=1}^{n}{(\tilde{f}_i - f_i)^2} = O_p(m_n/n + m_n^{-2\alpha})
\]

\end{lemma}

\begin{remark}[Convergence Rates for Semiparametric Series Estimators]

The exact rates given in Lemma~\ref{series_f_rates} are derived in \citet{li_racine_2007} for series estimators in nonparametric models. However, as other examples in Chapter 15 in \citet{li_racine_2007} show, similar rates can be derived in a wide class of semiparametric models, such as partially linear, varying coefficient, or additive models (see Theorems 15.5 and 15.7 in \citet{li_racine_2007}). In each of these cases, it is possible to replace the rate in Lemma~\ref{series_f_rates} with the rate for the particular case of interest. With appropriately modified rate conditions and assumptions, the results developed below will continue to hold.

\end{remark}

\begin{lemma}\label{lemma_ZZ}
Suppose that assumptions of Theorem~\ref{asy_distr_t_r_n_hc} hold. Then $\| \tilde{Z}'\tilde{Z}/n - Z'Z/ n \| = o_p(1/\sqrt{r_n})$. Moreover, the smallest and largest eigenvalues of $\tilde{Z}'\tilde{Z}/n$ converge in probability to one.
\end{lemma}

\begin{lemma}\label{lemma_Z}
Suppose that assumptions of Theorem~\ref{asy_distr_t_r_n_hc} hold. Then $\sup_{x \in \mathcal{X}}{\| \tilde{Z}^{r}(x) \|} \leq \zeta(r)(1 + o_p(1))$ and $\| Z^{r}(X_i) \| = O_p(r_n^{1/2})$.
\end{lemma}

\begin{lemma}\label{omegas_hc}
Let  $\tilde{\Omega}_{HC} = \tilde{Z}' \tilde{\Sigma} \tilde{Z}/n$, $\bar{\Omega}_{HC} = \sum_{i}{\varepsilon_i^2 \tilde{Z}_i \tilde{Z}_i'}/n$, and $\Omega_{HC} = \sum_{i}{\sigma_i^2 \tilde{Z}_i \tilde{Z}_i'}/n$, where $\sigma_i^2 = E[\varepsilon_i^2 | X_i]$.Suppose that assumptions of Theorem~\ref{asy_distr_t_r_n_hc} hold. Then
\begin{align*}
\| \tilde{\Omega}_{HC} - \bar{\Omega}_{HC} \| &=  O_p \left( \zeta(r_n)^2 (m_n/n + m_n^{-2\alpha}) \right) \\
\| \bar{\Omega}_{HC} - \Omega_{HC} \| &= O_p(\zeta(r_n) r_n^{1/2}/n^{1/2})
\end{align*}

Moreover, $1/C \leq \lambda_{\min}(\Omega_{HC}) \leq \lambda_{\max}(\Omega_{HC}) \leq C$, and w.p.a. 1, $1/C \leq \lambda_{\min}(\tilde{\Omega}_{HC}) \leq \lambda_{\max}(\tilde{\Omega}_{HC}) \leq C$ and $1/C \leq \lambda_{\min}(\bar{\Omega}_{HC}) \leq \lambda_{\max}(\bar{\Omega}_{HC}) \leq C$.

\end{lemma}

\begin{lemma}\label{diff_r_n_small_hc}
Suppose that Assumptions of Theorem~\ref{asy_distr_t_r_n_hc} hold. Then
\begin{align}
\frac{\varepsilon' \tilde{Z} (n \tilde{\Omega})_{HC}^{-1} \tilde{Z}' \varepsilon - \varepsilon' \tilde{Z} (n \Omega_{HC})^{-1} \tilde{Z}' \varepsilon}{\sqrt{r_n}} \overset{p}{\to} 0
\end{align}
\end{lemma}

\begin{lemma}\label{asy_norm}
Suppose that Assumptions of Theorem~\ref{asy_distr_t_r_n_hc} hold. Let $\mathcal{H}_n =  \tilde{Z} (n \Omega_{HC})^{-1} \tilde{Z}'$. Then
\begin{align}
\frac{\varepsilon' \mathcal{H}_n \varepsilon - r_n}{\sqrt{2 r_n}} \overset{d}{\to} N(0,1)
\end{align}
\end{lemma}

\begin{lemma}\label{psi_exp_var}
Suppose that Assumptions of Theorem~\ref{asy_distr_t_r_n_hc} hold. Then, conditional on the elements of $\mathcal{H}_n$,
\[
E\left[ \frac{1}{\sqrt{2 r_n}} \left( \varepsilon' \mathcal{H}_n \varepsilon - r_n \right) \right] = 0
\]
and
\[
Var \left( \frac{1}{\sqrt{2 r_n}} \left( \varepsilon' \mathcal{H}_n \varepsilon - r_n \right) \right) \overset{p}{\to} 1
\]
\end{lemma}

\subsection{Proof of Theorem~\ref{asy_distr_t_r_n_hc}}

Recall that $\tilde{\Omega}_{HC} = \tilde{Z}' \tilde{\Sigma} \tilde{Z}/n$. The test statistic becomes
\begin{align}\label{eqn_a1_hc}
t_{HC, r_n} = \frac{\xi_{HC} - r_n}{\sqrt{2 r_n}} = \frac{\tilde{\varepsilon}' \tilde{Z} (\tilde{Z}' \tilde{\Sigma} \tilde{Z})^{-1} \tilde{Z}' \tilde{\varepsilon} - r_n}{\sqrt{2 r_n}} = \frac{n^{-1} (\varepsilon + R)' \tilde{Z} \tilde{\Omega}_{HC}^{-1} \tilde{Z}' (\varepsilon + R) - r_n}{\sqrt{2 r_n}}
\end{align}
The proof consists of several steps.

Step 1. Decompose the test statistic and bound the remainder terms.
\begin{align}\label{eqn_a2_hc}
n^{-1} (\varepsilon + R)' \tilde{Z} \tilde{\Omega}_{HC}^{-1} \tilde{Z}' (\varepsilon + R) &= n^{-1} \varepsilon' \tilde{Z} \tilde{\Omega}_{HC}^{-1} \tilde{Z}' \varepsilon + 2 n^{-1} R' \tilde{Z} \tilde{\Omega}_{HC}^{-1} \tilde{Z}' \varepsilon + n^{-1} R' \tilde{Z} \tilde{\Omega}_{HC}^{-1} \tilde{Z}' R
\end{align}

By Lemma~\ref{lemma_ZZ}, the smallest and largest eigenvalues of $\tilde{Z}'\tilde{Z}/n$ converge to one. Because $\tilde{Z}' \tilde{Z}/n$ and $\tilde{Z} \tilde{Z}'/n$ have the same nonzero eigenvalues, $\lambda_{\max}(\tilde{Z} \tilde{Z}'/n)$ converges in probability to 1. Moreover, by Lemma~\ref{omegas_hc}, the eigenvalues of $\tilde{\Omega}_{HC}$ are bounded below and above. Thus, by Assumption \ref{series_approx},  w.p.a.1,
\[
n^{-1} R' \tilde{Z} \tilde{\Omega}_{HC}^{-1} \tilde{Z}' R \leq C R' (n^{-1} \tilde{Z} \tilde{Z}') R \leq C R' R = O_p(n m_n^{-2\alpha})
\]

Next, w.p.a.1,
\[
\Big{|} n^{-1} R' \tilde{Z} \tilde{\Omega}_{HC}^{-1} \tilde{Z}' \varepsilon \Big{|} \leq \Big{|} C \lambda_{\max}(\tilde{Z} \tilde{Z}'/n) R' \varepsilon \Big{|} \leq \Big{|} C R' \varepsilon \Big{|} = O_p(n^{1/2} m_n^{-\alpha})
\]

Thus,
\[
n^{-1} (\varepsilon + R)' \tilde{Z} \tilde{\Omega}_{HC}^{-1} \tilde{Z}' (\varepsilon + R) = n^{-1} \varepsilon' \tilde{Z} \tilde{\Omega}_{HC}^{-1} \tilde{Z}' \varepsilon + O_p(n m_n^{-2\alpha}) + O_p(n^{1/2} m_n^{-\alpha})
\]

Step 2. Deal with the leading term.

By Lemma~\ref{diff_r_n_small_hc},
\begin{align}\label{eqn_a3_hc}
\frac{\varepsilon' \tilde{Z} (n \tilde{\Omega}_{HC})^{-1} \tilde{Z}' \varepsilon}{\sqrt{r_n}} = \frac{\varepsilon' \tilde{Z} (n \Omega_{HC})^{-1} \tilde{Z}' \varepsilon}{\sqrt{r_n}} + o_p(1)
\end{align}

Finally, by Lemma~\ref{asy_norm},
\begin{align}\label{eqn_a4_hc}
\frac{\varepsilon' \tilde{Z} (n \Omega_{HC})^{-1} \tilde{Z}' \varepsilon - r_n}{\sqrt{2 r_n}} \overset{d}{\to} N(0,1)
\end{align}

The result of the theorem follows from Equations~(\ref{eqn_a2_hc}), (\ref{eqn_a3_hc}), and (\ref{eqn_a4_hc}). \qed

\subsection{Proof of Theorem~\ref{global_alternative_t_r}}

Note that
\[
\xi_{HC} = \tilde{\varepsilon}' \tilde{Z} (\tilde{Z}' \tilde{\Sigma} \tilde{Z})^{-1} \tilde{Z}' \tilde{\varepsilon}
\]
where $\tilde{\varepsilon} = M_W Y = M_W (W \beta_1 + Z \beta_2 + R + \varepsilon) = \tilde{Z} \beta_2 + M_W R + M_W \varepsilon$. Thus,
\begin{align*}
&t_{HC, r_n} = \frac{\xi_{HC} - r_n}{\sqrt{2 r_n}} = \frac{n^{-1} \varepsilon' \tilde{Z} \tilde{\Omega}_{HC}^{-1} \tilde{Z}' \varepsilon - r_n}{\sqrt{2 r_n}} + \frac{n^{-1} R' \tilde{Z} \tilde{\Omega}_{HC}^{-1} \tilde{Z}' R}{\sqrt{2 r_n}} + 2 \frac{n^{-1} R' \tilde{Z} \tilde{\Omega}_{HC}^{-1} \tilde{Z}' \varepsilon}{\sqrt{2 r_n}} \\
&+\frac{n^{-1} \beta_2 \tilde{Z}' \tilde{Z} \tilde{\Omega}_{HC}^{-1} \tilde{Z}' \tilde{Z} \beta_2}{\sqrt{2 r_n}}
 + 2 \frac{n^{-1} R' \tilde{Z} \tilde{\Omega}_{HC}^{-1} \tilde{Z}' \tilde{Z} \beta_2}{\sqrt{2 r_n}} + 2 \frac{n^{-1} \varepsilon' \tilde{Z} \tilde{\Omega}_{HC}^{-1} \tilde{Z}' \tilde{Z} \beta_2}{\sqrt{2 r_n}}
\end{align*}

As shown in the proof of Theorem~\ref{asy_distr_t_r_n_hc}, the first term converges in distribution to $N(0,1)$ and thus is $O_p(1)$. Because the eigenvalues of $\tilde{\Omega}_{HC}$ are bounded away from zero, $\sup_{x \in \mathcal{X}}{|R(x)|} \overset{p}{\to} 0$ and $\varepsilon_i$ have mean zero, the second term is $o_p(n r_n^{-1/2})$, the third term is $o_p(n^{1/2} r_n^{-1/2})$, the fifth term is $o_p(n r_n^{-1/2})$, and the last term is $o_p(n^{1/2} r_n^{-1/2})$. As for the fourth term, because $\| \tilde{\Omega}_{HC} - \Omega_{\eta,HC} \| \overset{p}{\to} 0$ and the eigenvalues of $\Omega_{\eta,HC}$ are bounded away from zero, the eigenvalues of $\tilde{\Omega}_{HC}$ are also bounded away from zero w.p.a. 1. Thus, w.p.a. 1,
\[
\frac{n^{-1} \beta_2 \tilde{Z}' \tilde{Z} \tilde{\Omega}_{HC}^{-1} \tilde{Z}' \tilde{Z} \beta_2}{\sqrt{2 r_n}} \leq C \frac{n \beta_2' \beta_2}{\sqrt{2 r_n}} = O_p(n r_n^{-1/2})
\]

Thus, the fourth term dominates the other terms, and $t_{HC, r_n} = O_p(n r_n^{-1/2})$. Note also that $n^{-1} \beta_2 \tilde{Z}' \tilde{Z} \tilde{\Omega}_{HC}^{-1} \tilde{Z}' \tilde{Z} \beta_2 \geq 0$. The result of the theorem follows.  \qed

\subsection{Sketch of Proof of Theorem~\ref{asy_distr_t_r_n_hc_boot}}

The bootstrap test statistic is given by
\[
\xi_{HC}^* = \tilde{\varepsilon}^{*\prime} \tilde{Z} (\tilde{Z}' \tilde{\Sigma}^* \tilde{Z})^{-1} \tilde{Z}' \tilde{\varepsilon}^*,
\]
where $\tilde{\varepsilon}^* = M_W \varepsilon^*$ and $\tilde{\Sigma}^* = diag(\tilde{\varepsilon}_1^{*2},...,\tilde{\varepsilon}_m^{*2})$. Note that because $M_W$ is idempotent,
\[
\xi_{HC}^* = \varepsilon^{*\prime} \tilde{Z} (\tilde{Z}' \tilde{\Sigma}^* \tilde{Z})^{-1} \tilde{Z}' \varepsilon^*
\]

Next, one can use an analog of Lemma~\ref{diff_r_n_small_hc} to show that
\[
\frac{ \varepsilon^{*\prime} \tilde{Z} (\tilde{Z}' \tilde{\Sigma}^* \tilde{Z})^{-1} \tilde{Z}' \varepsilon^* -  \varepsilon^{*\prime} \tilde{Z} (\tilde{Z}' \Sigma^* \tilde{Z})^{-1} \tilde{Z}' \varepsilon^*}{\sqrt{r_n}} \overset{p}{\to} 0,
\]
where $\Sigma^* = diag(\tilde{\varepsilon}_1^2,...,\tilde{\varepsilon}_n^2)$.

Finally, one can apply Lemma~\ref{asy_norm} to $\varepsilon^{*\prime} \tilde{Z} (n \Omega_{HC}^*)^{-1} \tilde{Z}' \varepsilon^*$, where $\Omega_{HC}^* =  \tilde{Z}' \Sigma^* \tilde{Z}/n$, conditional on the data $\mathcal{U}_n$, to obtain convergence in distribution in probability. \qed

\subsection{Proofs of Auxiliary Lemmas}

\begin{proof}[Proof of Lemma~\ref{lemma_ZZ}]
Note that
\[
\tilde{Z}'\tilde{Z}/n = Z'Z/ n - (Z'W/n) (W'W/n)^{-1} (W'Z/ n)
\]

Thus,
\[
\| \tilde{Z}'\tilde{Z}/n - Z'Z/ n \| = \| (Z'W/n) (W'W/n)^{-1} (W'Z/ n) \|
\]

By Lemma 15.2 in \citet{li_racine_2007}, $E[ \| P'P/n - I_{k_n}\|^2 ] = O_p(\zeta(k_n)^2 k_n/n)$ and $\| P'P/n - I_{k_n}\| = O_p(\zeta(k_n) \sqrt{k_n/n})$. Similarly, it can be shown that $\| W'W/n - I_{m_n}\| = O_p(\zeta(m_n) \sqrt{m_n/n})$ and $\| Z'Z/ n - I_{r_n}\| = O_p(\zeta(r_n) \sqrt{r_n/n})$. Moreover, note that
\[
P'P/n - I_{k_n} = \begin{pmatrix}
W'W/n & W'Z/ n \\
Z'W/n & Z'Z/ n
\end{pmatrix} - 
\begin{pmatrix}
I_{m_n} & \textbf{0}_{m_n \times r_n} \\
\textbf{0}_{r_n \times m_n} & I_{r_n}
\end{pmatrix}
\]

Hence, $\| W'Z/ n \| = O_p(\zeta(k_n) \sqrt{k_n/n})$. Thus, the eigenvalues of $W'W/n$ are bounded below and above w.p.a.1, and
\begin{align*}
&\Big{\|} (Z'W/n) (W'W/n)^{-1} (W'Z/ n) \Big{\|} \leq C \Big{\|} (Z'W/n) (W'Z/ n) \Big{\|} \\
&\leq C \Big{\|} (Z'W/n) \Big{\|} \; \Big{\|} (W'Z/ n) \Big{\|} = O_p(\zeta(k_n)^2 k_n/n),
\end{align*}
where the last inequality is due to the fact that $\|A B \|^2 \leq \| A \|^2 \| B \|^2$.

Note that $r_n/k_n \to C \leq 1$. Thus, if condition~(\ref{rate_cond_r_2_hc}) holds, then $\zeta(k_n)^2 k_n r_n^{1/2}/n \to 0$, and thus $\| \tilde{Z}'\tilde{Z}/n - Z'Z/ n \| = o_p(1/\sqrt{r_n})$. This also implies that the smallest and largest eigenvalues of $\tilde{Z}'\tilde{Z}/n$ converge to one. 
\end{proof}

\begin{proof}[Proof of Lemma~\ref{lemma_Z}]

Note that $\tilde{Z}_i = Z_i - (Z'W/n) (W'W/n)^{-1} W_i$. Thus,
\begin{align*}
\sup_{x \in \mathcal{X}}{\| \tilde{Z}^{r}(x) \|} &= \sup_{x \in \mathcal{X}}{\Big{\|} Z^{r}(x) - (Z'W/n) (W'W/n)^{-1} W^{m}(x) \Big{\|}} \\
&\leq \sup_{x \in \mathcal{X}}{\| Z^{r}(x) \|} + \Big{\|}  (Z'W)/n (W'W/n)^{-1} \Big{\|}  \sup_{x \in \mathcal{X}}{\| W^{m}(x) \|} \\
&\leq \zeta(r) + C O_p \left( \zeta(k_n) \sqrt{k_n/n} \right) \zeta(m)  = \zeta(r)(1+o_p(1))
\end{align*}

Similarly,
\begin{align*}
&\| \tilde{Z}_i \| = \Big{\|} Z_i - (Z'W/n) (W'W/n)^{-1} W_i \Big{\|} \leq \| Z_i \| + \Big{\|} (Z'W/n) (W'W/n)^{-1} W_i \big{\|} \\
&= O_p(r_n^{1/2}) + C O_p \left( \zeta(k_n) \sqrt{k_n/n} \right) O_p(m_n^{1/2}) = O_p(r_n^{1/2})
\end{align*}

In both results, I rely on the following:
\[
\zeta(k_n) k_n^{1/2} m_n^{1/2}/(r_n^{1/2} n^{1/2}) \sim \zeta(k_n) m_n^{1/2}/n^{1/2} \to 0
\]
by rate condition~(\ref{rate_cond_r_2_hc}).

\end{proof}

\begin{proof}[Proof of Lemma~\ref{omegas_hc}]

First,
\begin{align*}
&\| \tilde{\Omega}_{HC} - \bar{\Omega}_{HC} \| = \| \sum_{i}{\tilde{Z}_i \tilde{Z}_i' (\tilde{\varepsilon}_i^2 - \varepsilon_i^2)}/n \| =  \| \sum_{i}{\tilde{Z}_i \tilde{Z}_i' \left( (\varepsilon_i + f_i - \tilde{f}_i)^2 - \varepsilon_i^2 \right)}/n \| \\
&= \| \sum_{i}{\tilde{Z}_i \tilde{Z}_i' \left( (f_i - \tilde{f}_i)^2 + 2 \varepsilon_i (f_i - \tilde{f}_i) \right)}/n \| \leq \sup_{i}{\| \tilde{Z}_i \|^2} \Big{|} \sum_{i}{\left( (f_i - \tilde{f}_i)^2 + 2 \varepsilon_i (f_i - \tilde{f}_i) \right)}/n \Big{|} \\
&= \zeta(r_n)^2 \left[ O_p \left( m_n/n + m_n^{-2\alpha} \right) + O_p \left( n^{-1/2} m_n^{1/2} (m_n/n + m_n^{-2\alpha})^{1/2} \right) \right] = O_p \left( \zeta(r_n)^2 (m_n/n + m_n^{-2\alpha}) \right)
\end{align*}

The following result can be obtained exactly as in Lemma A.6 in \citet{donald_et_al_2003}:
\begin{align*}
\| \bar{\Omega}_{HC} - \Omega_{HC} \| &= \| \sum_{i}{\tilde{Z}_i \tilde{Z}_i' (\varepsilon_i^2 - \sigma_i^2)}/n \| = O_p(\zeta(r_n) \sqrt{r_n/n})
\end{align*}

Finally, the results about the eigenvalues can also be obtained in the same way as in Lemma A.6 in \citet{donald_et_al_2003}. 

\end{proof}

\begin{proof}[Proof of Lemma~\ref{diff_r_n_small_hc}]

Given the result of Lemma~\ref{omegas_hc},
\begin{align*}
&\Bigg{|} \frac{\varepsilon' \tilde{Z} (n \tilde{\Omega})_{HC}^{-1} \tilde{Z}' \varepsilon}{\sqrt{2 r_n}} - \frac{\varepsilon' \tilde{Z} (n \Omega_{HC})^{-1} \tilde{Z}' \varepsilon}{\sqrt{2r_n}} \Bigg{|} = \Bigg{|} \frac{n (n^{-1} \varepsilon' \tilde{Z}) (\tilde{\Omega}_{HC}^{-1} - \Omega_{HC}^{-1}) (n^{-1} \tilde{Z}' \varepsilon)}{\sqrt{2r_n}} \Bigg{|} \\
&\leq \frac{n \| \Omega_{HC}^{-1} n^{-1} \tilde{Z}' \varepsilon \|^2 (\| \tilde{\Omega}_{HC} - \Omega_{HC} \| + C \| \tilde{\Omega}_{HC} - \Omega_{HC} \|^2)}{\sqrt{2r_n}}
\end{align*}

Next,
\[
\| \Omega_{HC}^{-1} (n^{-1} \tilde{Z}' \varepsilon) \| \leq C \sqrt{(n^{-1} \varepsilon' \tilde{Z}) \Omega_{HC} ^{-1} (n^{-1} \tilde{Z}' \varepsilon)} = O_p(\sqrt{r_n/n})
\]

Then
\begin{align*}
\frac{n \| \Omega_{HC}^{-1} n^{-1} \tilde{Z}' \varepsilon \|^2 (\| \tilde{\Omega}_{HC} - \Omega_{HC} \| + C \| \tilde{\Omega}_{HC} - \Omega_{HC} \|^2)}{\sqrt{2r_n}} = \frac{n O_p(r_n/n) o_p(1/\sqrt{r_n})}{\sqrt{2r_n}} = \frac{o_p(\sqrt{r_n})}{\sqrt{2r_n}} = o_p(1),
\end{align*}
provided that $\|\tilde{\Omega}_{HC} - \Omega_{HC} \| = o_p(1/\sqrt{r_n})$, which holds under rate conditions~(\ref{rate_cond_r_1_hc}) and (\ref{rate_cond_r_2_hc}).

\end{proof}

\begin{proof}[Proof of Lemma~\ref{asy_norm}]

Denote $\mathcal{H}_n =  \tilde{Z} (n \Omega_{HC})^{-1} \tilde{Z}'$. Because the limiting distribution is independent of the elements of $\mathcal{H}_n$, one can prove the results conditional on these. All expectations below will be conditional on the elements of $\mathcal{H}_n$.
Denote by $h_{ij,n}$ the $(i,j)$th element of $\mathcal{H}_n$. Note that
\begin{align*}
&\frac{1}{\sqrt{2 r_n}} \left( \varepsilon' \tilde{Z} (n \Omega_{HC})^{-1} \tilde{Z}' \varepsilon - r_n \right) = \frac{1}{\sqrt{2 r_n}} \left( \varepsilon' \mathcal{H}_n \varepsilon - r_n \right) \\
& = \frac{1}{\sqrt{2 r_n}} \left( \sum_{i=1}^{n}{(\varepsilon_i^2 - \sigma_i^2) h_{ii,n}} + 2 \sum_{i=1}^{n}{\varepsilon_i \sum_{j < i}{\varepsilon_j h_{ij,n}}} \right) = \sum_{i=1}^{n}{\psi_{in}},
\end{align*}
with
\[
\psi_{in} = \frac{1}{\sqrt{2 r_n}} (\varepsilon_i^2 - \sigma_i^2) h_{ii,n} +  \frac{2}{\sqrt{2 r_n}} \varepsilon_i \sum_{j < i}{\varepsilon_j h_{ij,n}}
\]

Next, I follow the approach used in \citet{gupta_2018} and verify the conditions of Theorem 2 in \citet{scott_1973}. I need to prove:
\begin{align}
&\sum_{i=1}^{n}{\psi_{in}^2 \mathbbm{1}\{ |\psi_{in} | \geq \epsilon \}} \overset{p}{\to} 0, \text{ for any } \epsilon > 0, \label{scott_suff_cond_1} \\
&\sum_{i=1}^{n}{E[\psi_{in}^2 | \varepsilon_j, j < i]} \overset{p}{\to} 1 \label{scott_suff_cond_2}
\end{align}

To show (\ref{scott_suff_cond_1}), I check the sufficient Lyapunov condition:
\[
\sum_{i=1}^{n}{E[|\psi_{in}|^{2 + \delta/2}]} \overset{p}{\to} 0
\]

To verify it, I show that
\begin{align}
&\frac{1}{r_n^{1+\delta/4}} \sum_{i=1}^{n}{E\left[ \big{|} \varepsilon_i^2 - \sigma_i^2 \big{|}^{2+\delta/2} \big{|} h_{ii,n} \big{|}^{2+\delta/2}\right]} \overset{p}{\to} 0 \label{lyapunov_1} \\
&\frac{1}{r_n^{1+\delta/4}} \sum_{i=1}^{n}{E\left[ | \varepsilon_i |^{2+\delta/2} \right] E\left[ \Big{|} \sum_{j < i}{\varepsilon_j h_{ij,n}} \Big{|}^{2+\delta/2} \right]} \overset{p}{\to} 0 \label{lyapunov_2}
\end{align}

Observe that (\ref{lyapunov_1}) is bounded by
\[
C \frac{n}{r_n^{1+\delta/4}} \max_{1 \leq i \leq n}{ \big{|} h_{ii,n} \big{|}^{2+\delta/2} } = O_p \left( \frac{n r_n^{2 + \delta/2}}{r_n^{1+\delta/4} n^{2+\delta/2}} \right) = O_p \left( \frac{r_n^{1 + \delta/4}}{n^{1+\delta/2}} \right) = o_p(1)
\]

Next, by Jensen's, $c_r$, von Bahr-Esseen, triangle, and Cauchy-Schwarz inequalities, (\ref{lyapunov_2}) is bounded by
\begin{align*} 
&C \frac{n}{r_n^{1+\delta/4}} \max_{1 \leq i \leq n}{E\left[ \Big{|} \sum_{j < i}{\varepsilon_j h_{ij,n}} \Big{|}^{2+\delta/2} \right]} 
\leq C \frac{n}{r_n^{1+\delta/4}} \max_{1 \leq i \leq n}{\left( E\left[ \Big{|} \sum_{j < i}{\varepsilon_j h_{ij,n}} \Big{|}^{4+\delta} \right] \right)^{1/2}} \\ 
&\leq C \frac{n}{r_n^{1+\delta/4}} \max_{1 \leq i \leq n}{\left( E\left[ \Big{|} \sum_{j < i}{\varepsilon_j^2 h_{ij,n}^2} \Big{|}^{2+\delta/2} \right] \right)^{1/2}} 
\leq C \frac{n}{r_n^{1+\delta/4}} \max_{1 \leq i \leq n}{\left( E\left[ \Big{|} \sum_{j < i}{\varepsilon_j^4 h_{ij,n}^4} \Big{|}^{1+\delta/4} \right] \right)^{1/2}} \\ 
&\leq C \frac{n}{r_n^{1+\delta/4}} \max_{1 \leq i \leq n}{\left( E\left[ \Big{|} \sum_{j < i}{(\varepsilon_j^4 - \mu_{4,i}) h_{ij,n}^4} \Big{|}^{1+\delta/4} \right] + \Big{|} \sum_{j < i}{\mu_{4,i} h_{ij,n}^4} \Big{|}^{1+\delta/4}\right)^{1/2}} \\ 
&\leq C \frac{n}{r_n^{1+\delta/4}} \max_{1 \leq i \leq n}{\left( \sum_{j < i}{E\left[ \big{|} \varepsilon_j^4 - \mu_{4,i} \big{|}^{1 + \delta/4} \right] \big{|} h_{ij,n}^4 \big{|}^{1+\delta/4}} + \Big{|} \sum_{j < i}{h_{ij,n}^4} \Big{|}^{1+\delta/4}\right)^{1/2}} \\ 
&\leq C \frac{n}{r_n^{1+\delta/4}} \max_{1 \leq i \leq n}{ \left( \left(\sum_{j < i}{h_{ij,n}^4}\right)^{1+\delta/4} \right)^{1/2}} 
\leq C \frac{n}{r_n^{1+\delta/4}} \max_{1 \leq i \leq n}{ \left(\sum_{j < i}{h_{ij,n}^2}\right)^{1+\delta/4}} 
\end{align*}

Next, note that
\begin{align*}
&\sum_{j = 1}^{n}{h_{ij,n}^2} = \tilde{Z}_i' (\tilde{Z}' \Sigma \tilde{Z})^{-1} \tilde{Z}' \tilde{Z} (\tilde{Z}' \Sigma \tilde{Z})^{-1} \tilde{Z}_i \\
&= n^{-1}  \tilde{Z}_i' (\tilde{Z}' \Sigma \tilde{Z}/n)^{-1} (\tilde{Z}' \tilde{Z}/n) (\tilde{Z}' \Sigma \tilde{Z}/n)^{-1} \tilde{Z}_i \\
&\leq n^{-1} \| \tilde{Z}_i \|^2 \| (\tilde{Z}' \Sigma \tilde{Z}/n)^{-1} \|^2 \| (\tilde{Z}' \tilde{Z}/n) \|  = n^{-1} C \| \tilde{Z}_i \|^2 = O_p(n^{-1} r_n)
\end{align*}

Thus,
\[
\frac{n}{r_n^{1+\delta/4}} \max_{1 \leq i \leq n}{ \left(\sum_{j < i}{h_{ij,n}^2}\right)^{1+\delta/4}}  = O_p\left( \frac{n r_n^{1+\delta/4}}{n^{1+\delta/4} r_n^{1+\delta/4}} \right) = O_p(n^{-\delta/4}) = o_p(1)
\]

Hence, (\ref{scott_suff_cond_1}) holds. Next, I show that (\ref{scott_suff_cond_2}) is satisfied. Because
\[
\psi_{in}^2 = \frac{1}{2 r_n}(\varepsilon_i^4 - 2 \sigma_i^2 \varepsilon_i^2 + \sigma_i^4) h_{ii,n}^2 + \frac{4}{2 r_n} (\varepsilon_i^2 - \sigma_i^2) h_{ii,n} \varepsilon_i \sum_{j < i}{\varepsilon_j h_{ij,n}} + \frac{4}{2 r_n} \varepsilon_i^2 \left( \sum_{j < i}{\varepsilon_j h_{ij,n}} \right)^2,
\]
one can show that
\begin{align*}
&\sum_{i=1}^{n}{E[\psi_{in}^2 | \varepsilon_j, j < i]} - 1 = \frac{1}{2 r_n} \sum_{i=1}^{n}{(\mu_{4,i} - 3 \sigma_i^4) h_{ii,n}^2} + \frac{4}{2 r_n}\sum_{i=1}^{n}{\sum_{j < i}{\sum_{\begin{subarray}{c} l<i \\ l\neq j \end{subarray}}{\varepsilon_j \varepsilon_l h_{ij,n} h_{il,n}}}} \\
&+\frac{4}{2 r_n} \sum_{i=1}^{n}{\sigma_i^2 \sum_{j < i}{(\varepsilon_j^2 - \sigma_j^2) h_{ij,n}^2}} + \frac{4}{2 r_n} \sum_{i=1}^{n}{\mu_{3,i} h_{ii,n} \sum_{j < i}{\varepsilon_j h_{ij,n}}} \equiv I_1 + I_2 + I_3 + I_4
\end{align*}

As shown in the proof of Lemma~\ref{psi_exp_var} below, $I_1 = O_p(n^{-1} r_n) = o_p(1)$. Next, $I_2$ has zero mean and conditional variance bounded by
\begin{align}
&\frac{C}{r_n^2} \sum_{i=1}^{n}{\sum_{k=1}^{n}{\sum_{j < i}{\sum_{l < k}{\sigma_i^2 \sigma_k^2 \sigma_j^2 \sigma_l^2 h_{ij,n} h_{il,n} h_{kj,n} h_{kl,n}}}}} \leq \frac{C}{r_n^2} \sum_{i=1}^{n}{\sum_{k=1}^{n}{\sum_{j=1}^{n}{\sum_{l=1}^{n}{\sigma_i^2 \sigma_k^2 \sigma_j^2 \sigma_l^2 h_{ij,n} h_{il,n} h_{kj,n} h_{kl,n}}}}} \nonumber \\
&=\frac{C}{r_n^2} \sum_{j=1}^{n}{\sum_{l=1}^{n}{\sum_{i=1}^{n}{\sigma_l \sigma_j \sigma_i^2 h_{ij,n} h_{il,n}} \sum_{k=1}^{n}{\sigma_l \sigma_j \sigma_k^2 h_{kj,n} h_{kl,n}}}} \label{temp_1}
\end{align}

Denote $\tilde{\mathcal{H}}_n = \Sigma^{1/2} \mathcal{H}_n \Sigma^{1/2}$. Note that $\tilde{\mathcal{H}}_n$ is idempotent and symmetric and the $(i,j)$th element of $\tilde{\mathcal{H}}_n$ is given by $\tilde{h}_{ij,n} = \sigma_i \sigma_j h_{ij,n}$.
\[
\sigma_i \sigma_j h_{ij,n} = \tilde{h}_{ij,n} = \sum_{l=1}^{n}{\tilde{h}_{il,n} \tilde{h}_{lj,n}} = \sum_{l=1}^{n}{\sigma_i \sigma_j \sigma_l^2 h_{il,n} h_{lj,n}}
\]

Thus, it is also true that $h_{ij,n} = \sum_{l=1}^{n}{\sigma_l^2 h_{il,n} h_{lj,n}}$.

Thus, (\ref{temp_1}) becomes
\begin{align*}
\frac{C}{r_n^2} \sum_{j=1}^{n}{\sum_{l =1}^{n}{\sigma_j^2 \sigma_l^2 h_{jl,n}^2}} = \frac{1}{r_n^2} r_n = r_n^{-1}
\end{align*}

Thus, $I_2 \overset{p}{\to} 0$. Next, $I_3$ also has zero mean and conditional variance bounded by
\begin{align}
&\frac{C}{r_n^2} \sum_{i=1}^{n}{\sum_{k=1}^{n}{\sum_{j < i, k}{\sigma_i^2 \sigma_k^2 (\mu_{4,j} - \sigma_j^4) h_{ij,n}^2 h_{kj,n}^2}}} 
\leq \frac{C}{r_n^2} \sum_{i=1}^{n}{\sum_{k=1}^{n}{\sum_{j=1}^{n}{\sigma_i^2 \sigma_k^2 (\mu_{4,j} - \sigma_j^4) h_{ij,n}^2 h_{kj,n}^2}}} \nonumber \\
&\leq \frac{C}{r_n^2} \sum_{j=1}^{n}{\sum_{i=1}^{n}{\sigma_i^2 h_{ij,n}^2} \sum_{k=1}^{n}{\sigma_k^2  h_{kj,n}^2}} 
= \frac{C}{r_n^2} \sum_{j=1}^{n}{h_{jj,n}^2} = \frac{C}{r_n^2} O_p(n^{-1} r_n^2) = O_p(n^{-1})
\end{align}

Thus, $I_3 \overset{p}{\to} 0$. Finally, $I_4$ also has zero mean conditional variance bounded by
\begin{align}
&\frac{C}{r_n^2} \sum_{i=1}^{n}{\sum_{k=1}^{n}{\sum_{j < i, k}{\mu_{3,i} \mu_{3,k} h_{ii,n} h_{kk,n} \sigma_j^2 h_{ij,n} h_{kj,n}}}}
\leq \frac{C}{r_n^2} \sum_{i=1}^{n}{\sum_{k=1}^{n}{h_{ii,n} h_{kk,n} \sum_{j=1}^{n}{ \sigma_j^2 h_{ij,n} h_{kj,n}}}} \nonumber \\
&= \frac{C}{r_n^2} \sum_{i=1}^{n}{\sum_{k=1}^{n}{h_{ii,n} h_{kk,n} h_{ik,n}}} = \frac{C}{r_n^2} \max_{1 \leq i, k \leq n}{h_{ik,n}} \left(\sum_{i=1}^{n}{h_{ii,n}}\right)^2 \label{temp_2}
\end{align}

Note that because $C^{-1} \leq \sigma_i^2 \leq C$,
\[
C^{-1} \sum_{i=1}^{n}{\sigma_i^2 h_{ii}} \leq \sum_{i=1}^{n}{h_{ii}} \leq C \sum_{i=1}^{n}{\sigma_i^2 h_{ii}}
\]

Hence, $\sum_{i=1}^{n}{h_{ii}} = O_p(r_n)$, and (\ref{temp_2}) is
\[
\frac{C}{r_n^2} O_p(n^{-1} r_n) O_p(r_n^2) = O_p(n^{-1} r_n),
\]

Therefore, $I_4 \overset{p}{\to} 0$, (\ref{scott_suff_cond_2}) holds, and the lemma is now proved.

\end{proof}

\begin{proof}[Proof of Lemma~\ref{psi_exp_var}]

First,
\begin{align*}
&E\left[\varepsilon' \mathcal{H}_n \varepsilon \right] = E \left[ tr \left( \varepsilon' \tilde{Z} (n \Omega_{HC})^{-1} \tilde{Z}' \varepsilon \right) \right]  = E \left[ tr \left( (n \Omega_{HC})^{-1} \tilde{Z}' \varepsilon \varepsilon' \tilde{Z} \right) \right]  \\
&= E \left[ tr \left( (\tilde{Z}' \Sigma \tilde{Z})^{-1} \tilde{Z}' \varepsilon \varepsilon' \tilde{Z} \right) \right]  = tr(I_{r_n}) = r_n
\end{align*}

Next,
\begin{align*}
Var \left( \frac{1}{\sqrt{2 r_n}} \left( \varepsilon' \mathcal{H}_n \varepsilon - r_n \right) \right) = \frac{1}{2 r_n} Var \left(\varepsilon' \mathcal{H}_n \varepsilon\right) = \frac{1}{2 r_n} \left( E\left[ \left(\varepsilon' \mathcal{H}_n \varepsilon\right)^2 \right] - r_n^2 \right)
\end{align*}

Now, note that
\begin{align*}
&E\left[ \left(\varepsilon' \mathcal{H}_n \varepsilon\right)^2 \right] = E\left[ \left( \sum_{i=1}^{n}{\sum_{j=1}^{n}{\varepsilon_i \varepsilon_j h_{ij,n}}} \right)^2 \right] = E\left[ \sum_{i=1}^{n}{\sum_{j=1}^{n}{\sum_{k=1}^{n}{\sum_{l=1}^{n}{\varepsilon_i \varepsilon_j \varepsilon_k \varepsilon_l h_{ij,n} h_{kl,n}}}}}\right] \\
&= E\left[ \sum_{i=1}^{n}{\varepsilon_i^4 h_{ii,n}^2} \right]  + E\left[ \sum_{i=1}^{n}{\sum_{k \neq i}{\varepsilon_i^2 \varepsilon_k^2 h_{ii,n} h_{kk,n}}} \right] + 2 E\left[ \sum_{i=1}^{n}{\sum_{j \neq i}{\varepsilon_i^2 \varepsilon_j^2 h_{ij,n}^2}} \right] \\
&= \sum_{i=1}^{n}{(\mu_{4,i} - 3 \sigma_i^4) h_{ii,n}^2} + \sum_{i=1}^{n}{\sum_{k = 1}^{n}{\sigma_i^2 \sigma_k^2 h_{ii,n} h_{kk,n}}} + 2 \sum_{i=1}^{n}{\sum_{j = 1}^{n}{\sigma_i^2 \sigma_j^2 h_{ij,n}^2}}
\end{align*}

Next,
\begin{align*}
&\sum_{i=1}^{n}{\sum_{k = 1}^{n}{\sigma_i^2 \sigma_k^2 h_{ii,n} h_{kk,n}}} = \left( \sum_{i=1}^{n}{\sigma_i^2 h_{ii,n}}\right)^2 = tr(\Sigma^{1/2} \mathcal{H}_n \Sigma^{1/2})^2 \\
&= tr\left( \Sigma^{1/2} \tilde{Z} (\tilde{Z}' \Sigma \tilde{Z})^{-1} \tilde{Z}' \Sigma^{1/2} \right)^2 = tr \left( (\tilde{Z}' \Sigma \tilde{Z})^{-1}  \tilde{Z}' \Sigma \tilde{Z} \right)^2 = r_n^2
\end{align*}
and
\begin{align*}
&\sum_{i=1}^{n}{\sum_{j = 1}^{n}{\sigma_i^2 \sigma_j^2 h_{ij,n}^2}} = tr \left( (\Sigma^{1/2} \mathcal{H}_n \Sigma^{1/2})^2 \right) = tr\left( \Sigma^{1/2}  \tilde{Z} (\tilde{Z}' \Sigma \tilde{Z})^{-1}\tilde{Z}' \Sigma  \tilde{Z} (\tilde{Z}' \Sigma \tilde{Z})^{-1} \tilde{Z}' \Sigma^{1/2} \right) = r_n
\end{align*}

Finally,
\begin{align*}
\sum_{i=1}^{n}{(\mu_{4,i} - 3 \sigma_i^4) h_{ii,n}^2} \leq C \sum_{i=1}^{n}{h_{ii,n}^2} = O_p(n^{-1} r_n^2),
\end{align*}
because
\[
| h_{ij,n} | = | \tilde{Z}_i' (n \Omega_{HC})^{-1} \tilde{Z}_j | = O_p(n^{-1} \| \tilde{Z}_i \| \: \| \tilde{Z}_j \|) = O_p(n^{-1} r_n)
\]

Thus,
\[
Var \left( \frac{1}{\sqrt{2 r_n}} \left( \varepsilon' \mathcal{H}_n \varepsilon - r_n \right) \right) = \frac{1}{2 r_n} (r_n^2 + 2 r_n + O_p(n^{-1} r_n^2) - r_n^2) = 1 + O_p(n^{-1} r_n) = 1 + o_p(1)
\]

\end{proof}

\clearpage

\bibliography{literature_spec_testing}

\begin{mytitlepage}

\setcounter{page}{1}

\title{Supplement to ``A Consistent Heteroskedasticity Robust LM Type Specification Test for Semiparametric Models''}
\author{Ivan Korolev}
\date{November 7, 2019}

\maketitle

\begin{abstract}
This supplement is divided in several sections. Section~\ref{implementation} contains additional details on the implementation of the test and computation of the test statistic. Section~\ref{simulations_additional} describes the implementation of my simulation analysis in greater detail and presents additional simulation results. \ref{supplement_appendix_tables_figures} contains all relevant tables and figures.
\end{abstract}

\end{mytitlepage}



\renewcommand\thesection{S.\arabic{section}}

\renewcommand\thesubsection{S.\arabic{section}.\arabic{subsection}}

\renewcommand\thetheorem{S.\arabic{theorem}}

\renewcommand\theassumption{S.\arabic{assumption}}

\setcounter{section}{0}

\newpage

\doublespacing

\section{Implementation of the Test}\label{implementation}

\subsection{Computing the Test Statistic}

In this section I describe how to implement the proposed test. The test statistic can be computed using the following steps:

\begin{enumerate}

\item

Pick the sequence of approximating functions $W^{m_n}(x) = (W_1(x), ..., W_{m_n}(x))'$ that will be used to estimate the semiparametric model. Let $W_i := W^{m_n}(X_i)$.

\item

Estimate the semiparametric model $Y_i = f(X_i,\theta,h) + \varepsilon_i \approx W_i' \beta_1 + \varepsilon_i$ using series methods. Obtain the estimates $\tilde{\beta}_1 = (W'W)^{-1} W'Y$ and residuals $\tilde{\varepsilon}_i = Y_i - W_i' \tilde{\beta}_1$.

\item

Pick the sequence of approximating functions $Z^{r_n}(x) = (Z_1(x), ..., Z_{r_n}(x))'$ that will complement $W^{m_n}(x)$ to form the matrix $P^{k_n}(x) = (W^{m_n}(x)', Z^{r_n}(x)')'$, $k_n = m_n + r_n$, which corresponds to a general nonparametric model. $P^{k_n}(x)$ should be able to approximate any unknown function in a desired class sufficiently well. Common choices of basis functions include power series  (see Equation~(\ref{series}) in the main text) and splines (see Equation~(\ref{splines}) in the main text). Let $Z_i := Z^{r_n}(X_i)$ and $P_i := P^{k_n}(X_i)$.

\item

Compute the quadratic form  $\xi_{HC} = \tilde{\varepsilon}' \tilde{Z} (\tilde{Z}' \tilde{\Sigma} \tilde{Z})^{-1} \tilde{Z}' \tilde{\varepsilon}$, where $\tilde{\Sigma} = diag(\tilde{\varepsilon}_1^2,...,\tilde{\varepsilon}_n^2)$ and $\tilde{Z}_i = Z_i - W_i' (W'W)^{-1} W'Z$ are the residuals from the regression of each element of $Z_i$ on $W_i$.

Note that $\xi_{HC}$ can be computed as $nR^2$ from the regression of 1 on $\tilde{Z}_i \tilde{\varepsilon}_i$. The matrix of regressors can be written as $\tilde{\Sigma}^{1/2} \tilde{Z}$, where $\tilde{\Sigma}^{1/2} = diag(\tilde{\varepsilon}_1,...,\tilde{\varepsilon}_n)$. Let $\iota_n$ be a $n$-vector of ones. Note that $\tilde{\Sigma}^{1/2} \iota_n = \tilde{\varepsilon}$. Then
\[
n R^2 = n \left( 1 - \frac{e'e}{\iota_n' \iota_n} \right) = n \frac{\iota_n' \tilde{\Sigma}^{1/2} \tilde{Z} (\tilde{Z}' \tilde{\Sigma} \tilde{Z})^{-1} \tilde{Z}' \tilde{\Sigma}^{1/2} \iota_n}{n} = \tilde{\varepsilon}' \tilde{Z} ( \tilde{Z}' \tilde{\Sigma}  \tilde{Z})^{-1}  \tilde{Z}' \tilde{\varepsilon}
\]

\item

Compute the test statistic which is asymptotically standard normal under the null:
\[
t_{HC} = \frac{\xi_{HC} - r_n}{\sqrt{2 r_n}} \overset{a}{\sim} N(0,1)
\]

Reject the null if $t > z_{1-\alpha}$, the $(1-\alpha)$-quantile of the standard normal distribution.

Alternatively, use the $\chi^2$ approximation directly: $\xi_{HC} \overset{a}{\sim} \chi^2(r_n)$, reject the null if $\xi_{HC} > \chi_{1-\alpha}^2(r_n)$, the $(1-\alpha)$-quantile of the $\chi^2$ distribution with $r_n$ degrees of freedom.

\end{enumerate}

\subsection{Computing the Bootstrap Test Statistic}

The restricted residuals are given by $\tilde{\varepsilon} = M_W Y$. The bootstrap data satisfies $Y^* = W \tilde{\beta}_1 + \varepsilon^*$. In turn, the bootstrap residuals are equal to $\tilde{\varepsilon}^* = M_W Y^* = M_W (W \tilde{\beta}_1 + \varepsilon^*) = M_W \varepsilon^*$. Thus, in fact, one does not need to obtain $Y_i^*$ and re-estimate the model during each bootstrap iteration. It suffices to compute $\tilde{\varepsilon}^* = M_W \varepsilon^*$. The bootstrap test statistic is then given by
\[
\xi_{HC}^* = \tilde{\varepsilon}^{*\prime} \tilde{Z} (\tilde{Z}' \tilde{\Sigma}^* \tilde{Z})^{-1} \tilde{Z}' \tilde{\varepsilon}^*,
\]
where $\tilde{\Sigma}^* = diag(\tilde{\varepsilon}_1^{*2},...,\tilde{\varepsilon}_n^{*2})$.

\section{Simulations}\label{simulations_additional}

\subsection{DGPs and Implementation Details}\label{simulations_implementation}

In my simulations, I use the data generating process given by
\begin{align*}
Y_i &= 3+ 2 X_{1i} + 2(\exp(X_{2i})-2 \ln(X_{2i}+3)) + \varepsilon_i \\
\end{align*}

To estimate the restricted model, I replace the function $g(x_2)$ with its series expansion: $g(x_2) \approx Q^{a_n}(x_2)' \gamma$, where $Q^{a_n}$ is an $a_n$-dimensional vector of approximating functions. I use power series with
\[
Q^{a_n}(x_2) = (1, x_2, ..., x_2^{a_n-1})'
\]
and cubic splines with
\[
Q^{a_n}(z) = (1, x_2, ..., x_2^3, \mathbbm{1}\{ x_2 > t_1 \} (x_2-t_1)^3, ..., \mathbbm{1}\{ x_2 > t_{a_n - 4} \} (x_2 - t_{a_n - 4})^3)
\]

When $a_n = 4$, power series and splines coincide, as there are no knots yet. When $a_n > 4$, I place knots uniformly at the empirical quantiles of $X_2$.

To compute the test statistic, I need to construct $P^{k_n}(X_{1i},X_{2i})$. In order to do this, I first construct series terms in $X_1$, $Q^{a_n}(x_1)$. Next, I have two options: I can either use all possible interactions of $Q^{a_n}(x_1)$ and $Q^{a_n}(x_2)$, which would lead to $k_n = a_n^2$, or I could restrict the number of interaction terms.

Using all possible interaction terms may lead to severe multicollinearity and unstable behavior of the test. In order to make the test more stable and include more series terms in the restricted model, in my main analysis I reduce the number of interaction terms as follows. Suppose that $a_n$ is the number of interaction terms in univariate series expansions in $X_1$ and $X_2$. To construct the interaction terms, let $\bar{a}_n = \max \{ \min \{a_n,5 \}, \floor{a_n^{0.9}} \}$. In other words,
\[ \bar{a}_n =
    \begin{cases}
    a_n       & \quad \text{if } a_n \leq 5 \\
    5 & \quad \text{if } 5 < a_n \leq 7 \\
    \floor{a_n^{0.9}} & \quad \text{if } a_n > 7
  \end{cases}
\]

Then I form series terms $\bar{Q}^{\bar{a}_n}(X_{1i})$ and $\bar{Q}^{\bar{a}_n}(X_{1i})$, drop the constant, and take all their element by element interactions to construct the interaction terms. The total number of terms under the alternative is given by $k_n = 2 a_n - 1 + (\bar{a}_n - 1)^2$. Table~\ref{tbl_number_terms} shows how the number of terms under the null, $m_n$, the number of terms under the alternative, $k_n$, and the number of restrictions, $r_n$, change with $a_n$. While my approach is somewhat arbitrary, it presents a practical way to accommodate larger $m_n$ without running into severe multicollinearity issues. Alternatively, I could restrict the growth of $m_n$ by requiring that $a_n$ not grow beyond 6 or 7, and then include all possible interaction terms. But because I am interested in estimating and testing the semiparametric model, I take the approach that allows me to estimate it more flexibly. I should also note that the approach that restricts the number of interaction terms is not unique to my paper: \citet{chen_2007} also uses it in her simulations in Section 2.4.

\subsection{Comparison of Various Tests}\label{simulations_comparison}

In the main text, I compare the performance of my test,
\begin{align}\label{korolev_main}
t_{HC} = \frac{\tilde{\varepsilon}' \tilde{Z} (\tilde{Z}' \tilde{\Sigma} \tilde{Z})^{-1} \tilde{Z}' \tilde{\varepsilon} - r_n}{\sqrt{2 r_n}}
\end{align}
with that of the test in \citet{gupta_2018}:
\begin{align}\label{gupta_main}
t_{HC,G} = \frac{\tilde{\varepsilon}' \tilde{\Sigma}^{-1} Z \left( Z' \tilde{\Sigma}^{-1} Z - Z' \tilde{\Sigma}^{-1} W (W' \tilde{\Sigma}^{-1} W)^{-1} W' \tilde{\Sigma}^{-1} Z \right)^{-1} Z' \tilde{\Sigma}^{-1} \tilde{\varepsilon} - r_n}{\sqrt{2 r_n}}
\end{align}

I find that my test has superior finite sample behavior, but it may not be exactly clear why. Namely, my test is different from the one in \citet{gupta_2018} in two respects. First, my test is based on the quadratic form in $\tilde{\varepsilon}$, i.e. the OLS residuals, while Gupta's test is based on the quadratic form in $\tilde{\Sigma}^{-1/2} \tilde{\varepsilon}$, i.e. the FGLS residuals. Second, when estimating variance, I only take into account the moments $\tilde{Z}' \tilde{\varepsilon}$ and use the estimate $\tilde{Z}' \tilde{\Sigma} \tilde{Z}$, while Gupta's test estimates the variance of the entire set of moment conditions $P' \tilde{\Sigma}^{-1} \tilde{\varepsilon}$ and then takes the corresponding block of its inverse, $( Z' \tilde{\Sigma}^{-1} Z - Z' \tilde{\Sigma}^{-1} W (W' \tilde{\Sigma}^{-1} W)^{-1} W' \tilde{\Sigma}^{-1} Z)^{-1}$.

To better understand the finite sample performance of the two tests, I consider two alternative test statistics. The first one corresponds to my test, i.e. is based on the OLS residuals, but uses the ``long'' variance estimate:
\begin{align}\label{korolev_alt}
t_{HC,alt} = \frac{\tilde{\varepsilon}' Z \left( Z' \tilde{\Sigma} Z - Z' \tilde{\Sigma}W (W' \tilde{\Sigma} W)^{-1} W' \tilde{\Sigma} Z \right)^{-1} Z' \tilde{\varepsilon} - r_n}{\sqrt{2 r_n}}
\end{align}

The second one corresponds to the test in \citet{gupta_2018}, i.e. is based on the FGLS residuals, but uses the ``short'' variance estimate:
\begin{align}\label{gupta_alt}
t_{HC,G,alt} = \frac{\tilde{\varepsilon}' \tilde{\Sigma}^{-1} \tilde{Z} (\tilde{Z}' \tilde{\Sigma}^{-1} \tilde{Z})^{-1} \tilde{Z}' \tilde{\Sigma}^{-1} \tilde{\varepsilon} - r_n}{\sqrt{2 r_n}}
\end{align}

Figure~\ref{fig_simulated_size_power_comparison} presents the simulated size and power of the four tests. For each sample size, the upper two plots show the simulated size and the lower two graphs show the simulated power. The red solid line corresponds to my test in Equation~(\ref{korolev_main}), the cyan dash-dotted line to my test in Equation~(\ref{korolev_alt}), the blue dashed line to Gupta's test in Equation~(\ref{gupta_main}), and the magenta dotted line to Gupta's test in Equation~(\ref{gupta_alt}). As we can see, both version of Gupta's test are undersized and have low power. In fact, the version with the ``short'' variance estimate performs even worse than the version with the ``long`` one. In turn, the version of my test with the ``long'' variance estimate is severely oversized.

Next, Figure~\ref{fig_simulated_size_power_comparison_infeasible} presents the simulated size and power of the infeasible tests that replace the estimated matrix $\tilde{\Sigma} = diag(\tilde{\varepsilon}_1^2,...,\tilde{\varepsilon}_n^2)$ with the true matrix $\Sigma$. The infeasible version of my test with the ``long'' variance estimate is still severely oversized. The infeasible version of my test with the ``short'' variance estimate is somewhat oversized when $n=250$ but has decent size control when $n=1,000$. The infeasible version of Gupta's test with the ``long'' variance estimate is actually very similar to the infeasible version of my test with the ``short'' variance estimate. Finally, the infeasible version of Gupta's test with the ``short'' variance estimate has the best size control, but also has somewhat lower power.

All in all, it appears that the problems with the test in \citet{gupta_2018} are caused by estimation of the variance. The infeasible versions of Gupta's test work perfectly fine, while the feasible versions do not. This is perhaps not surprising, as FGLS typically relies on consistency of the variance-covariance matrix estimate, which does not hold in the current setup. My test with the ``long'' variance estimate does not work either, but my test with the ``short'' variance estimate works. Thus, it appears that both using the OLS residuals, as opposed to FGLS, and using the ``short'' variance estimate are crucial for good finite sample behavior of my test.

\subsection{Test with Data-Driven Choice of Tuning Parameters}\label{simulations_data_driven}

Below, I analyze the finite sample performance of my test when tuning parameters are chosen in a data-driven way. I use Mallows's $C_p$ and generalized cross-validation, as discussed in Section 15.2 in \citet{li_racine_2007}, to select the number of series terms under the null. These methods provide a fast and computationally simple alternative to leave-one-out cross-validation. Then, to choose the number of terms under the alternative, I use a modified version of the approach proposed in \citet{guay_guerre_2006}. I pick the value $\hat{r}_n$ that maximizes
\[
\xi_{HC}(r_n) - r_n - \gamma_n \sqrt{2(r_n - r_{n,min})},
\]
where $\gamma_n = c \sqrt{2 \ln{\Card(r_n)}}$, $c$ is a constant that satisfies $c \geq 1 + \epsilon$ for some $\epsilon > 0$, $\Card(r_n)$ is the cardinality of the set of possible numbers of restrictions, and $r_{n,min}$ is the lowest possible number of restrictions across different choices of $r_n$. The notation $\xi_{HC}(r_n)$ emphasizes the dependence of the test statistic on the number $r_n$ of elements in $\tilde{Z}$. Intuitively, $r_n$ is the center term of $\xi_{HC}(r_n)$, while $\gamma_n \sqrt{2(r_n - r_{n,min})}$ is the penalty term that rewards simpler alternatives. In my analysis, I set $c=3$.

The resulting test is based on the test statistic $\xi_{HC}(\hat{r}_n)$ and asymptotic $\chi^2(r_{n,min})$ critical values. Table~\ref{tbl_simulated_size_data_driven} presents the results. We can see that the data-driven test controls size well and has good power. While the theoretical development of data-driven methods for tuning parameter choice is beyond the scope of this paper, it appears that the simple procedure presented here works well in practice.

\renewcommand\thesection{Appendix S.\Alph{section}}

\renewcommand\thesubsection{S.\Alph{section}.\arabic{subsection}}

\setcounter{section}{0}

\renewcommand{\thetheorem}{S.\arabic{theorem}}

\renewcommand{\thelemma}{S.\arabic{lemma}}

\renewcommand{\theassumption}{S.\arabic{assumption}}

\renewcommand{\theremark}{S.\arabic{remark}}

\renewcommand{\theequation}{S.\arabic{equation}}

\renewcommand{\thetable}{S\arabic{table}}

\renewcommand{\thefigure}{S\arabic{figure}}

\setcounter{theorem}{0}

\setcounter{lemma}{0}

\setcounter{assumption}{0}

\setcounter{remark}{0}

\setcounter{table}{0}

\setcounter{figure}{0}

\clearpage

\section{Tables and Figures}\label{supplement_appendix_tables_figures}

\onehalfspacing

\begin{figure}[H]
\begin{center}
\caption{Simulated Size and Power of the Feasible Tests}\label{fig_simulated_size_power_comparison}

$n=250$

\includegraphics[scale=0.33]{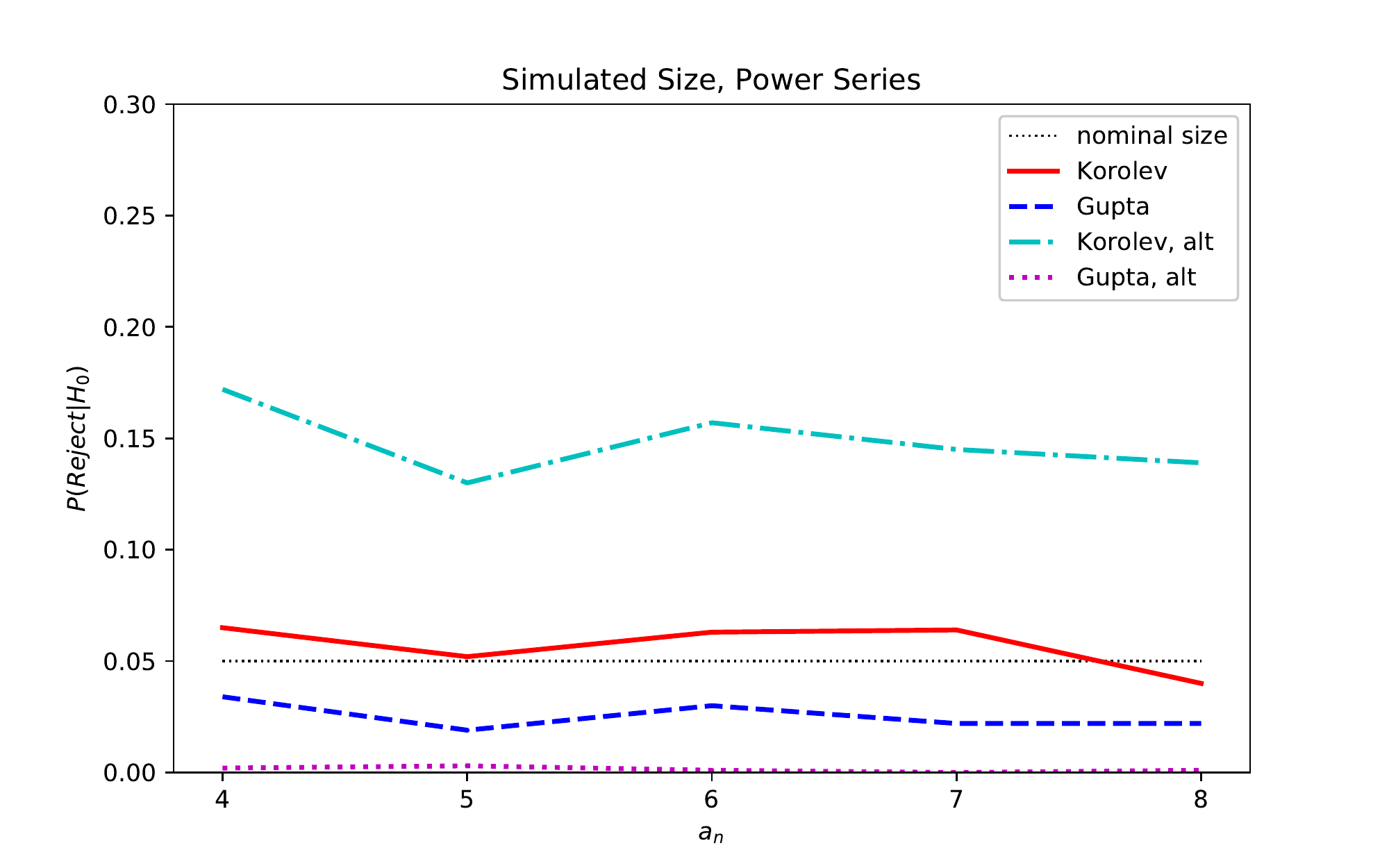} \includegraphics[scale=0.33]{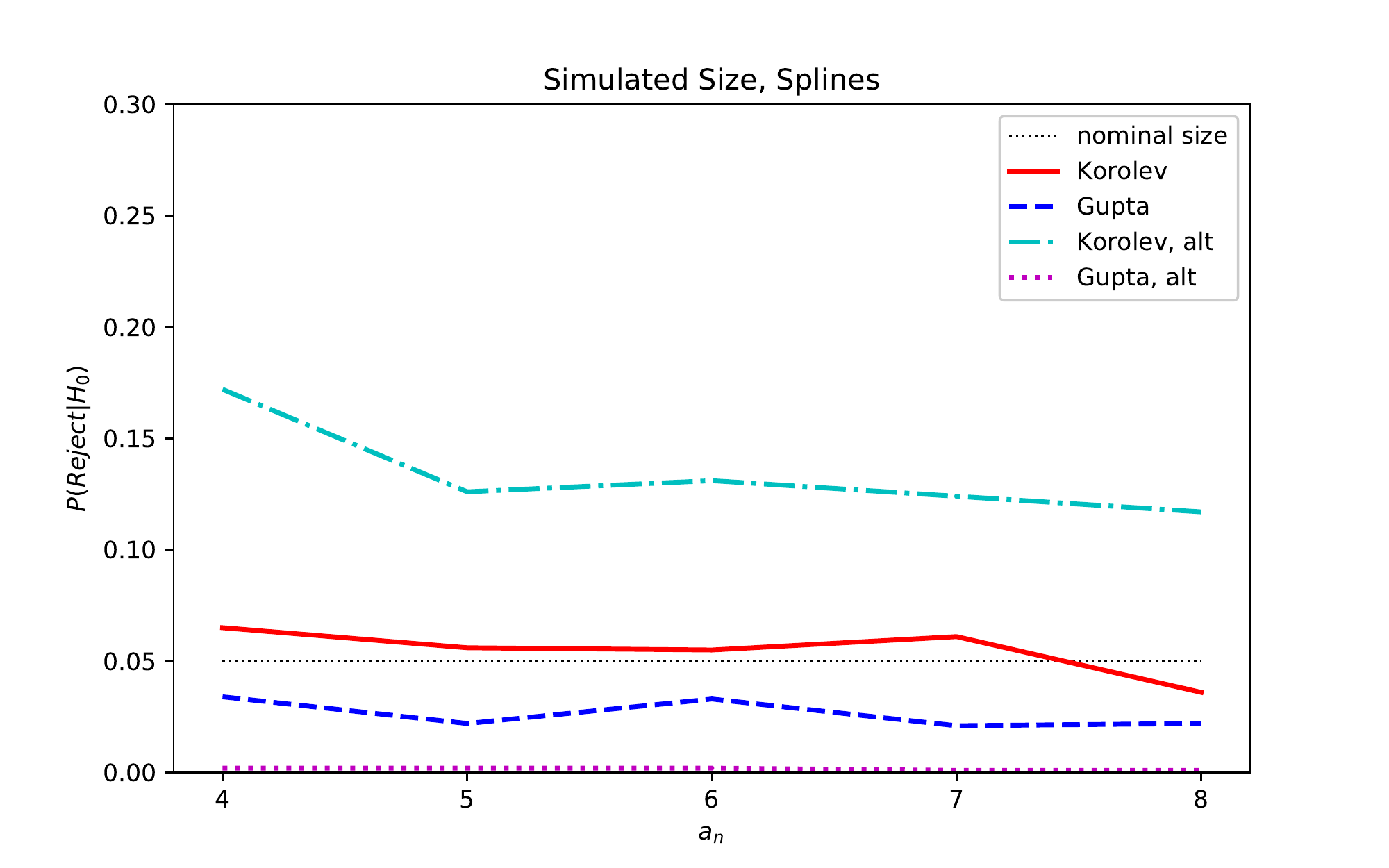}

\includegraphics[scale=0.33]{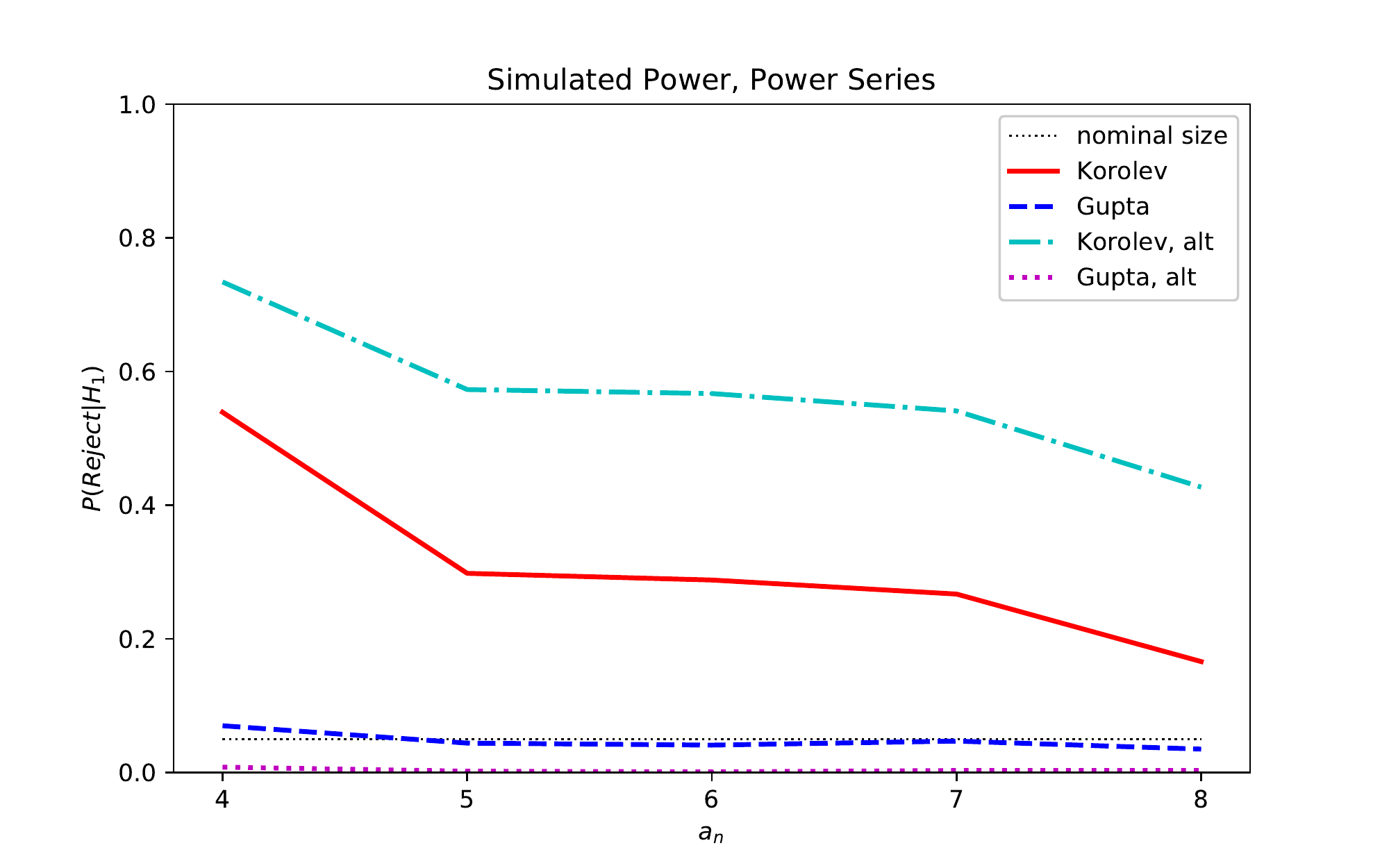} \includegraphics[scale=0.33]{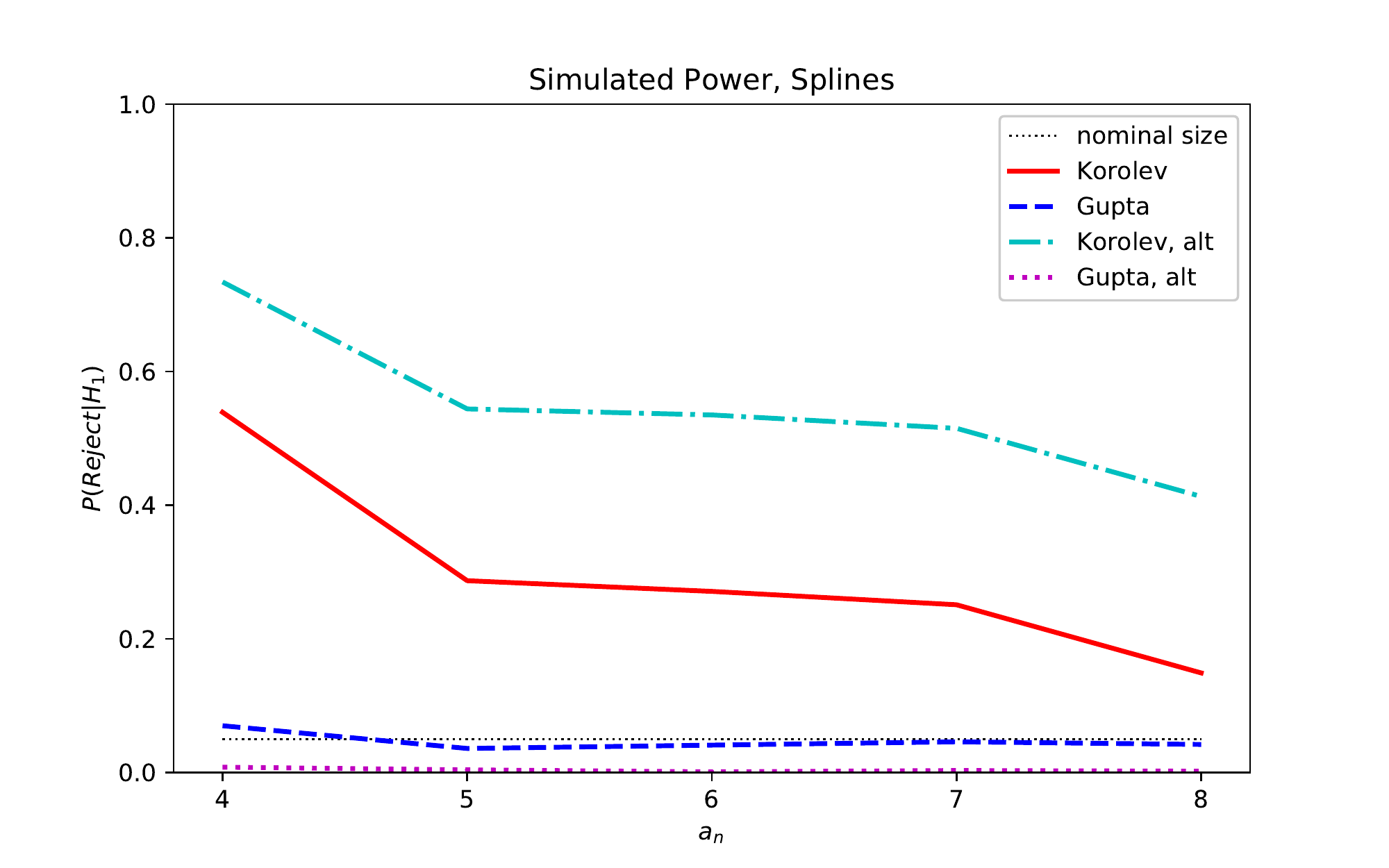}

$n=1,000$

\includegraphics[scale=0.33]{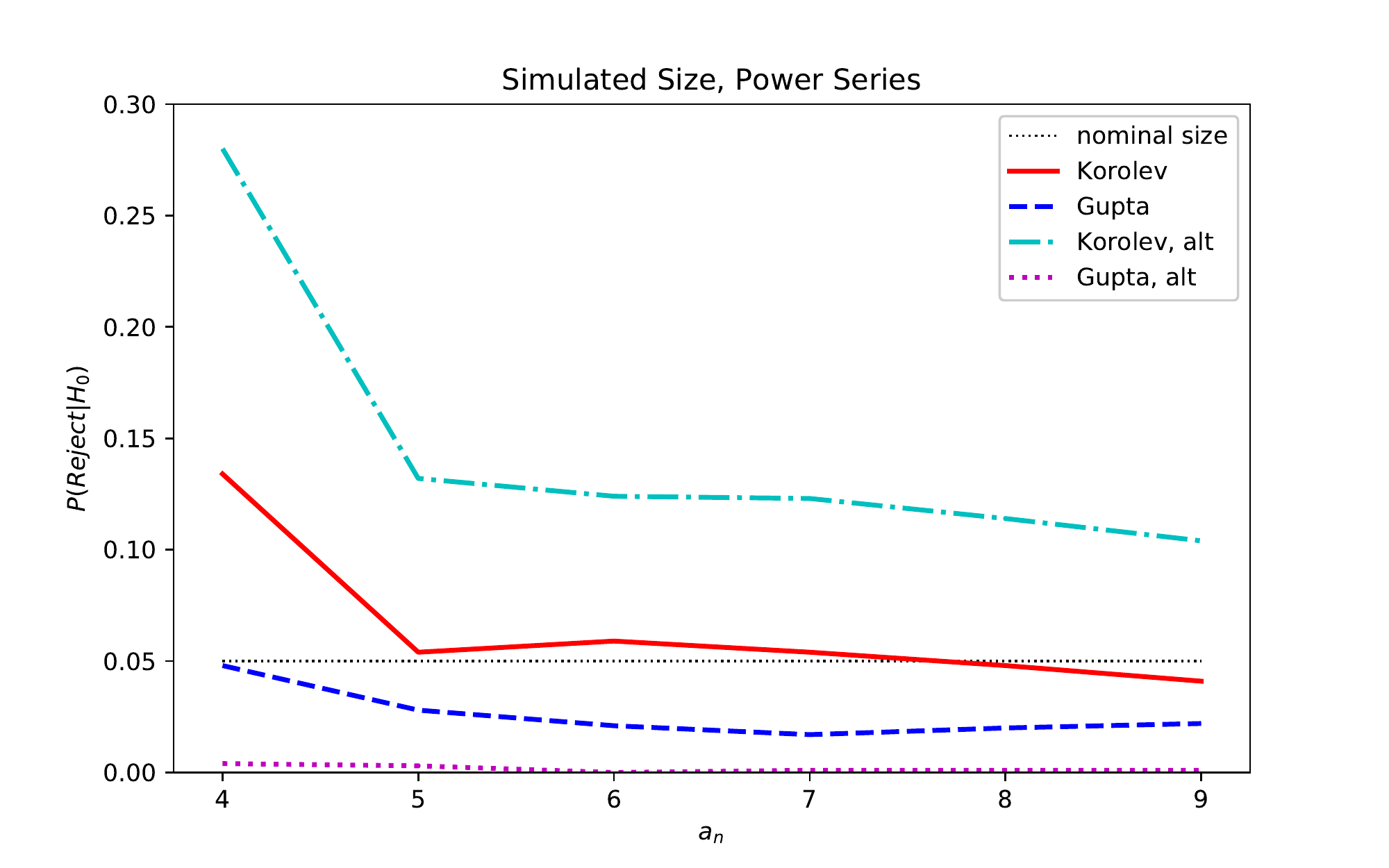} \includegraphics[scale=0.33]{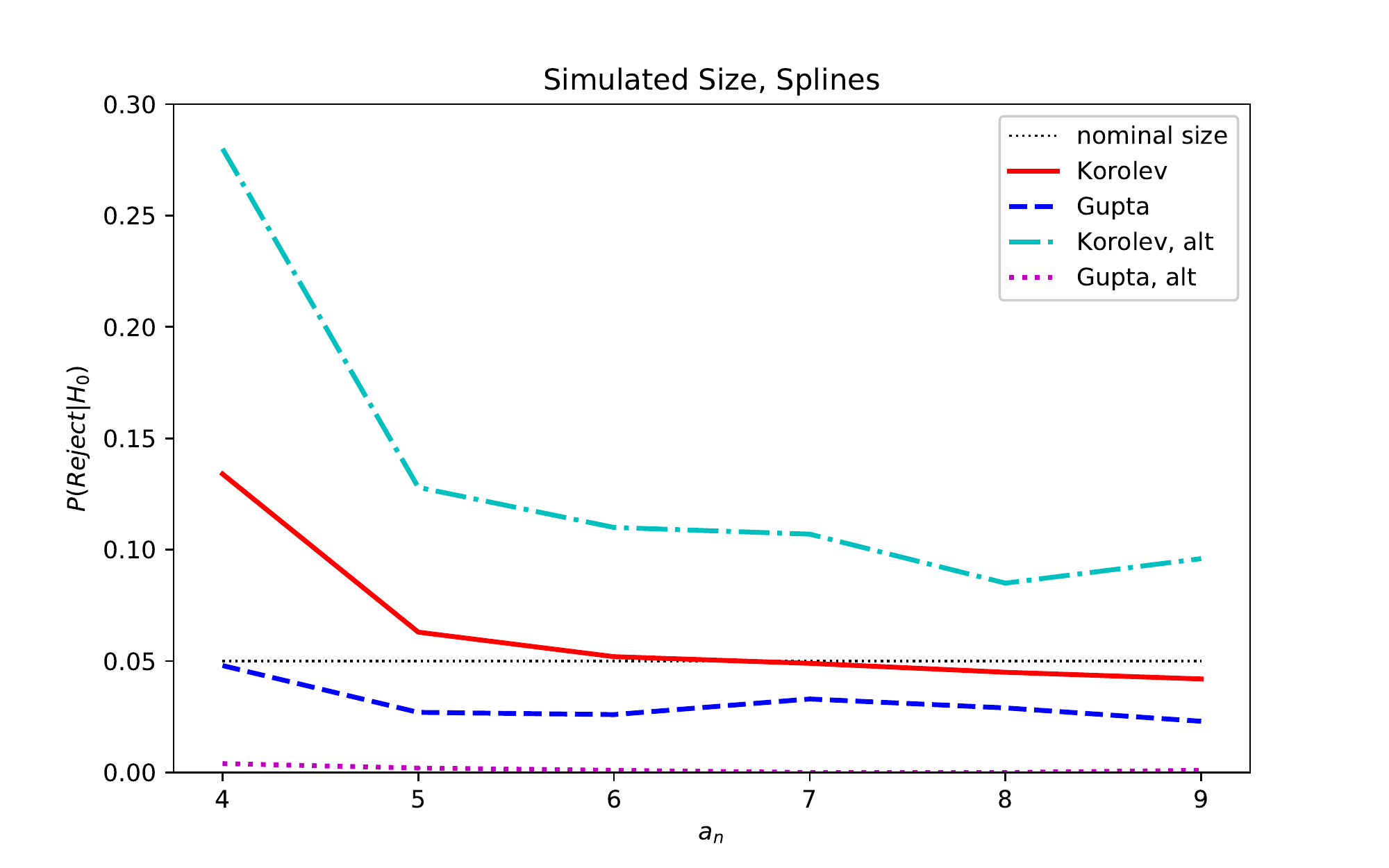}

\includegraphics[scale=0.33]{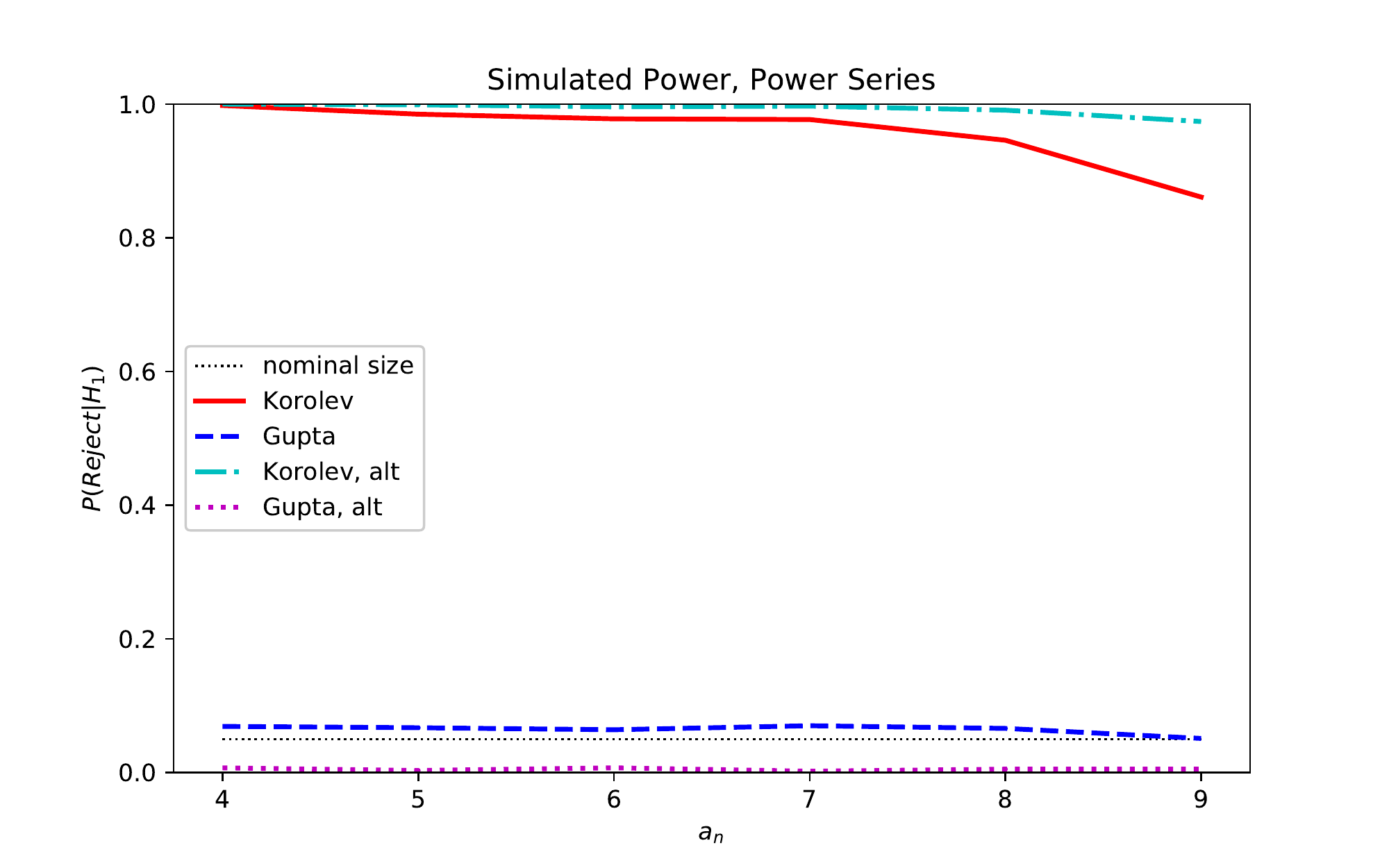} \includegraphics[scale=0.33]{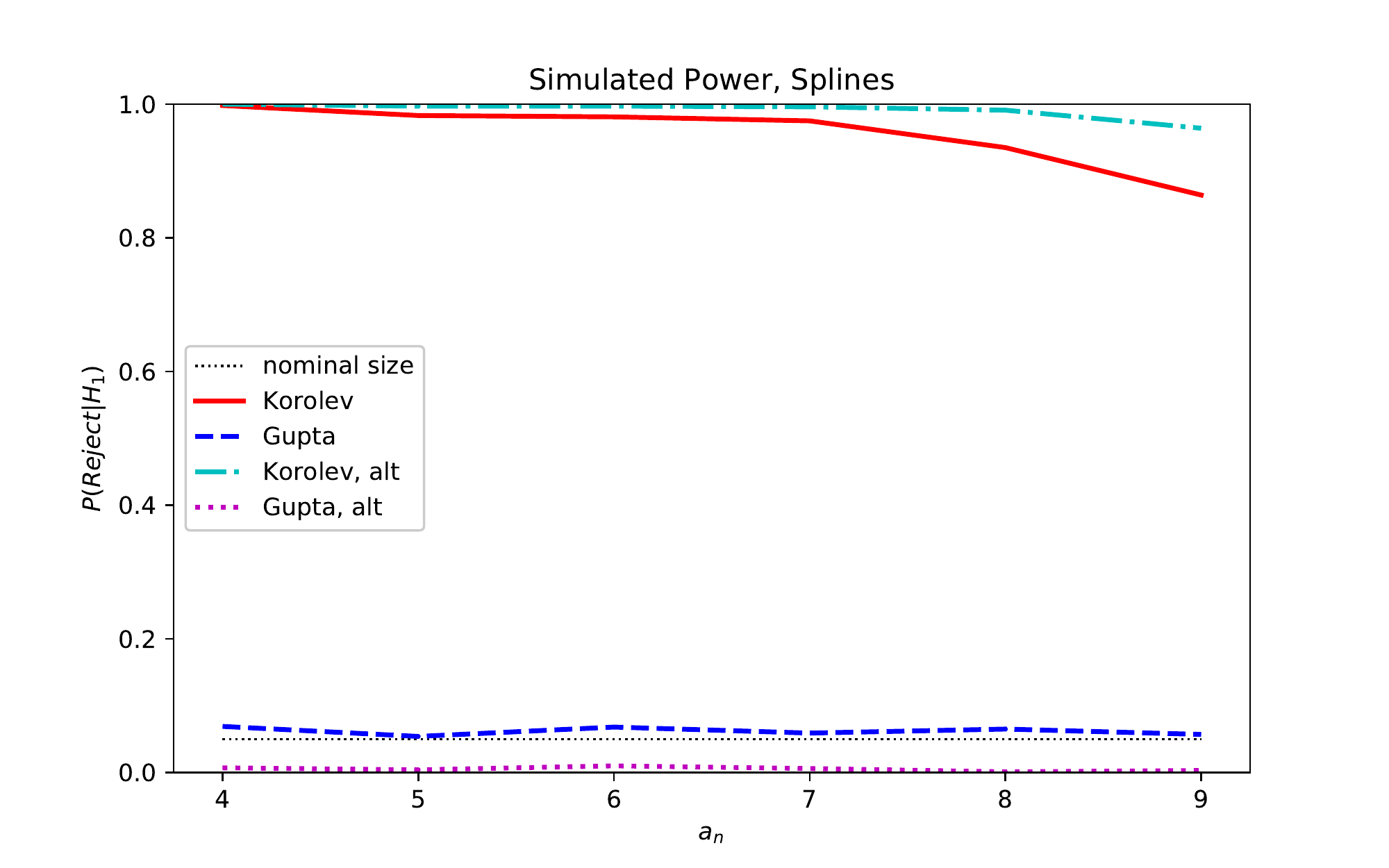}

\footnotesize{\vspace{-0.4cm}\singlespacing Left: power series. Right: splines. The results are based on $M=1,000$ simulation draws.}
\end{center}
\end{figure}

\begin{figure}[H]
\begin{center}
\caption{Simulated Size and Power of the Infeasible Tests}\label{fig_simulated_size_power_comparison_infeasible}

$n=250$

\includegraphics[scale=0.33]{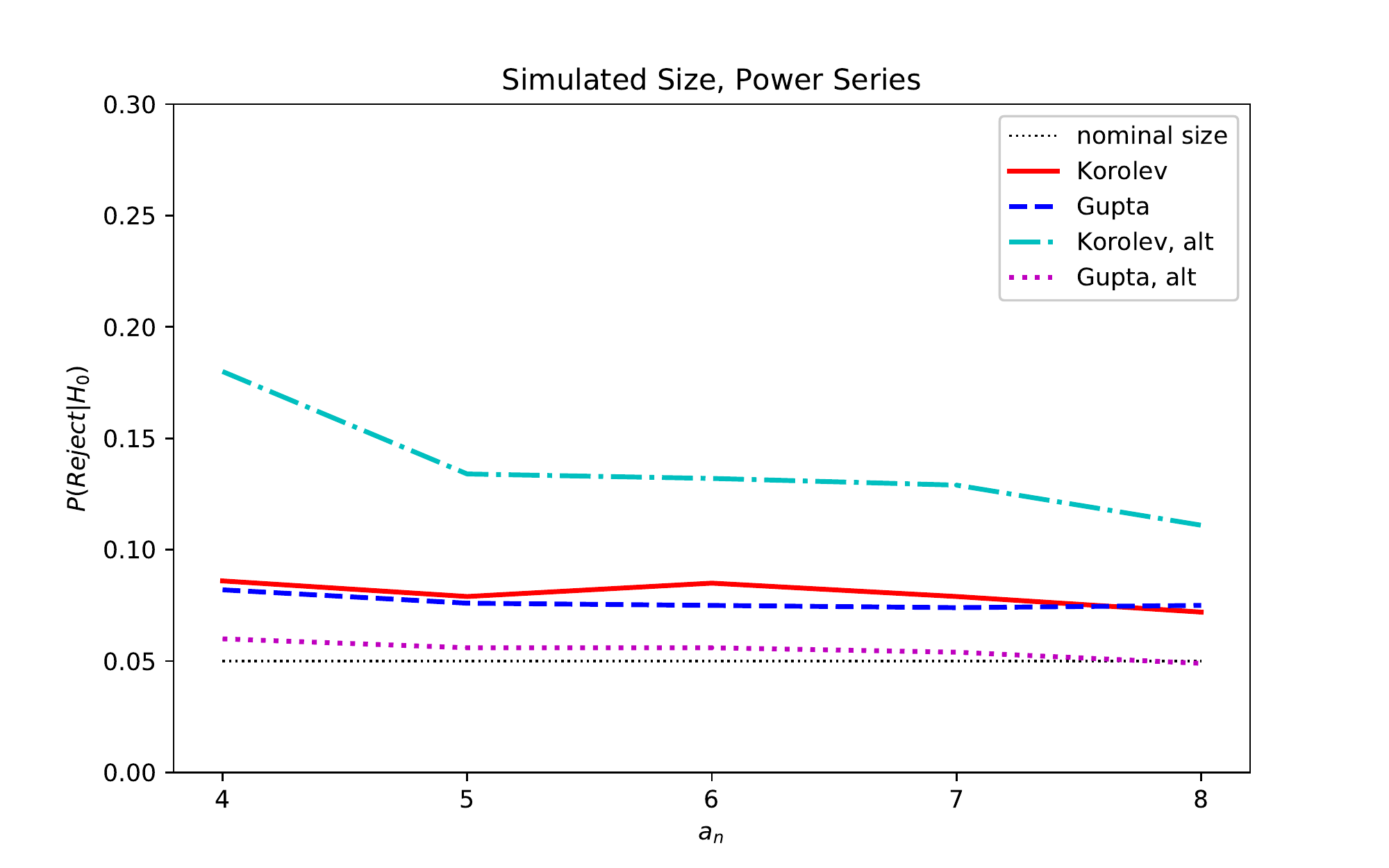} \includegraphics[scale=0.33]{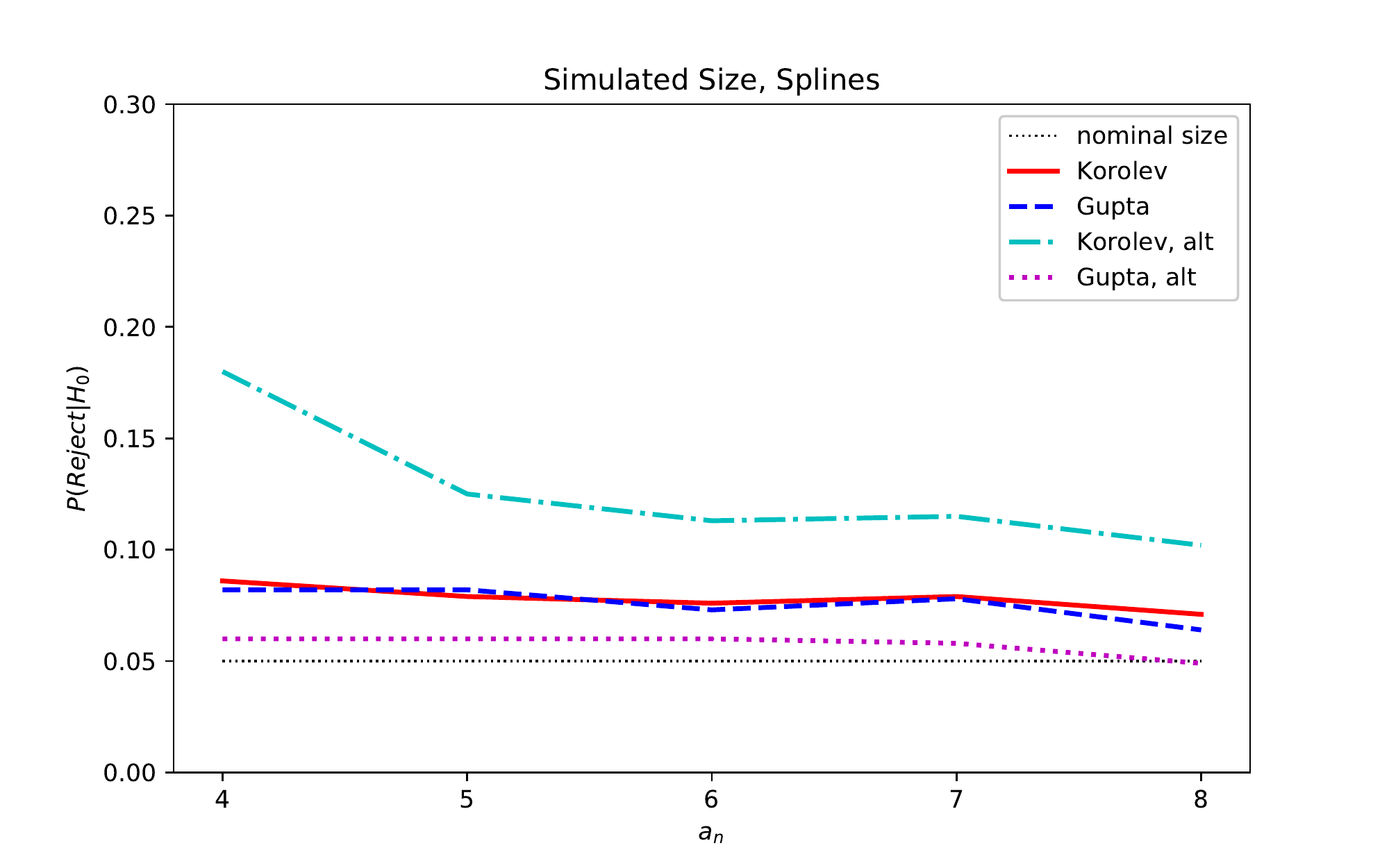}

\includegraphics[scale=0.33]{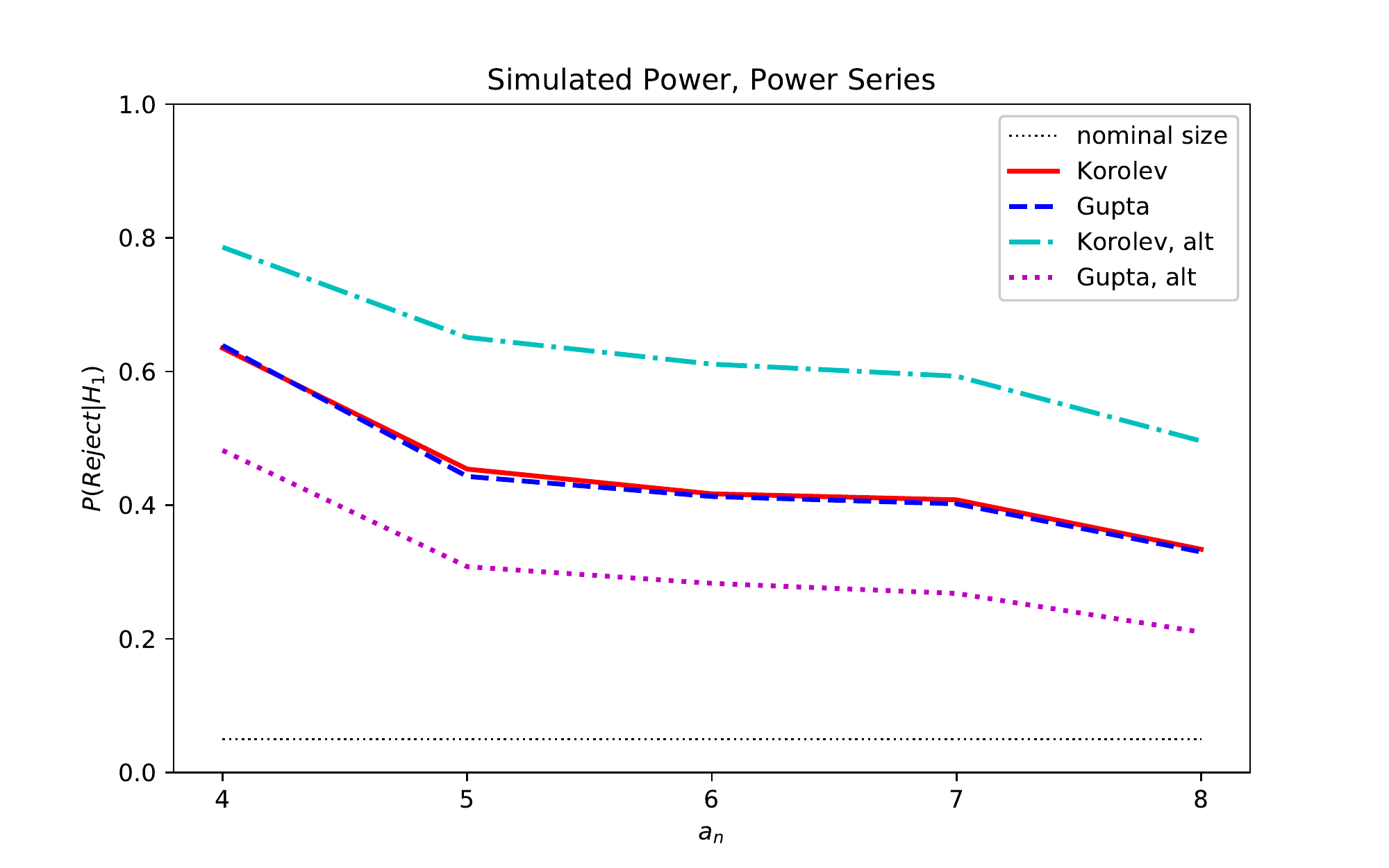} \includegraphics[scale=0.33]{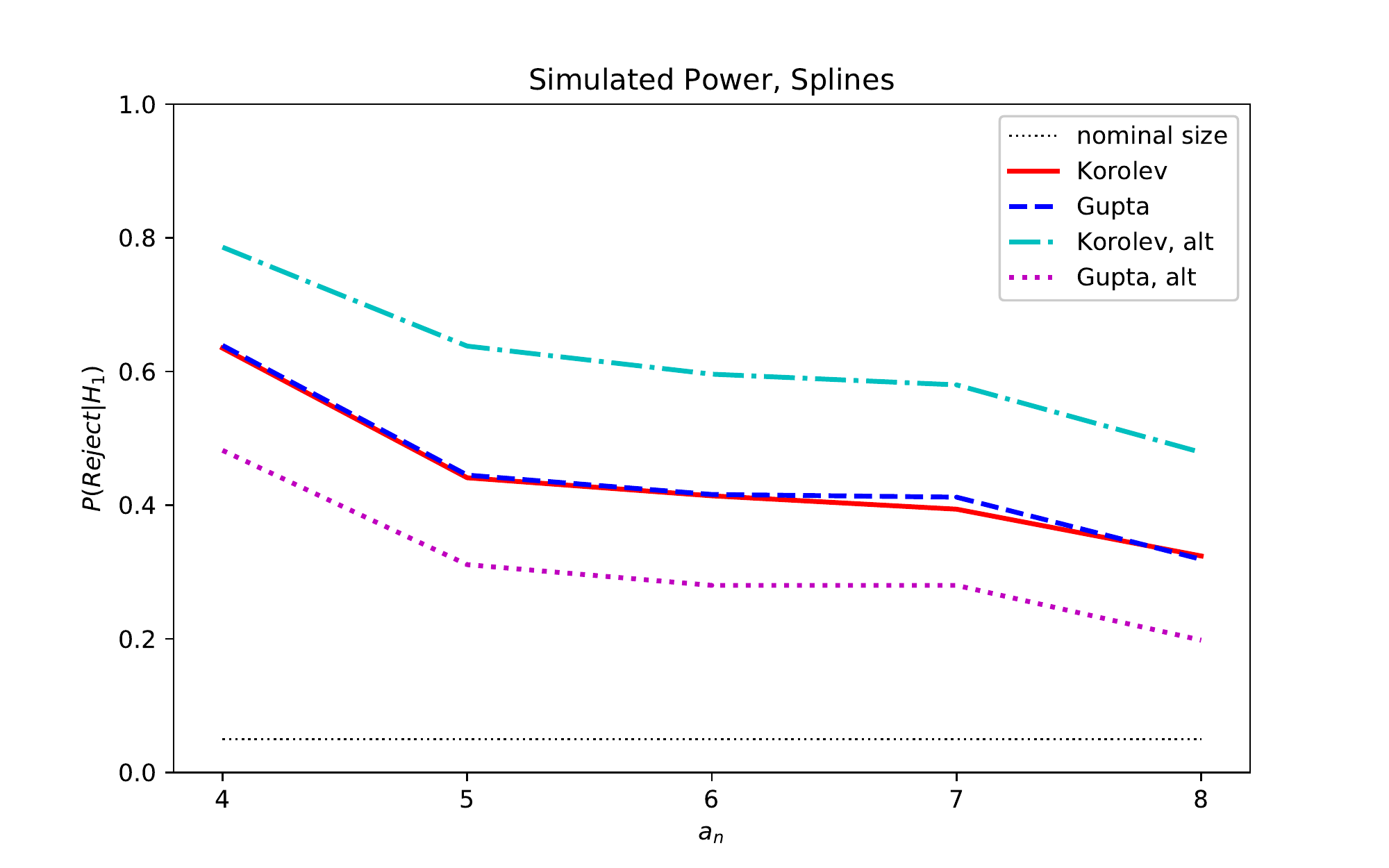}

$n=1,000$

\includegraphics[scale=0.33]{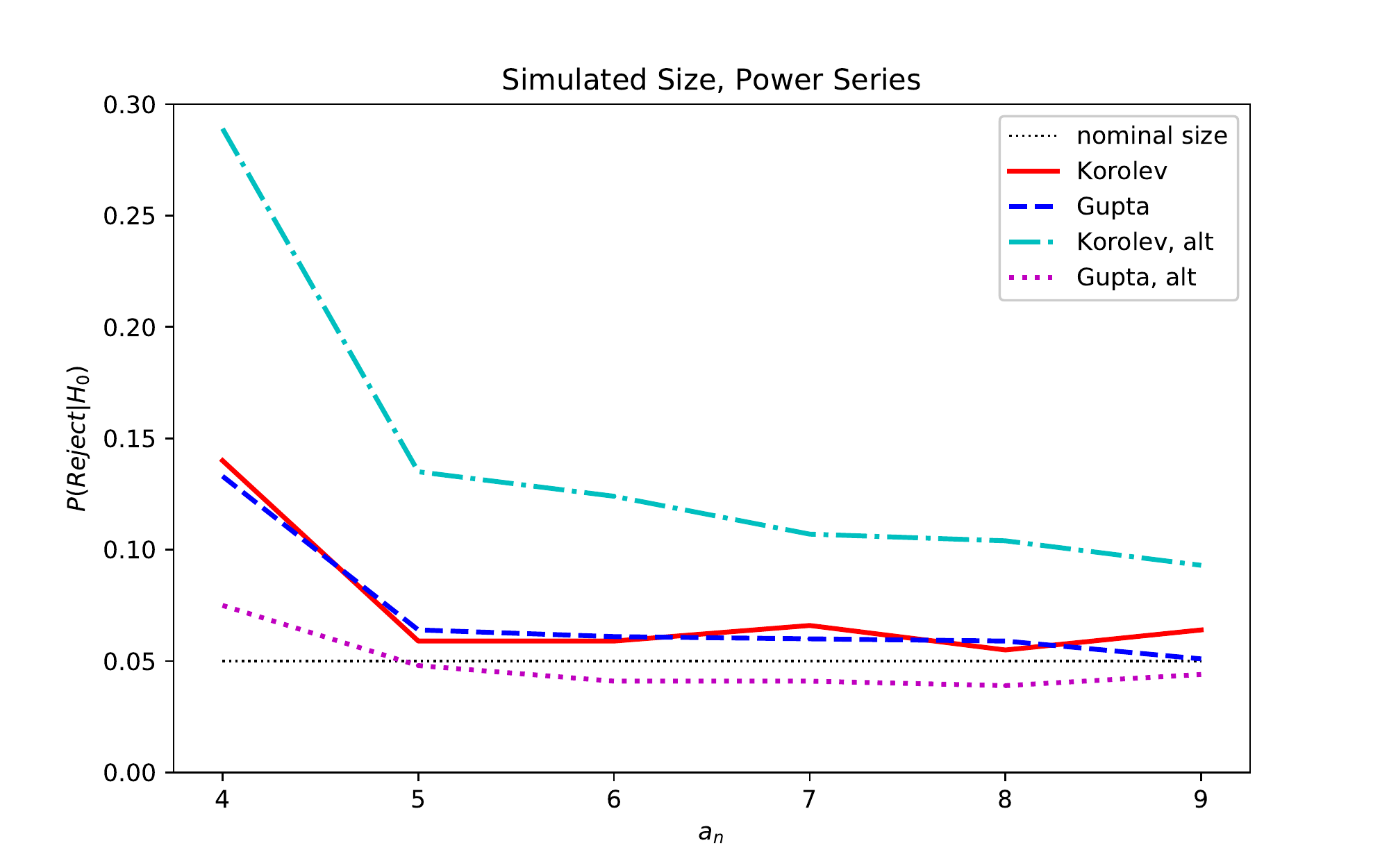} \includegraphics[scale=0.33]{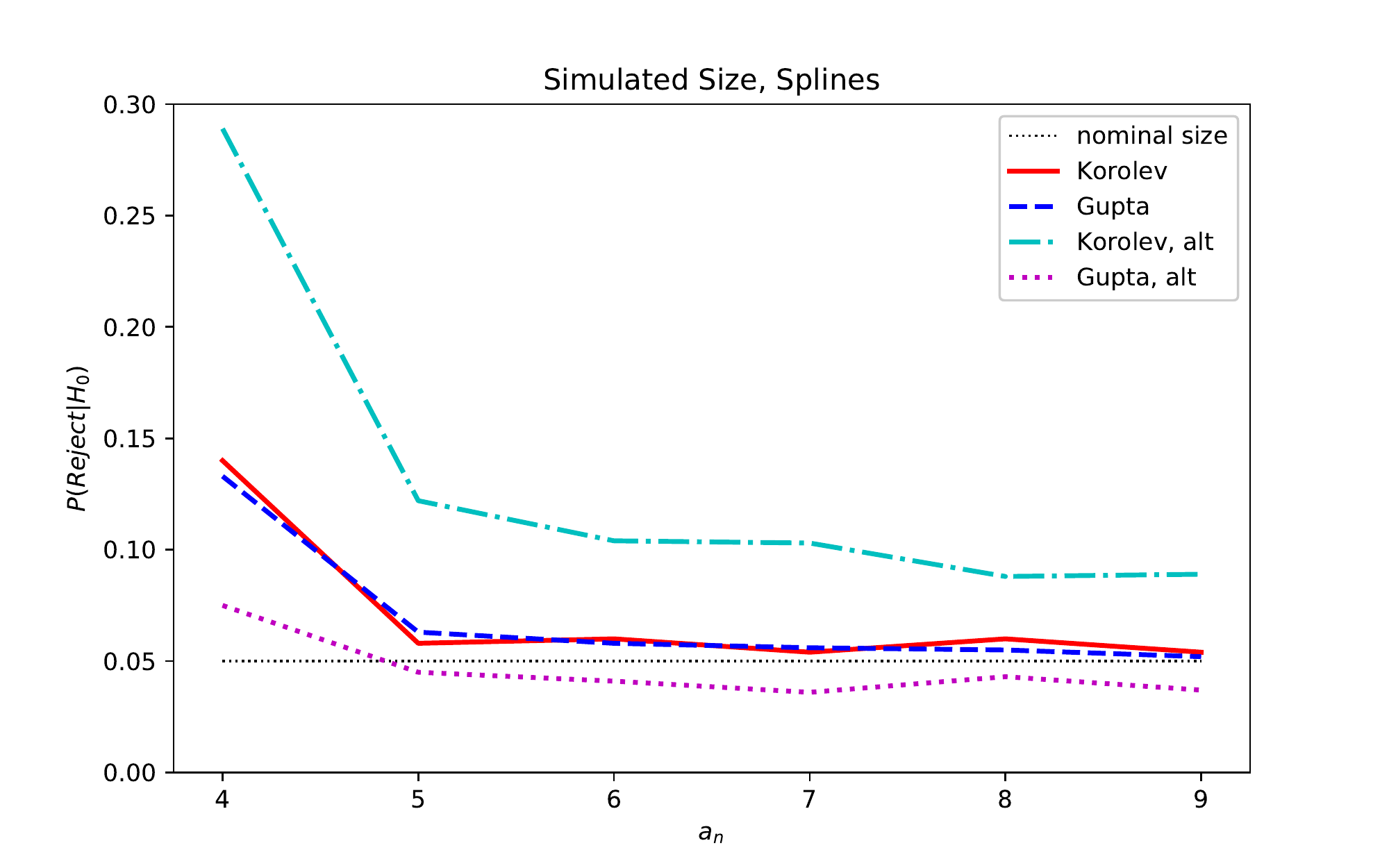}

\includegraphics[scale=0.33]{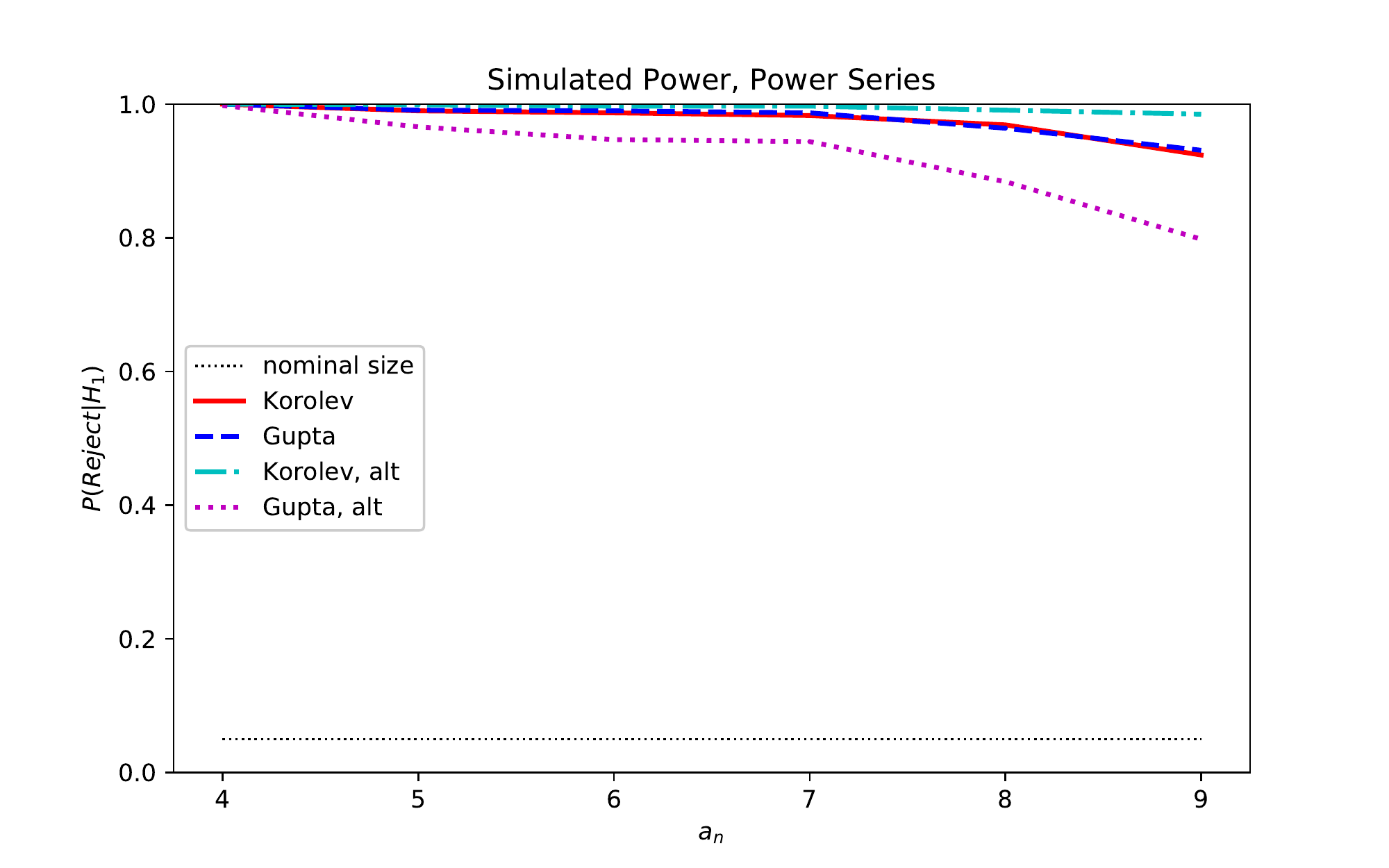} \includegraphics[scale=0.33]{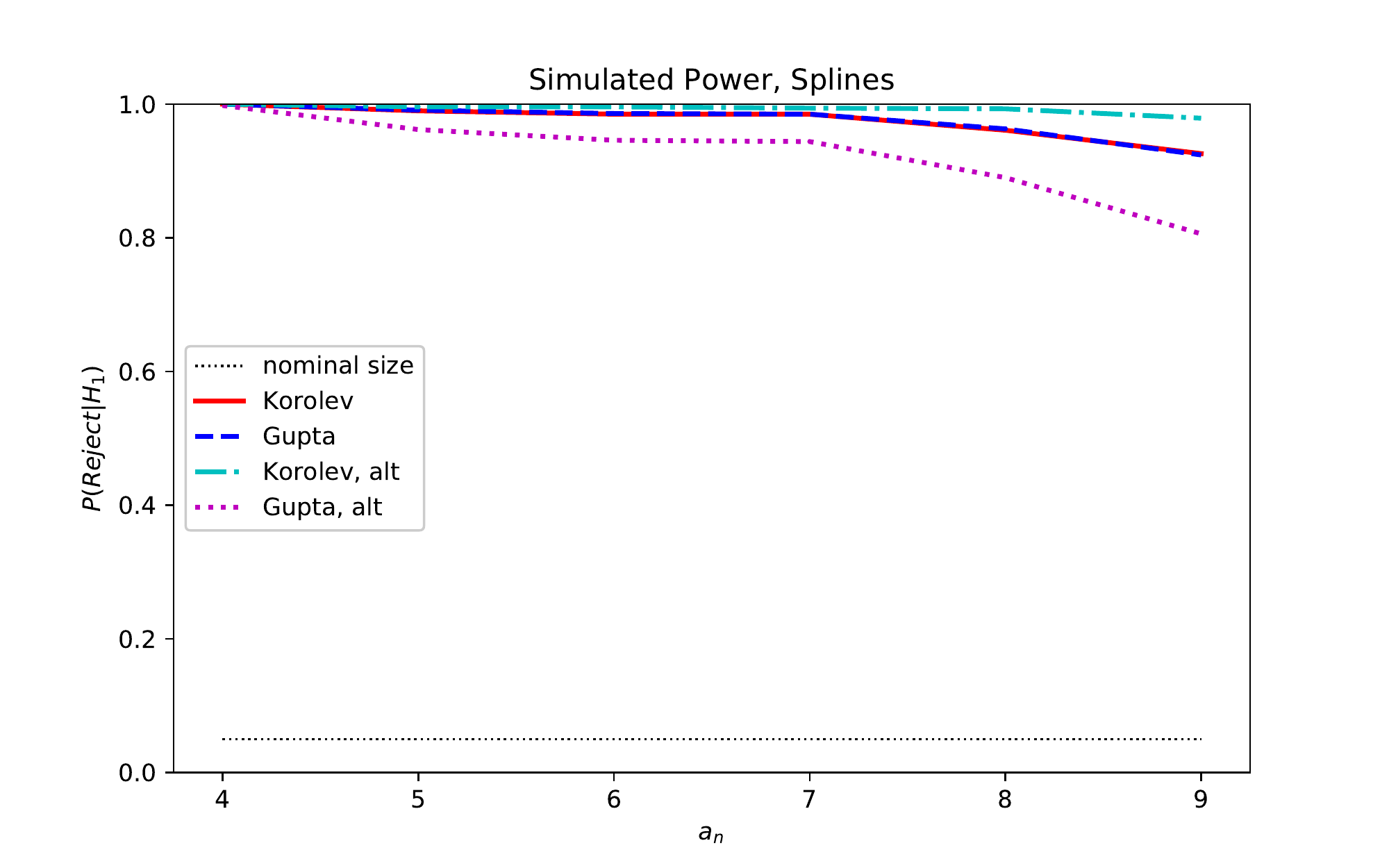}

\footnotesize{\vspace{-0.4cm}\singlespacing Left: power series. Right: splines. The results are based on $M=1,000$ simulation draws.}
\end{center}
\end{figure}

\begin{table}[h]
\begin{center}
\caption{Number of Series Terms}\label{tbl_number_terms}
\begin{tabular}{c | c c c c}
$a_n$ & $a_n^2$ & $m_n$ & $k_n$ & $r_n$ \\ \hline\hline
4  & 16 & 5  & 16  & 11 \\
5  & 25  & 6  & 25  & 19 \\
6  & 36  & 7  & 27  & 20 \\
7  & 49  & 8  & 29  & 21 \\
8  & 64  &  9  & 40  & 31 \\
9  & 81  & 10  & 53  & 43 \\
\end{tabular}
\end{center}
\footnotesize{\singlespace The table shows how the number of parameters under the null $m_n$, the total number of parameters $k_n$, and the number of restrictions $r_n$ change as a function of the number of parameters in univariate series expansion $a_n$. As a baseline, I also include $a_n^2$, the total number of parameters if all interaction terms were included.}
\end{table}

\begin{table}[h]
\begin{center}
\caption{Simulated Size and Power of the Data-Driven Test}\label{tbl_simulated_size_data_driven}
\begin{tabular}{l | c c | c c}
& $C_p$ & GCV & $C_p$ & GCV \\ \hline\hline
& \multicolumn{4}{c}{n = 250} \\
& \multicolumn{2}{c}{Power Series} & \multicolumn{2}{c}{Splines}\\
Size & 0.037 & 0.035 & 0.039 & 0.039 \\
Power & 0.390 & 0.393 & 0.391 & 0.394 \\ \hline
& \multicolumn{4}{c}{n = 1,000} \\
& \multicolumn{2}{c}{Power Series} & \multicolumn{2}{c}{Splines}\\
Size & 0.047 & 0.047 & 0.048 & 0.048 \\
Power & 0.991 & 0.991 & 0.991 & 0.991 \\ \hline
\end{tabular}
\footnotesize{\singlespace The results are based on $M=1,000$ simulations.}
\end{center}
\end{table}

\end{document}